\documentclass{article}

\usepackage{amsthm}
\usepackage{amsmath}
\usepackage{authblk}

\usepackage{mathtools}

\usepackage{tikz}
\usetikzlibrary{calc}
\usepackage{bm}
\usepackage{wrapfig}
\usepackage{enumitem}

\usepackage{multirow}
\usetikzlibrary{patterns.meta}
\usepackage{float}

\usepackage{xurl}

\usepackage{tabularx}
\usepackage{booktabs}
\usepackage{tabularray}
\usepackage{pdflscape}
\usepackage{algorithm}
\usepackage{algpseudocode}
\usepackage{caption}
\usepackage{subcaption}

\newcommand\numberthis{\addtocounter{equation}{1}\tag{\theequation}}

\usepackage{graphicx}
\usepackage[american]{babel}
\usepackage{amssymb}

\usepackage{epsfig}

\usepackage{geometry}

\usepackage[round]{natbib}
\newtheorem{theorem}{Theorem}
\newtheorem{proposition}[theorem]{Proposition}
 
\newtheorem{remark}[theorem]{Remark}
\newtheorem{corollary}[theorem]{Corollary}
\newtheorem{lemma}[theorem]{Lemma}
\theoremstyle{definition}
\newtheorem{definition}[theorem]{Definition}
\newtheorem{example}{Example}
\numberwithin{equation}{section}
\numberwithin{theorem}{section}
\numberwithin{example}{section}
\DeclareCaptionFormat{algor}{%
	\hrulefill\par\offinterlineskip\vskip1pt%
	\textbf{#1#2}#3\offinterlineskip\hrulefill}
\DeclareCaptionStyle{algori}{singlelinecheck=off,format=algor,labelsep=space}
\definecolor{c1}{RGB}{74,143,222}
\definecolor{c2}{RGB}{195,135,44}
\definecolor{c3}{RGB}{198,198,198}

\title{Identifying Direct Causal Effects in Latent Factor Models by Accounting for Unidentified Parents}
\author[1,2]{Tom Hochsprung\thanks{tom.hochsprung@dlr.de}}
\author[3]{Nils Sturma\thanks{nils.sturma@epfl.ch}}
\author[4]{Jakob Runge\thanks{jakob.runge@uni-potsdam.de}}
\author[5,6]{Mathias Drton\thanks{mathias.drton@tum.de}}
\author[1]{Andreas Gerhardus\thanks{andreas.gerhardus@dlr.de}}
\affil[1]{German Aerospace Center (DLR), Institute of Data Science, Jena, Germany}
\affil[2]{Technische Universität Berlin, Berlin, Germany}
\affil[3]{Institute of Mathematics, École Polytechnique Fédérale de Lausanne (EPFL), Lausanne, Switzerland}
\affil[4]{Institute of Computer Science, University of Potsdam, Potsdam, Germany}
\affil[5]{Department of Mathematics, TUM School of Computation, Information and Technology, Technical University of Munich, Munich, Germany}
\affil[6]{Munich Center for Machine Learning, Munich, Germany}
\date{\vspace{-5ex}}
\begin{document}
\maketitle	
\begin{abstract}
We consider linear structural equation models 
with explicitly modelled latent variables.
 In such models, observed and latent variables solve linear equations including stochastic noise terms.
The goal of our work is to identify the direct causal effects between the observed variables of interest by providing (rational) formulas in the observed covariances.
Most prior identification approaches 
operate in the latent projection framework, where latent variables are projected away into dependent error terms. However, when the observed variables are densely confounded, even if only by a few latent variables,
the projection-based approaches are unable to certify identifiability of most effects. For such problems, approaches that explicitly use the latent variables are more effective, but algorithms that were recently proposed for this purpose often remain inconclusive for denser causal graphs.
We develop a new identification criterion that is able to better handle dense graphs by leveraging the key insight that recursive identification schemes can be generalized by explicitly accounting for causal parents with (yet) unidentified direct effects. 
Combinatorial search problems in our new criterion can be tackled with the help of network-flow computations, leading to a practical useful algorithmic tool that we also make available in software.\\

\noindent
\textbf{Keywords:}  causal inference, causal
effect identification, latent variables, factor analysis,
structural equation model
\end{abstract}

\section{Introduction}
\label{sec_intro}

We consider linear structural equation models (LSEMs) with explicitly modelled latent variables. 
In such LSEMS, structural equations state that the components of the observed random vector $X$ are noisy linear functions of a subset of the other components of $X$ and a vector of latent variables $L$.  Solving the equations induces the joint distribution the model postulates for $X$.
For LSEMs with latent variables, it is important to clarify which of the coefficients in the structural equations are identifiable (that is, uniquely recoverable) from the covariance matrix of $X$ and also to derive (rational) identification formulas.
This parameter identification problem is relevant for a broad range of applied sciences because one can naturally interpret the coefficients as direct causal effects between different variables (e.g., \citealp{spirtes2000causation, pearl2009causality, runge2023causal}).

The setting in our paper is similar to the one in \citet{barber2022half}. 
Consider a collection $X=(X_v)_{v\in V}$ of $d=|V|$ observed random variables and a further collection $L=(L_h)_{h\in \mathcal{L}}$ of $\ell=|\mathcal{L}|$ latent (unobserved) random variables.
We assume that $V$ and $\mathcal{L}$ are both finite 
and that all variables $X_v$ are related by linear equations as
\begin{align}
	\label{eq_model}
	X_v=\sum_{w\neq v}\lambda_{wv}X_w+\sum_{h\in \mathcal{L}}\gamma_{hv}L_h+\epsilon_v,\quad v\in V.
\end{align}
The $\epsilon_v$ constitute noise in the form of jointly independent random variables with zero mean and variances $\omega_{vv}:=\mathbb{E}[\epsilon_v^2]\in(0,\infty)$. The $\lambda_{wv}$ and $\gamma_{hv}$ are real-valued parameters commonly referred to as the \emph{direct causal effects} of $X_w$ on $X_v$ and $L_h$ on $X_v$, respectively. We assume the latent variables $(L_h)_{h\in \mathcal{L}}$ to be directly sampled from some underlying distribution with strictly positive and finite variance and that the latent variables $(L_h)_{h\in \mathcal{L}}$ are jointly independent
 and independent of the noise terms $\epsilon=(\epsilon_v)_{v\in V}$. Furthermore, without loss of generality we  assume that the means of the latent variables $(L_h)_{h\in \mathcal{L}}$ equal $0$.

Understanding $X$, $L$ and $\epsilon$ as vectors, we can recast \eqref{eq_model} as the matrix equation
\begin{align}
\label{eq_matrix_eq}
	X=\Lambda^TX + \Gamma^T L +\epsilon,
\end{align}
where $\Lambda = (\lambda_{wv})\in\mathbb{R}^{d\times d}$ and $\Gamma=(\gamma_{hv})\in\mathbb{R}^{\ell\times d}$ are parameter matrices. The matrix $\Lambda$ has zero diagonal entries and, to ensure solvability of \eqref{eq_matrix_eq}, it is assumed that $I_d-\Lambda$ is invertible, where $I_d$ is the $d\times d$-identity matrix.  Solving for $X$ in \eqref{eq_matrix_eq}, we have 
\begin{align}
\label{eq_dep_error_terms}
	X=(I_d-\Lambda)^{-T}(\Gamma^T L + \epsilon).
\end{align}
The vector $\Gamma^TL+\epsilon$ follows a latent-factor model and has covariance matrix 
\begin{align*}
	\Omega := \textnormal{Var}(\epsilon+\Gamma^TL)=\Omega_{\textnormal{diag}}+\Gamma^T\mathcal{V}_L\Gamma,
\end{align*}
where $\Omega_{\textnormal{diag}}:= \textnormal{Var}(\epsilon)$ is diagonal, and we write  $\mathcal{V}_L:=\textnormal{Var}(L)$.
From \eqref{eq_dep_error_terms}, we see that $X$ has covariance matrix
\begin{align}
\label{eq_Sigma}
  \Sigma = (I_d-\Lambda)^{-T}\Omega (I_d-\Lambda)^{-1}.
\end{align} 
The goal of this paper is to derive conditions
under which $\Lambda$ (or some of its entries) can be identified from
the observed covariance matrix $\Sigma$.

We study the identification problem from a graphical perspective and encode the parame\-trized equation in \eqref{eq_matrix_eq} by a graph whose nodes represent the latent and observed variables in $L$ and $X$.
If the model allows an effect $\lambda_{vw}$ or $\gamma_{hw}$, for $v,w\in V$ and $h\in \mathcal{L}$, to be non-zero, then the graph contains a directed edge $v\rightarrow w$ or $h\rightarrow w$, respectively. We call this graph a \emph{latent-factor graph} and emphasize that it \emph{explicitly} shows the latent variables. 

		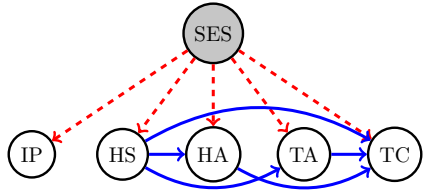
\begin{figure}
		\center
\scalebox{0.8}{
\begin{tikzpicture}
	\node[circle, draw, fill = c3, line width=0.4mm] (h1) at (0.0,0) {SES};
	\node[circle, draw, line width=0.4mm] (1) at (-3, -2) {IP};
	\node[circle, draw, line width=0.4mm] (2) at (-1.5, -2) {HS};
	\node[circle, draw, line width=0.4mm] (3) at (0, -2) {HA};
	\node[circle, draw, line width=0.4mm] (4) at (1.5, -2) {TA};
	\node[circle, draw, line width=0.4mm] (5) at (3, -2) {TC};
	
	\path [->, line width = 0.5mm, dashed, color = red] (h1) edge node[left] {} (1);
	\path [->, line width = 0.5mm, dashed, color = red] (h1) edge node[left] {} (2);
	\path [->, line width = 0.5mm, dashed, color = red] (h1) edge node[left] {} (3);
	\path [->, line width = 0.5mm, dashed, color = red] (h1) edge node[left] {} (4);
	\path [->, line width = 0.5mm, dashed, color = red] (h1) edge node[left] {} (5);

	\path [->, line width = 0.5mm, color = blue] (2) edge node[left] {} (3);
	\path [->, line width = 0.5mm, color = blue, bend right] (2) edge node[left] {} (4);
	\path [->, line width = 0.5mm, color = blue, bend left] (2) edge node[left] {} (5);
	\path [->, line width = 0.5mm, color = blue, bend right] (3) edge node[left] {} (5);
	\path [->, line width = 0.5mm, color = blue] (4) edge node[left] {} (5);
\end{tikzpicture}
}
\caption{Latent-factor graph corresponding to an example about the total energy consumption (TC) of a household; see Example \ref{example_application}.}
		\label{fig_application}
	\end{figure}

\begin{example}
\label{example_application} We consider an example about the \emph{total energy consumption (TC)} of a household, adapted 
from \citet{estiri2014building} and \citet{karatasou2019socio}. Three variables, namely, \emph{total heated floor area (TA)}, \emph{household appliances (HA)} and \emph{household size (HS)}, are assumed to have a direct causal effect on total energy consumption. Additionally, the household size is assumed to have a direct effect on total heated floor area and household appliances. All of these four observed variables are also caused by the unobserved variable \emph{social economic status (SES)}. Furthermore, one measures a proxy of SES that is conditionally independent of all the other observed variables given SES; for example, this \emph{independent observed proxy} (IP) could be the car value. The corresponding LSEM with $L=(L_{\textnormal{SES}})$ and $X=(X_{\textnormal{IP}},X_{\textnormal{HS}},X_{\textnormal{HA}},X_{\textnormal{TA}},X_{\textnormal{TC}})^T$ is given by
\begin{align*}
    X=\underbrace{\begin{pmatrix}
    0 & 0 & 0 & 0 &0 \\
    0 & 0 & 0 &0 &0\\
    0 & \lambda_{\textnormal{HS}, \textnormal{HA}} & 0 & 0 & 0\\
    0 & \lambda_{\textnormal{HS}, \textnormal{TA}} & 0 & 0 & 0\\
    0 & \lambda_{\textnormal{HS}, \textnormal{TC}} & \lambda_{\textnormal{HA}, \textnormal{TC}} & \lambda_{\textnormal{TA}, \textnormal{TC}} & 0
\end{pmatrix}}_{=\Lambda^T}
X + \underbrace{\begin{pmatrix}
    \gamma_{\textnormal{SES}, \textnormal{IP}}\\
    \gamma_{\textnormal{SES},\textnormal{HS}}\\
    \gamma_{\textnormal{SES},\textnormal{HA}}\\
    \gamma_{\textnormal{SES},\textnormal{TA}}\\
    \gamma_{\textnormal{SES},\textnormal{TC}}
    \end{pmatrix}}_{=\Gamma^T}
L + \epsilon.
\end{align*}
The induced latent-factor graph is depicted in Figure \ref{fig_application}. 
\end{example}

Explicitly modelling the latent variables is particularly effective when the latent confounding is dense but low-dimensional, that is, when only a few latent variables confound many observed variables. The literature offers several graphical methods that seek to leverage this perspective. For example, \citet{stanghellini2005identification} and \citet{leung2016identifiability} consider problems that feature only one latent variable. The conditional instrument approach of \citet{van2015efficiently} only considers cases where no confounder directly influences all observed variables, and another line of work assumes availability of observed proxies of the latent confounder \citep{kuroki2014measurement,miao2018identifying,lee2021causal,shpitser2023proximal,tchetgen2024introduction}. A criterion that allows the treatment of general latent-factor graphs is the latent-factor half-trek criterion (LF-HTC) of
\citet{barber2022half}, which applies exactly to our considered setting with latent variables that are independent and source variables. There also exist recent works such as \citet{ankan2023combining}, \citet{dong2024parameter}, or  \citet{sturma2025trek}, which generalize LF-HTC type methods to settings with causal effects among latents. Furthermore, there exist recent identification results from \citet{hochsprung2025using} for linear time series models having explicitly modelled latent confounding time series.


 Already for our latent-factor graph setting, there exist many graphs that are not identifiable via the LF-HTC, which to our knowledge is currently the most encompassing criterion, despite the fact that (rational) identifiability can be proven with computer algebra methods. In this paper, motivated by the aforementioned related results in the time series setting \citep{hochsprung2025using}, we propose a more general criterion for latent-factor graphs that is able to certify identifiability of direct effects for significantly more graphs.
 \addtocounter{example}{-1}
 \begin{example}[continued]
  Consider again the model for total energy consumption of a household, with the graph from Figure~\ref{fig_application}. It turns out that in this example all direct effects between observed variables (corresponding to the blue edges in the figure) are identifiable.  While the LF-HTC is only able to identify the effects for $\textnormal{HS}\rightarrow \textnormal{HA}$ and $\textnormal{HS}\rightarrow \textnormal{TA}$ and, thus, none of those for edges pointing into the total energy consumption, our new criterion graphically detects identifiability of all effects associated to edges between observed variables.
 \end{example}

It is worth mentioning that Gröbner basis computations allow one to perfectly decide identifiability of causal effects  (\citealp{garcia2010identifying}, Section C in \citealp{barber2022supplement}). However, this tends to be feasible only for problems of small to moderate size. While the precise computational complexity of Gr\"obner basis approaches to LSEM identifiability is unknown, applying complexity results for finding Gr\"obner bases for general polynomial ideals \citep{huynh1986superexponential,mayr1997some} to our setting yields double exponential time complexity in the size of the graph. In contrast, the LF-HTC has polynomial time complexity when bounding the number of the to-be-searched-over latent factors.  While our criterion has a larger than exponential time complexity in the size of the graph,
we also achieve polynomial time complexity in the size of the graph under practically useful algorithmic simplifications.

\begin{remark}
As pioneered by \citet{wright1921correlation, wright1934method}, there also exists the approach of using latent projections, where the latent variables are merged into error terms that then may become dependent. Such dependence is graphically represented by bidirected edges ($\leftrightarrow$). The latent projection approach allows one to apply results as in \cite{drton2011global}, which are connected to nonparametric approaches based on the do-calculus \citep{tian2002general, tian2002testable,tian2003technicalreport, shpitser2006bidentification, shpitser2006identification, huang2006identifiability, huang2006pearl, pearl2009causality, shpitser2018identification}. Moreover, there exist a host of methods taylored to linear models under latent projection; see the review of \cite{drton2018}. However, as shown by \cite{barber2022half}, these methods not only remain inconclusive under dense confounding, but can also yield incorrect conclusions on (rational/generic) identifiability under sparse but low-dimensional confounding.
\end{remark}
 
We now repeat and highlight \textbf{main contributions} of our new work:
 \begin{itemize}
     \item we propose a new broadly applicable criterion for identifying  direct causal effects in (linear) latent-factor models, which subsumes the previous state-of-the-art methods for this class of models,
     \item we provide the accompanying identification formulas, and
     \item we provide R-code  \citep{Rlanguage} at \url{https://gitlab.com/dlr-dw/elfhtc} that allows one to check identifiability, obtain  identification formulas in LaTeX-code, and compute estimates of direct causal effects from a given covariance matrix estimate.
 \end{itemize}
We structure our paper as follows: \textbf{Section \ref{sec_preliminaries}} introduces necessary preliminaries. \textbf{Section \ref{section_elfhtc}} contains the main pillar of our identification criterion, namely, a more general version of the LF-HTC from \citet{barber2022half} that allows for edgewise identification and for auxiliary sets with unidentified parents. We term this criterion eLF-HTC.
In \textbf{Section \ref{sec_further_complementary_results}}, we present two further identification results complementing the eLF-HTC.
In \textbf{Section \ref{section_computation}}, we discuss how one can computationally implement the results from Sections \ref{section_elfhtc} and \ref{sec_further_complementary_results} yielding a combined identification algorithm/criterion. In \textbf{Section \ref{sec_numerical_experiments}}, we present experiments that demonstrate that our combined algorithm from Section \ref{section_computation} identifies significantly more latent-factor graphs than the LF-HTC. Finally, in \textbf{Section \ref{sec_conclusion}}, we finish with a conclusion. Most proofs are deferred to the Appendix, which also  contains further explanations.

\section{Preliminaries}
\label{sec_preliminaries}
In this section, we introduce required preliminaries: In Section \ref{sec_basic_graphical_concepts}, we recall basic graphical concepts. The exposition hereby bears some similarity to the one in \citet{hochsprung2025using}. Then, in Section \ref{sec_generic_id}, we introduce our notion of rational identifiability. The reader might skip this section in case they are not interested in algebraic or topological arguments; for the remainder of the paper it suffices to know that rational identifiability roughly means that direct causal effects can be identified by rational functions in the entries of $\Sigma$. Finally, in Section \ref{sec_lfhtc}, we briefly state the LF-HTC from \citet{barber2022half} and introduce further required graphical concepts.

\subsection{Basic Graphical Concepts}
\label{sec_basic_graphical_concepts}

One can represent the model from equation \eqref{eq_model} using a directed graph $G=(V\cup \mathcal{L},D)$. Here, $V$ and $\mathcal{L}$ are finite disjoint sets of observed and latent nodes, respectively, and $D\subseteq (V\cup \mathcal{L})\times (V\cup \mathcal{L})$ is the set of directed edges.  We assume that the set $D$ partitions as $D=D_V\cup D_{\mathcal{L}V}$, where $D_V=D\cap (V\times V)$ contains the directed edges that are just between the observed vertices and where $D_{\mathcal{L}V}=D\cap (\mathcal{L}\times V)$ is the set of directed edges that point from a latent vertex to an observed vertex. We frequently write $G^\mathcal{L}$ instead of $G$ to indicate that $G$ is a latent-factor graph. Besides, we assume that there are no self-loops, that is, $v\rightarrow v \notin D$ for all $v\in V\cup \mathcal{L}$. However, we do not prohibit reciprocal edge pairs, that is, $v\rightarrow w \in D_V$ and $w\rightarrow v\in D_V$. 

We further adopt the following common graphical terminology (see, e.g., \citealp{lauritzen1996graphical} or \citealp{pearl2009causality}):
 For a vertex $v\in V\cup \mathcal{L}$, we define its \emph{parents} by $\textnormal{pa}(v):=\{w\in V\cup \mathcal{L}:\;w\rightarrow v\}$, and for a set of vertices $A\subseteq V\cup \mathcal{L}$ by $\textnormal{pa}(A):=\bigcup_{v\in A} \textnormal{pa}(v)$. Note that in this paper, $\textnormal{pa}(\mathcal{L})=\emptyset$. To distinguish between latent and observed  parents of some vertex $v\in V\cup\mathcal{L}$ (and analogously of a set $A\subseteq V\cup\mathcal{L}$), we write $\textnormal{pa}_V(v)=\{w\in V:\;w\rightarrow v\}$ and $\textnormal{pa}_{\mathcal{L}}(v)=\{w\in  \mathcal{L}:\;w\rightarrow v\}$, respectively. Moreover, for a vertex $v\in V\cup \mathcal{L}$, we define its \emph{children} by $\textnormal{ch}(v):=\{w\in V\cup \mathcal{L}:\;v\rightarrow w\}$,  similarly $\textnormal{ch}(A):=\bigcup_{v\in A} \textnormal{ch}(A)$ for a set of vertices $A\subseteq V\cup \mathcal{L}$. As vertices in our paper only have observed children, we do not further distinguish the children-notation into observed and latent children.

We further call a finite sequence of not necessarily distinct vertices $(v_1,v_2,\ldots,v_n)$ a \emph{walk} if every pair of successive vertices in this sequence is connected by a directed edge pointing in any direction. If additionally
 all vertices with the potential exception of the endpoint vertices $v_1$ and $v_n$ in this sequence are unique, then we also call a walk a \emph{path}. We say that a path is \emph{trivial} if it just consists of a single vertex. We further say that a path is a \emph{directed path} 
if every edge in this path is pointing in the same direction and away from the starting vertex $v_1$ towards the end-vertex $v_n$.  We call a directed path a \emph{directed cycle} if $v_1=v_n$. We emphasize that our considered setting allows for directed cycles. 
If for vertices $v,w\in V\cup\mathcal{L}$ there is a directed path with at least one edge from $v$ to $w$, then we say that $w$ is a \emph{descendant} of $v$ (and that $v$ is an \emph{ancestor} of $w$, however, we do not require this notion), so in our terminology, $v$ might be but not necessarily is a descendant of itself. We write $\textnormal{dec}(v)$ to denote the set of descendants of $v$, similarly, $\textnormal{dec}(A):=\bigcup_{v\in A}\textnormal{dec}(v)$ for a set of vertices $A\subseteq V\cup\mathcal{L}$. As vertices in our paper only have observed children and hence observed descendants, we do not further distinguish the descendants-notation into observed and latent descendants. However, we sometimes write $\textnormal{dec}_{G^\mathcal{L}}$ to indicate the graph we are referring to. Besides, we say that a non-endpoint vertex $v\in V\cup\mathcal{L}$ is a collider in a triplet of subsequent vertices $(u, v, w)$ in a walk $\pi$ if the corresponding edges to that triplet in $\pi$ are $u\rightarrow v\leftarrow w$, otherwise, we say that $v$ is a non-collider in that triplet. 

\subsection{Rational Identifiability}
\label{sec_generic_id}
In this section, we introduce our notion of rational identifiability. For better connectivity to other existing results in Section \ref{sec_determinantal_results}, our notion refers to a different underlying space than the notion of rational identifiability from the related LF-HTC paper from \citet{barber2022half}. As we show in Section \ref{sec_terminologies_implications},
both notions are equivalent in our setting. We start this section by recalling some basic required algebraic and topological terms, then continue by introducing further required notation, and then finish by introducing our notion of rational identifiability.

We first recall some basic algebraic and topological terms (c.f.\ \citealp{cox2015ideals}):
A set $A\subseteq \mathbb{R}^m:=\{(k_1,\ldots,k_m):\;k_1,\ldots,k_m\in \mathbb{K}\}$ with $m\in \mathbb{N}$ is called an \emph{affine variety} or \emph{algebraic set} if one can for some finite set of polynomials $S\subseteq \mathbb{R}[x_1,\ldots,x_m]$ write $A=V(S):=\{a\in \mathbb{R}^m:f(a)=0 \;\forall f\in S\}$. 
The \emph{Zariski topology} over $\mathbb{R}^m$ is the topology whose closed sets are exactly all the algebraic subsets of $\mathbb{R}^m$. The \emph{Zariski closure} of a set $A\subseteq \mathbb{R}^m$ is the smallest algebraic subset of $\mathbb{R}^m$ which contains $A$, and which we denote by $\overline{A}$. An algebraic set is called \emph{irreducible} if it cannot be written as the union of two proper algebraic subsets.
It holds that every proper algebraic subset of an irreducible algebraic set has Lebesgue measure zero; see, for example (for a slightly simpler statement) the lemma in \citet{okamoto1973distinctness}.

We further employ the following notation (c.f.\ \citealp{barber2022half}):
 We write $\mathbb{R}^{D_V}$ to denote the set of all real $d\times d$-matrices $\Lambda = (\lambda_{vw})$ with support $D_V$, that is, if $v\rightarrow w\notin D_V$, then $\lambda_{vw}=0$. Similarly, we write $\mathbb{R}^{D_V}_{\textnormal{reg}}$ to denote the set of all matrices $\Lambda\in \mathbb{R}^{D_V}$ for which $I_d-\Lambda$ is invertible. 
Also, we write $\mathbb{R}^{D_{\mathcal{L}V}}$ for the set of real $\ell \times d$ matrices $\Gamma = (\gamma_{hv})$ with support $D_{\mathcal{L}V}$, that is, if $h\rightarrow v\notin D_{\mathcal{L}V}$, then $\gamma_{hv}=0$. Furthermore, we let $\textnormal{diag}^+_d$ denote the set of all $d\times d$ diagonal matrices with positive diagonal indexed by the elements of $V$, and analogously, $\textnormal{diag}^+_\ell$ denotes the set of all $\ell\times \ell$ diagonal matrices with positive diagonal indexed by the elements of $\mathcal{L}$. Finally, we let $\textnormal{PD}(d)$ denote the cone of positive definite symmetric $d\times d$-matrices.
We also define the parameter space $\Theta:= \mathbb{R}^{D_V}_{\textnormal{reg}}\times \mathbb{R}^{D_{\mathcal{L}V}}\times \textnormal{diag}^+_d\times \textnormal{diag}^+_\ell$ and the map
\begin{align*}
    \chi_{G^{\mathcal{L}}}:\Theta \longrightarrow \textnormal{PD}(d),\;
    (\Lambda, \Gamma, \Omega_{\textnormal{diag}}, \mathcal{V}_L)\mapsto (I_d-\Lambda)^{-T}(\Omega_{\textnormal{diag}}+\Gamma^T\mathcal{V}_L\Gamma )(I_d-\Lambda)^{-1}.
\end{align*}

We now define our notion of rational identifiability.
\begin{definition}[Rational identifiability] \label{def:rat-id}
\mbox{}
\begin{enumerate}[label=(\alph*)]
    \item     The latent-factor graph $G^\mathcal{L}$ is rationally identifiable if there exists a proper algebraic subset $A_\Theta\subsetneq \overline{\Theta}$ and a rational map $\psi:\textnormal{PD}(d)\rightarrow \mathbb{R}^{D_V}_{\textnormal{reg}}$ such that $(\psi\circ \chi_{G^\mathcal{L}})(\Lambda,\allowbreak\Gamma,\allowbreak\Omega_{\textnormal{diag}},\allowbreak\mathcal{V}_L)=\Lambda$ for all $(\Lambda,\Gamma, \Omega_{\textnormal{diag}},\mathcal{V}_L)\in \Theta\setminus A_\Theta$.
    \item The direct causal effect $\lambda_{vw}$, or also simply the edge $v\rightarrow w\in D_V$, or if the vertex $w$ is clear from the context simply the parent $v$, is rationally identifiable if there exists a proper algebraic subset $A_\Theta\subsetneq \overline{\Theta}$ and a rational map $\psi:\textnormal{PD}(d)\rightarrow \mathbb{R}$ such that ($\psi\circ\chi_{G^\mathcal{L}})(\Lambda, \Gamma, \Omega_{\textnormal{diag}},\mathcal{V}_L)=\lambda_{vw}$ for all $(\Lambda,\Gamma, \Omega_{\textnormal{diag}},\mathcal{V}_L)\in \Theta\setminus A_\Theta$.
\end{enumerate}
\end{definition}
Because $\overline{\Theta}\cong\mathbb{R}^{m_\Theta}$, where $m_\Theta:=|D_V| + |D_{\mathcal{LV}}|+d+\ell$, it follows that $\overline{\Theta}$ is irreducible (see Section \ref{app_further_exp} in the Appendix for a more detailed explanation of this fact) and thus that every proper algebraic subset $A_\Theta$ of $\overline{\Theta}$ has Lebesgue measure zero. Therefore, our notion of rational identifiability implies that an absolutely continuous random choice of $(\Lambda,\Gamma,\Omega_{\textnormal{diag}},\mathcal{V}_L)\in \Theta$ almost surely results in a covariance matrix $\Sigma$ from which  one can recover $\Lambda$ or some causal effect $\lambda_{vw}$ by rational formulas in the entries of $\Sigma$. 
Besides, note that rational identifiability of $G^\mathcal{L}$ is equivalent to rational identifiability of all edges in $D_V$. 

We now finish this section with the following remark.
\begin{remark}
\label{remark_notions}
As we further explain in Section \ref{sec_terminologies_implications}, all results  in this paper are also true when working with the notion of rational identifiability from the related LF-HTC paper from \citet{barber2022half}, which refers to a different underlying space.
\end{remark}

\subsection{Treks, Latent-Factor Half-Treks and the Latent-Factor Half-Trek Criterion}
\label{sec_lfhtc}

In this section, we first define certain types of paths, namely \emph{treks} and \emph{latent-factor half-treks}, both of which are essential for the remainder of this work.
Treks are related to covariances via the trek rule (e.g.,  \citealp{wright1921correlation}, \citealp{wright1934method}, Section 2 in \citealp{sullivant2010trek},  Section 2 in \citealp{foygel2012half}). We then define several further trek-based concepts. Finally, we present the already existing latent-factor half-trek criterion (LF-HTC) from \citet{barber2022half} which also applies to latent-factor graphs and which to the best of our knowledge is the current most encompassing criterion for this type of graphs. In addition to the just mentioned references for the trek rule, the definitions and results in this section are also further based on Section 3 in \citet{barber2022half}.

\begin{definition}[Trek and latent-factor half-trek]
A walk $\pi$ in the latent-factor graph $G^\mathcal{L}$ is a \emph{trek from source $s$ to target $t$} if it is a walk from $s\in V\cup\mathcal{L}$ to $t\in V\cup \mathcal{L}$ and does not contain any colliders, that is, if it is of the form 
\begin{align*}
    s=v^L_l\leftarrow v^L_{l-1}\leftarrow \cdots \leftarrow v^L_1\leftarrow v^T\rightarrow v^R_1\rightarrow \cdots \rightarrow v^R_{r-1}\rightarrow v^R_r=t,
\end{align*}
where $v^T, v^L_1,\ldots,v^L_l,v^R_1,\ldots,v^R_r\in V\cup \mathcal{L}$.
The left-hand side of $\pi$ is defined by $\textnormal{Left}(\pi)=\{v^T,\allowbreak v^L_1,\allowbreak\ldots,\allowbreak v^L_l\}$ and the right-hand side of $\pi$ is defined by $\textnormal{Right}(\pi)=\{v^T,\allowbreak v^R_1,\allowbreak\ldots,\allowbreak v^R_r\}$. The top node of $\pi$ is $v^T$. It is allowed that $\textnormal{Left}(\pi)=\emptyset$ or $\textnormal{Right}(\pi)=\emptyset$. If $\pi$ just consists of a single vertex, then $\pi$ is called trivial and in this case, the left and right of $\pi$ just consist of $\pi$'s single vertex. For vertices $s, t\in V\cup \mathcal{L}$, we write $\mathcal{T}(s,t)$ to denote the set of all treks from $s$ to $t$ in $G^\mathcal{L}$.

A \emph{latent-factor half-trek from source $s$ to target $t$} is a special type of trek of the form
\begin{align*}
    s\rightarrow v_1\rightarrow \cdots \rightarrow v_n\rightarrow t
\end{align*}
    or of the form
    \begin{align*}
        s\leftarrow h \rightarrow v_1 \rightarrow \cdots \rightarrow v_n \rightarrow t
    \end{align*}
    for $v_1,\ldots, v_n\in V$ and $h\in \mathcal{L}$.
\end{definition} 
		\begin{figure}
		\center
\begin{subfigure}[t]{0.4\textwidth}
		\scalebox{0.7}{
\begin{tikzpicture}
	\node[circle, draw, fill = c3, line width=0.4mm] (h1) at (0.75,0) {$h_1$};
	\node[circle, draw, line width=0.4mm] (1) at (-3, -2) {1};
	\node[circle, draw, line width=0.4mm] (2) at (-1.5, -2) {2};
	\node[circle, draw, line width=0.4mm] (3) at (0, -2) {3};
	\node[circle, draw, line width=0.4mm] (4) at (1.5, -2) {4};
	\node[circle, draw, line width=0.4mm] (5) at (3, -2) {5};
	\node[circle, draw, line width=0.4mm] (6) at (4.5, -2) {6};
	
	\path [->, line width = 0.5mm, dashed, color = red] (h1) edge node[left] {} (1);
	\path [->, line width = 0.5mm, dashed, color = red] (h1) edge node[left] {} (2);
	\path [->, line width = 0.5mm, dashed, color = red] (h1) edge node[left] {} (3);
	\path [->, line width = 0.5mm, dashed, color = red] (h1) edge node[left] {} (4);
	\path [->, line width = 0.5mm, dashed, color = red] (h1) edge node[left] {} (5);
	\path [->, line width = 0.5mm, dashed, color = red] (h1) edge node[left] {} (6);
	
	\path [->, line width = 0.5mm, color = blue] (1) edge node[left] {} (2);
	\path [->, line width = 0.5mm, color = blue] (2) edge node[left] {} (3);
	\path [->, line width = 0.5mm, color = blue] (3) edge node[left] {} (4);
	\path [->, line width = 0.5mm, color = blue] (4) edge node[left] {} (5);
	\path [->, line width = 0.5mm, color = blue] (5) edge node[left] {} (6);
	\path [->, line width = 0.5mm, color = blue, bend right] (1) edge node[left] {} (5);
\end{tikzpicture}
}
\subcaption{}
\end{subfigure}
\begin{subfigure}[t]{0.4\textwidth}
		\scalebox{0.7}{
\begin{tikzpicture}
	\node[circle, draw, fill = c3, line width=0.4mm] (h1) at (0.75,0) {$h_1$};
	\node[circle, draw, line width=0.4mm] (1) at (-1.5, -1.5) {1};
	\node[circle, draw, line width=0.4mm] (2) at (0, -2) {2};
	\node[circle, draw, line width=0.4mm] (3) at (1.5, -2) {3};
	\node[circle, draw, line width=0.4mm] (4) at (3, -1.5) {4};

	\node[circle, draw, line width=0.4mm, color = white] (p) at (-2.3, -2.15) {placeholder};

	\path [->, line width = 0.5mm, dashed, color = red] (h1) edge node[left] {} (1);
	\path [->, line width = 0.5mm, dashed, color = red] (h1) edge node[left] {} (2);
	\path [->, line width = 0.5mm, dashed, color = red] (h1) edge node[left] {} (3);
	\path [->, line width = 0.5mm, dashed, color = red] (h1) edge node[left] {} (4);

	\path [->, line width = 0.5mm, color = blue] (1) edge node[left] {} (2);
	\path [->, line width = 0.5mm, color = orange] (2) edge node[left] {} (3);
	\path [->, line width = 0.5mm, color = blue] (4) edge node[left] {} (3);
\end{tikzpicture}
}
\subcaption{}
\end{subfigure}

\caption{Left: Example latent-factor graph which is not identifiable via the LF-HTC but via the eLF-HTC, see Example \ref{example_elfhtc}. Right: Two-proxy graph for which the edge $2\rightarrow 3$ (\color{orange} orange\color{black}) is identifiable via the eLF-HTC but not via the LF-HTC, see Example \ref{ex_two_proxies}.}
		\label{fig_0}
	\end{figure}
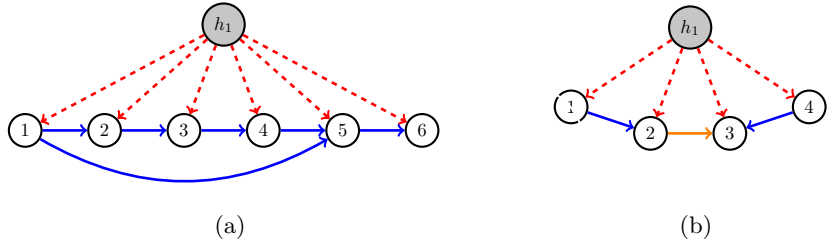
\begin{example}
Consider the latent-factor graph in Figure \ref{fig_0}a. The path $\pi_1:2\leftarrow h_1\rightarrow 3$ is a trek because it contains no colliders. The trek $\pi_1$ is even a latent-factor half-trek because it has the special required form. Similarly, the path $\pi_2:1\rightarrow 2\rightarrow 3$ is a latent-factor half-trek. The path $\pi_3:1\rightarrow 2\leftarrow h_1\rightarrow 3$, however, is not a trek (and hence, also not latent-factor half-trek) because it contains a collider, namely $1\rightarrow 2\leftarrow h_1$.
\end{example}
Several treks can be considered together, which gives rise to the following definition.
\begin{definition} [System of treks]
A \emph{system of treks $\Pi$} is a set of treks $\Pi=\{\pi_1,\ldots,\pi_n\}$.
    One says that $\Pi$ has no sided intersection if
    \begin{align*}
        \textnormal{Left}(\pi_i)\cap \textnormal{Left}(\pi_j)=\emptyset = \textnormal{Right}(\pi_i)\cap \textnormal{Right}(\pi_j) \text{ for all $i,j\in\{1,\ldots, n\}$ with $i\neq j$}.
    \end{align*}
\end{definition}
\addtocounter{example}{-1}
\begin{example}[continued]
The system of (latent-factor half-)treks $\Pi:=\{\pi_1,\pi_2\}$ has sided intersection because $3\in \textnormal{Right}(\pi_1)\cap \textnormal{Right}(\pi_2)$.
\end{example}

To state the LF-HTC, there is a further concept that we need to introduce.
\begin{definition}[Latent-factor half-trek reachability] For two distinct vertices $v,w\in V$ in a latent-factor graph $G^\mathcal{L}$ and for a set $H\subseteq \mathcal{L}$, we say that \emph{$w$ is half-trek reachable from $v$ while avoiding $H$}, denoted by $w\in \textnormal{htr}_H(v)$, if there exists a non-trivial latent-factor half-trek from $v$ to $w$ in the latent-factor graph $G^\mathcal{L}$ which does not pass through any node in $H$. For a set $U\subseteq V$, we write $w\in \textnormal{htr}_H(U)$ provided that $w\in \textnormal{htr}_H(u)$ for some $u\in U$.
\end{definition}
\addtocounter{example}{-1}
\begin{example}[continued] Again consider the latent-factor graph from Figure \ref{fig_0}a.
For $H=\{h_1\}$, we have $2\notin \textnormal{htr}_H(3)$, however, for $H=\emptyset$, we have $2\in \textnormal{htr}_H(3)$ due to the latent-factor half-trek $3\leftarrow h_1\rightarrow 2$.
\end{example}
We are now ready to state the LF-HTC and its accompanying identification result.

\begin{definition}
    For a node $v\in V$, the triple $(Y, Z, H)\in 2^{V\setminus \{v\}}\times 2^{V\setminus \{v\}}\times 2^\mathcal{L}$ satisfies the latent-factor half-trek criterion (LF-HTC) with respect to $v$ if:
    \begin{enumerate}[label=(\roman*)]
    \item $|Y|=|\textnormal{pa}_V(v)| + |H|$ and $|Z|=|H|$ with $Z\cap \{v\}=\emptyset$ and $Z\cap \textnormal{pa}_V(v)=\emptyset$,\footnote{The condition $Z\cap \{v\}=\emptyset$ does not occur in the original formulation of the LF-HTC in \citet{barber2022half}. However, when $v\in Z$, the subsequent LF-HTC identification theorem (Theorem \ref{theorem_lfhtc}) does not yield any further identification because every observed edge into $Z$ is required to be rationally identifiable.}
    \item $Y\cap (Z\cup\{v\})=\emptyset$ and $\textnormal{pa}_\mathcal{L}(Y)\cap \textnormal{pa}_\mathcal{L}(Z\cup \{v\})\subseteq H$, and
    \item there exists a system of latent-factor half-treks with no sided intersection from $Y$ to $Z\cup \textnormal{pa}_V(v)$ in $G^\mathcal{L}$, such that for each $z\in Z$, the half-trek terminating at $z$ takes the form $y\leftarrow h\rightarrow z$ for some $y\in Y$ and some $h\in H$.
    \end{enumerate}
\end{definition}

\begin{theorem}[LF-HTC identifiability]
\label{theorem_lfhtc}
Suppose the triple $(Y, Z, H)\in 2^{V\setminus \{v\}}\times 2^{V\setminus \{v\}}\times 2^\mathcal{L}$ satisfies the LF-HTC with respect to $v\in V$. If all directed edges $u\rightarrow y\in D_V$ with head $y\in Z\cup (Y\cap \textnormal{htr}_H(Z\cup \{v\}))$ are rationally identifiable, then all directed edges in $D_V$ with $v$ as head are rationally identifiable.
\end{theorem}

Theorem \ref{theorem_lfhtc} lays the foundation for a practical identification algorithm which is able to certify that a given graph is rationally identifiable and which is also able to print out the corresponding rational identification formulas.
This algorithm iteratively certifies rational identifiability of the observed edges in a given graph by applying Theorem \ref{theorem_lfhtc}. This algorithm also iteratively constructs the rational identification formulas for the remaining edge coefficients by reusing the rational identification formulas constructed in previous iterations. Because  edges from other observed vertices into the auxiliary set $Z$ must be known to be rationally identifiable and also because the corresponding rational formulas must be known, the algorithm must start with auxiliary sets $Z$ that have no observed parents.
Our later improvements in Section \ref{section_elfhtc} allow for auxiliary sets $Z$ with a non-empty observed parent set already at the beginning of such an algorithm.

\section{Edgewise Latent-Factor Half-Trek Criterion Allowing for Auxiliary Sets with Unidentified Parents}
\label{section_elfhtc}
The LF-HTC requires that all edge coefficients $\lambda_{wz}$ with $w\in V$ and $z\in Z$ are known to be rationally identifiable thus giving rise to an iterative procedure where one must start with auxiliary sets $Z$ for which $\textnormal{pa}_V(Z)=\emptyset$. For our new criterion which we call eLF-HTC and which we present in this section, we relax that requirement by allowing $z\in Z$ for which some $\lambda_{wz}$ are not already known to be rationally identifiable. 
Furthermore, in comparison to the LF-HTC, we build upon the fact that parents of the target vertex $v$ which are already known to be rationally identifiable need not to be re-identified and can be excluded from $\textnormal{pa}_V(v)$.

For a better intuitive understanding of the following definition of the eLF-HTC, the reader might for now translate the occuring set $W_v$ in the definition of the eLF-HTC with the observed parents of $v$ which have not yet been rationally identified. Similarly, each set $W_z$ might for now be translated with the parents of $z$  which have not yet been rationally identified. Afterwards, in Remark \ref{remark_W_v}, we will see that 
this intuition for $W_v$ suffices and does not reduce the number of eLF-HTC identifiable direct causal effects. However, in Remark \ref{remark_W_z} we will see that this intuition for the sets $W_z$ does not suffice as it reduces the number of eLF-HTC identifiable direct causal effects.
\begin{definition} [Edgewise latent-factor half-trek criterion allowing for auxiliary sets with unidentified parents (eLF-HTC)]
	Let $v\in V$ and let $W_v\subseteq \textnormal{pa}_V(v)$. In addition, let $Z\subseteq V\setminus \{v\}$ and $W_z\subseteq \textnormal{pa}_V(z)$ for all $z\in Z$. Write $W_Z:=\bigcup_{z\in Z}W_z$ and $Z^{(1)}:=\{z\in Z:\;W_z\subsetneq\textnormal{pa}_V(z)\}$ and $Z^{(2)}:=\{z\in Z:\;W_z=\textnormal{pa}_V(z)\}$. Furthermore, let $Y\subseteq V\setminus \{v\}$ and $H\subseteq \mathcal{L}$. 
	The tuple $(Y, Z, (W_z)_{z\in Z}, H)$ satisfies the eLF-HTC with respect to $(v, W_v)$ if 
	\begin{enumerate}[label=(\roman*)]
		\item $|Z| = |H|$ and $|Y|=|W_v\cup Z\cup W_Z|$ with $Z\cap \{v\} = \emptyset$ and $Z^{(1)}\cap (W_Z\cup W_v)=\emptyset$,
		\item $Y\cap (Z\cup \{v\}) = \emptyset$ and $\textnormal{pa}_{\mathcal{L}}(Y)\cap \textnormal{pa}_{\mathcal{L}}(Z\cup \{v\})\subseteq H$, and
		\item there exists a system of latent-factor half-treks with no sided intersection from $Y$ to $W_v\cup Z \cup W_Z$ in $G^{\mathcal{L}}$, such that for each $z\in Z$, the half-trek terminating at $z$ takes the form $y\leftarrow h \rightarrow z$ for some $y\in Y$ and some $h\in H$.
	\end{enumerate}
\end{definition}

\begin{theorem}[eLF-HTC identification result]
	\label{theorem_elfhtc}
	Suppose the tuple $(Y, Z, (W_z)_{z\in Z}, H)$ satisfies the eLF-HTC with respect to $(v, W_v)$. Also suppose that all directed edges $u\rightarrow y\in D_V$ with head $y\in (Y\cap \textnormal{htr}_H(Z\cup \{v\}))$,  all directed edges $p^z\rightarrow z\in D_V$ with $p^z \in \textnormal{pa}_V(z)\setminus W_z$ for all $z\in Z^{(1)}$, and all directed edges $p^v\rightarrow v\in D_V$ with $p^v\in \textnormal{pa}_V(v)\setminus W_v$ are rationally identifiable. Then, all directed edges $p^v\rightarrow v \in D_V$ with $p^v\in  W_v\setminus (Z^{(2)}\cup W_Z)$ are rationally identifiable.
\end{theorem}
The proof of Theorem \ref{theorem_elfhtc} also gives explicit rational formulas to identify the direct causal effect of interest, which the reader can find in Section \ref{sec_proof_elfhtc}.

We now explain an example. (We present the rational identification formulas for this example in Section \ref{sec_extended_example_elfhtc}).

\begin{example}[eLF-HTC identification example]
\label{example_elfhtc}
Consider the latent-factor graph from Figure \ref{fig_0}a with $V=\{1,\allowbreak 2,\allowbreak 3,\allowbreak 4,\allowbreak 5,\allowbreak 6\}$ and $\mathcal{L}=\{h_1\}$.
This latent-factor graph is not LF-HTC identifiable which one can quickly see by using the R-package SEMID \citep{Rlanguage, barber2025package} in which the LF-HTC is implemented, yet this latent-factor graph is eLF-HTC identifiable (Theorem \ref{theorem_elfhtc}) as we now explain:

\noindent\underline{$v=4$:} Take $W_v=\{3\}$ and $(Y=\{1, 2, 3\}, Z=\{6\}, (W_z)_{z\in Z}=(\{5\}), H=\{h_1\})$. For this choice, we have $\textnormal{pa}_{\mathcal{L}}(Y)=\{h_1\}$ and $\textnormal{htr}_H(Z\cup \{v\})=\{4, 5, 6\}$ and $\textnormal{pa}_{\mathcal{L}}(Z\cup \{v\})=\{h_1\}$. A possible system of latent-factor half-treks is $\{1\rightarrow 5 , 2\leftarrow h_1 \rightarrow 6, 3\}$. In addition, note that no other edges need to be rationally identifiable because $Y\cap (\textnormal{htr}_H(Z\cup \{v\}))=\emptyset$, because $W_v = \textnormal{pa}_V(v)=\{3\}$, and because $Z^{(1)}=\emptyset$.
The identified edge is $3\rightarrow 4$.

\noindent\underline{$v=6$:} Take $W_v=\{5\}$ and $(Y=\{1, 2, 3\}, Z=\{4\}, (W_z)_{z\in Z}=(\{3\}), H=\{h_1\})$. For this choice, we have $\textnormal{pa}_{\mathcal{L}}(Y)=\{h_1\}$ and $\textnormal{htr}_H(Z\cup \{v\})=\{4, 5, 6\}$ and $\textnormal{pa}_{\mathcal{L}}(Z\cup \{v\})=\{h_1\}$. A possible system of latent-factor half-treks is $\{1\rightarrow 5 , 2\leftarrow h_1 \rightarrow 4, 3\}$.  In addition, note that no other edges need to be rationally identifiable because $Y\cap (\textnormal{htr}_H(Z\cup \{v\}))=\emptyset$, because $W_v = \textnormal{pa}_V(v)=\{5\}$, and because $Z^{(1)}=\emptyset$.
The identified edge is $5\rightarrow 6$.

\noindent\underline{$v=2$:} Take $W_v=\{1\}$ and $(Y=\{1, 6\}, Z=\{4\}, (W_z)_{z\in Z}=(\emptyset), H=\{h_1\})$. For this choice, we have $\textnormal{pa}_{\mathcal{L}}(Y)=\{h_1\}$ and $\textnormal{htr}_H(Z\cup \{v\})=\{2, 3, 4, 5, 6\}$ and $\textnormal{pa}_{\mathcal{L}}(Z\cup \{v\})=\{h_1\}$. The system of latent-factor half-treks is $\{1, 6\leftarrow h_1 \rightarrow 4\}$. In addition, note that $Y\cap (\textnormal{htr}_H(Z\cup \{v\}))=\{6\}$, that $W_v = \textnormal{pa}_V(v)=\{1\}$, and that $Z^{(1)}=\{4\}$ and $\textnormal{pa}_V(z)\setminus W_z=\{3\}$, and thus all observed edges required to be rationally identifiable, namely $5\rightarrow 6$ and $3\rightarrow 4$, are rationally identifiable as argued previously. The identified edge is $1\rightarrow 2$.

\noindent\underline{$v=3$:} Take $W_v=\{2\}$ and $(Y=\{2, 4\}, Z=\{1\}, (W_z)_{z\in Z}=(\emptyset), H=\{h_1\})$.  For this choice, we have $\textnormal{pa}_{\mathcal{L}}(Y)=\{h_1\}$ and $\textnormal{htr}_H(Z\cup \{v\})=\{1, 2, 3, 4, 5, 6\}$ and $\textnormal{pa}_{\mathcal{L}}(Z\cup \{v\})=\{h_1\}$. The system of latent-factor half-treks is $\{2, 4\leftarrow h_1 \rightarrow 1\}$. In addition, note that $Y\cap (\textnormal{htr}_H(Z\cup \{v\}))=\{2, 4\}$, that $W_v = \textnormal{pa}_V(v)=\{2\}$, and that $Z^{(1)}=\emptyset$, and thus all observed edges required to be rationally identifiable, namely $1\rightarrow 2$ and $3\rightarrow 4$, are rationally identifiable as argued previously. The identified edge is $2\rightarrow 3$.

\noindent\underline{$v=5$:} Take $W_v=\{1, 4\}$ and $(Y=\{1, 3, 4\}, Z=\{2\}, (W_z)_{z\in Z}=(\emptyset), H=\{h_1\})$.  For this choice, we have $\textnormal{pa}_{\mathcal{L}}(Y)=\{h_1\}$ and $\textnormal{htr}_H(Z\cup \{v\})=\{2, 3, 4, 5, 6\}$ and $\textnormal{pa}_{\mathcal{L}}(Z\cup \{v\})=\{h_1\}$. The system of latent-factor half-treks is $\{1, 3\leftarrow h_1 \rightarrow 2, 4\}$. In addition, note that $Y\cap (\textnormal{htr}_H(Z\cup \{v\}))=\{1, 3, 4\}$, that $W_v = \textnormal{pa}_V(v)=\{1,4\}$, and that $Z^{(1)}=\{2\}$ and $\textnormal{pa}_V(z)\setminus W_z=\{1\}$, and thus all observed edges required to be rationally identifiable, namely $1\rightarrow 2$ and $2\rightarrow 3$ and $3\rightarrow 4$, are rationally identifiable as argued previously. The identified edges are $1\rightarrow 5$ and $4\rightarrow 5$.
\end{example}

Our eLF-HTC criterion is also able to identify the edge from treatment to response in the well-known two-proxy graph \citep{kuroki2014measurement, miao2018identifying}, which the LF-HTC is not able to identify.
\begin{example}[Two-proxy graph]
\label{ex_two_proxies}
Consider the latent-factor graph in Figure \ref{fig_0}b. The edge $2\rightarrow 3$ is not identifiable via the LF-HTC which one can again quickly see by using the R-package SEMID \citep{Rlanguage, barber2025package}, however, this edge is eLF-HTC identifiable. To see this fact, take 
$W_v=\{2,4\}$ and $(Y=\{1, 2\}, Z=\{4\}, (W_z)_{z\in Z} = \emptyset, H=\{h_1\})$. For the corresponding rational formula for identification, see Section \ref{sec_rat_two_proxies}.
\end{example}

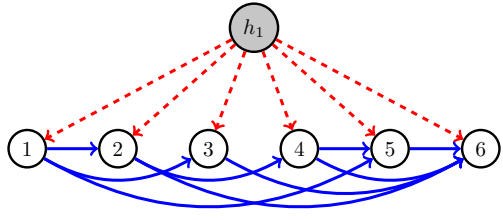
\begin{figure}
	\center
	\scalebox{0.8}{
	\begin{tikzpicture}
	\node[circle, draw, fill = c3, line width=0.4mm] (h1) at (0.75,0) {$h_1$};
	\node[circle, draw, line width=0.4mm] (1) at (-3, -2) {1};
	\node[circle, draw, line width=0.4mm] (2) at (-1.5, -2) {2};
	\node[circle, draw, line width=0.4mm] (3) at (0, -2) {3};
	\node[circle, draw, line width=0.4mm] (4) at (1.5, -2) {4};
	\node[circle, draw, line width=0.4mm] (5) at (3, -2) {5};
	\node[circle, draw, line width=0.4mm] (6) at (4.5, -2) {6};
	
	\path [->, line width = 0.5mm, dashed, color = red] (h1) edge node[left] {} (1);
	\path [->, line width = 0.5mm, dashed, color = red] (h1) edge node[left] {} (2);
	\path [->, line width = 0.5mm, dashed, color = red] (h1) edge node[left] {} (3);
	\path [->, line width = 0.5mm, dashed, color = red] (h1) edge node[left] {} (4);
	\path [->, line width = 0.5mm, dashed, color = red] (h1) edge node[left] {} (5);
	\path [->, line width = 0.5mm, dashed, color = red] (h1) edge node[left] {} (6);
	
	\path [->, line width = 0.5mm, color = blue] (1) edge node[left] {} (2);
	\path [->, line width = 0.5mm, color = blue, bend right] (1) edge node[left] {} (3);
	\path [->, line width = 0.5mm, color = blue, bend right] (3) edge node[left] {} (6);
	\path [->, line width = 0.5mm, color = blue] (4) edge node[left] {} (5);
	\path [->, line width = 0.5mm, color = blue] (5) edge node[left] {} (6);
	\path [->, line width = 0.5mm, color = blue, bend right] (1) edge node[left] {} (5);
	\path [->, line width = 0.5mm, color = blue, bend right] (2) edge node[left] {} (4);
	\path [->, line width = 0.5mm, color = blue, bend right] (4) edge node[left] {} (6);
	\path [->, line width = 0.5mm, color = blue, bend right] (2) edge node[left] {} (6);
\end{tikzpicture}
}
\caption{Example graph for Remark \ref{remark_W_z} in Section \ref{section_elfhtc}.}
	\label{fig_1.5}
\end{figure}

We now present a proposition stating that the LF-HTC implies the eLF-HTC.
\begin{proposition}[Identifiability via LF-HTC$\implies$ Identifiability via eLF-HTC]
\label{lf_htc_implications}
	Whenever an edge is identifiable via the LF-HTC, it is identifiable via the eLF-HTC. 
\end{proposition}
\begin{proof}
    Choose $W_v=\textnormal{pa}_V(v)$ and $W_z=\emptyset$ for all $z\in Z$. If $\textnormal{pa}_V(z)=\emptyset$, then $z\in Z^{(2)}$, otherwise $z\in Z^{(1)}$. Because the condition $Z^{(1)}\cap (W_Z\cup W_v)=\emptyset$ in the eLF-HTC becomes $Z^{(1)}\cap \textnormal{pa}_V(v) = \emptyset$ for these choices of $W_v$ and $W_z$, the conditions in the LF-HTC imply the conditions from the eLF-HTC.
\end{proof}
We now come back to the intuition behind $W_v$ and the $W_z$'s.
\begin{remark}[Choosing $W_v$]
\label{remark_W_v}
 Theorem \ref{theorem_elfhtc} in principle allows $W_v$ to also contain nodes corresponding to edges $p^v \rightarrow v$ which are already certified to be rationally identifiable. However, in  practice, it is enough to choose the set $W_v \subseteq \text{pa}_V(v)$ to be equal to the subset of observed parents of $v$ such that each edge $p^v \rightarrow v$ for $p^v \in W_v$ is not yet certified to be rationally identifiable.

To see this, we will now argue that if  $(Y,\allowbreak Z,\allowbreak (W_z)_{z\in Z},\allowbreak H)$ satisfies the eLF-HTC with respect to $(v, W_v)$, then, for all $W_v'\subseteq W_v$, there exists $Y'\subseteq Y$ such that $(Y',\allowbreak Z,\allowbreak (W_z)_{z\in Z},\allowbreak H)$ satisfies the eLF-HTC with respect to $(v, W'_v)$. For this, let $Y'$ contain all the elements of $Y$ that correspond to the latent-factor half-treks ending in $W_v'\cup Z\cup W_Z$. Then, all conditions of the eLF-HTC are still satisfied. Condition (iii) holds by design, and for Condition (i) note that $|Y'|=|W'   _v\cup Z\cup W_Z|$ and  $Z^{(1)}\cap (W_Z\cup W_v')=\emptyset$ because $W_v'\subseteq W_v$. Finally, Condition (ii) is satisfied because $Y'\subseteq Y$. Furthermore, if the further assumptions in Theorem \ref{theorem_elfhtc} are satisfied for $(Y,\allowbreak Z,\allowbreak (W_z)_{z\in Z},\allowbreak H)$ and $W_v$, then the further assumptions in Theorem \ref{theorem_elfhtc} are satisfied for $(Y',\allowbreak Z,\allowbreak (W_z)_{z\in Z},\allowbreak H)$ and $W'_v$ if all edges in $\textnormal{pa}_V(v)\setminus W'_v$ are rationally identifiable. This fact follows because $(Y'\cap \textnormal{htr}_H(Z\cup \{v\}))\subseteq (Y\cap \textnormal{htr}_H(Z\cup \{v\}))$. We conclude that if $W_v$ only contains the observed parents $p^v$ of $v$ such that each edge $p^v \rightarrow v$ has not been previously shown to be rationally identifiable, then the number of eLF-HTC identifiable direct causal effects is not reduced.
\end{remark}
\begin{remark}[Choosing $W_z$]
\label{remark_W_z}
In contrast, it is \emph{not} enough to choose the sets $W_z \subseteq \text{pa}_V(z)$ to be equal to the subset of observed parents of $z$ such that each edge $p^z \rightarrow v$ for $p^z \in W_z$ is not yet certified to be rationally identifiable. In this case,  satisfaction of the eLF-HTC for some choice $(Y,\allowbreak Z,\allowbreak (W_z)_{z\in Z},\allowbreak H)$ with respect to $(v, W_v)$ does not necessarily imply satisfaction of the eLF-HTC for $(Y, Z,\allowbreak (W'_z)_{z\in Z},\allowbreak H)$ with respect to $(v, W_v)$ for which each $W'_z\subseteq W_z$.

We will now construct an example to illustrate this. Consider the latent-factor graph from Figure \ref{fig_1.5} and suppose that all edges in $D_V$ are known to be rationally identifiable except for the edges from $\textnormal{pa}_V(6)=\{2,3,4,5\}$ to $v=6$. Then, the eLF-HTC is satisfied for $(Y=\{1,\allowbreak 2,\allowbreak 3,\allowbreak 5\}, Z=\{4\}, (W_z)_{z\in Z}=(\{2\}), H=\{h_1\})$ and $W_v=\{2,\allowbreak 3,\allowbreak 4,\allowbreak 5\}$. To see this fact, note that $Z^{(1)}=\emptyset$ and that a possible system of latent-factor half-treks is $\{1\leftarrow h \rightarrow 4,\allowbreak 2,\allowbreak 3,\allowbreak 5\}$. Furthermore, all edges in $D_V$ with head in $Y\cap \textnormal{htr}_H(Z\cup \{v\})=\{5\}$ are rationally identifiable by assumption. Therefore, Theorem \ref{theorem_elfhtc} yields that the edges $3\rightarrow 6$ and $5\rightarrow 6$ are rationally identifiable by the eLF-HTC. However, the eLF-HTC for $(Y=\{1,\allowbreak 2,\allowbreak 3,\allowbreak 5\}, Z=\{4\}, (W_z)_{z\in Z}=(\emptyset), H=\{h_1\})$ and $W_v=\{2,\allowbreak 3,\allowbreak 4,\allowbreak 5\}$ is not satisfied because $Z^{(1)}=\{4\}$ and thus, $Z^{(1)}\cap (W_Z\cup W_v)$ is violated. Nor is the eLF-HTC satisfied for any other choice of $Y$ and the same choices of $W_v$, $Z$, $(W_z)_{z\in Z}$ and $H$ because $Z^{(1)}\cap (W_Z\cup W_v)=\emptyset$ is only possible if $Z=\{1\}$ and hence $Z^{(1)}=\emptyset$, in which case $Y$ needs to be a subset of $\{2,\allowbreak 3,\allowbreak 4,\allowbreak 5\}$ such that it fulfills $Y\cap (Z\cup \{v\})=\emptyset$. Hence, it must be that  $|Y|<|W_v\cup Z \cup W_Z|$, so Condition (i) of the eLF-HTC is violated. We conclude that if we choose the sets $W_z$ such that each $W_z$ only contains those observed parents $p^z$ of $z$ such that each edge $p^z \rightarrow z$ has not been previously shown to be rationally identifiable,  the number of eLF-HTC identifiable direct causal effects can be reduced.
\end{remark}

\section{Further Complementary Identification Results}
\label{sec_further_complementary_results}
In this section, we present two further identification approaches which well complement the eLF-HTC identification result  from Section \ref{section_elfhtc}. We present the first complementary result in Section \ref{sec_determinantal_results} and the second complementary result in Section \ref{sec_recursive}.

 \subsection{Determinantal Identification Result}
\label{sec_determinantal_results}
This section revisits an identification result from \citet{weihs2018determinantal} which applies to LSEMs where latent variables are projected away into dependent error terms, and transfers this identification result to our considered setting with explicitly modelled latent variables. This transfer works as follows:
The originally considered graphs in \citet{weihs2018determinantal} are mixed graphs potentially including both directed and bidirected edges. We then apply the results from \citet{weihs2018determinantal} to the entire latent-factor graphs including the latent variables; the idea behind doing so is that these latent-factor graphs, forgetting about the latentness of the latent variables for a moment, are just directed graphs corresponding to an LSEM without bidirected edges. Finally, we restrict the occuring variable sets in the resulting theorem to only observed variables, so that the resulting theorem can be applied to latent-factor graphs. Along the way, we need to take care of some subtleties regarding the notion of rational identifiability.

The identification result from \citet{weihs2018determinantal} relies on the theory of flow-networks (see, e.g., Section 26 in \citealp{cormen2022introduction} for an introduction). A flow-network is a directed graph whose vertices and edges have a certain given capacity. \citet{weihs2018determinantal} introduce the following flow-network corresponding to each (latent-factor) graph $G^\mathcal{L}$; for an example, see Figure \ref{fig_3}.

\begin{definition}[Determinantal flow-network]
\label{def_flow_graph}
For a given latent-factor graph $G^\mathcal{L}=(V\cup \mathcal{L}, D)$, the corresponding flow-graph $G^\mathcal{L}_{\textnormal{flow,det}} = (V^f\cup \mathcal{L}^f,D^f)$ is a directed graph where $V^f:=V\cup V'$ with $V'$ being a copy of $V$, that is, for each $v\in V$ there corresponds a $v'\in V'$, and where $\mathcal{L}^f:=\mathcal{L}\cup \mathcal{L}'$ with $\mathcal{L}'$ being a copy of $\mathcal{L}$, and where $D^f$ contains the edges
\begin{align*}
\label{eq_flow_graph_rules}
    &i\rightarrow j \text{ if } j\rightarrow i \in D,\\
    &i\rightarrow i' \text{ for all } i\in V\cup \mathcal{L}, \text{ and}\\
    &i'\rightarrow j' \text{ if } i\rightarrow j\in D.\numberthis
\end{align*}
Turn $G^\mathcal{L}_{\textnormal{flow,det}}$ into a network by giving all vertices and edges capacity $1$.
\end{definition}
\begin{remark}
\label{remark_flow_graph_extension}
The original flow-network framework (as, e.g., layed out in Section 26 in \citealp{cormen2022introduction}) does not allow for vertices to have capacities, the presence of multiple sources or targets, or the presence of reciprocal edge pairs $(v,w)$, that is, both $v\rightarrow w$ and $w\rightarrow v$ are in $D^f$, which might occur as we do not prohibit directed cycles. However, extending the original theory to include these further cases is rather straightforward, see, for example, Section 6 in \citet{foygel2012halfSupplement}.
\end{remark}
Transferred to our setting, the identification result from \citet{weihs2018determinantal} relies on calculating maximum flows in the flow-network $G^\mathcal{L}_{\textnormal{flow,det}}$ and a version of $G^\mathcal{L}_{\textnormal{flow,det}}$ where some edges have been deleted. The maximum-flow problem is a well-known problem in computer science, and there exist several standard algorithms for solving it; the oldest one being the Ford-Fulkerson algorithm  (\citealp{ford1956maximal}, also see Section 26 in \citealp{cormen2022introduction}). 
To state this result, we for vertex sets $A=\{a_1,\ldots, a_{n}\}\subseteq V$ and $B=\{b_1,\ldots,b_{n}\} \subseteq V$ write $\Sigma_{A, B}$ to refer to the submatrix of $\Sigma$
whose $(i,j)$-th entry with $i,j\in \{1,\ldots,n\}$ is given by $\Sigma_{a_ib_j}$.

\begin{theorem}[Determinantal identification result]
\label{theorem_determinantal}
    Let $w_0,w_1,\ldots, w_m\in V$. Let $w_0\rightarrow v\in D_V$ and suppose that the edges $w_1\rightarrow v,\ldots, w_m\rightarrow v\in D_V$ are rationally identifiable. Let $\overline{G}^\mathcal{L}_{\textnormal{flow,det}}$ be $G^\mathcal{L}_{\textnormal{flow,det}}$ with the edges $w_0'\rightarrow v',\ldots,w'_m\rightarrow v'$ removed. Suppose there are sets $S\subseteq V$ and $T\subseteq V\setminus \{v,w_0\}$ such that $|S|=|T|+1=k$ and such that
    \begin{enumerate}[label=(\alph*)]
    \item in $G^\mathcal{L}$ it holds that $(\textnormal{dec}(v))\cap (T\cup \{v\})=\emptyset,$
    \item the max-flow from $S$ to $T'\cup \{w'_0\}$ in $G^\mathcal{L}_{\textnormal{flow,det}}$ equals $k$, and
    \item the max-flow from $S$ to $T'\cup \{v'\}$ in $\overline{G}^\mathcal{L}_{\textnormal{flow,det}}$ is $<k$.
    \end{enumerate}
    Then, $w_0\rightarrow v$ is rationally identifiable via
    \begin{align*}
        \lambda_{w_0v}=\frac{\det{(\Sigma_{S, T\cup \{v\}})} - \sum_{i=1}^m\lambda_{w_iv}\det{(\Sigma_{S, T\cup \{w_i\}})}}{\det{(\Sigma_{S, T\cup \{w_0\}})}}.
    \end{align*}
\end{theorem}

We now briefly discuss the following example which  illustrates that already Theorem \ref{theorem_determinantal} alone identifies edges that the LF-HTC identification result (Theorem \ref{theorem_lfhtc}) cannot identify.
		\begin{figure}
		\center
		\begin{subfigure}[t]{0.4\textwidth}
		\scalebox{0.8}{
	\begin{tikzpicture}
	\node[circle, draw, fill = c3, line width=0.4mm] (h1) at (0.75,0) {$h_1$};
	\node[circle, draw, line width=0.4mm] (1) at (-3, -2) {1};
	\node[circle, draw, line width=0.4mm] (2) at (-1.5, -2) {2};
	\node[circle, draw, line width=0.4mm] (3) at (0, -2) {3};
	\node[circle, draw, line width=0.4mm] (4) at (1.5, -2) {4};
	\node[circle, draw, line width=0.4mm] (5) at (3, -2) {5};
	\node[circle, draw, line width=0.4mm] (6) at (4.5, -2) {6};
	
	\path [->, line width = 0.5mm, dashed, color = red] (h1) edge node[left] {} (1);
	\path [->, line width = 0.5mm, dashed, color = red] (h1) edge node[left] {} (2);
	\path [->, line width = 0.5mm, dashed, color = red] (h1) edge node[left] {} (3);
	\path [->, line width = 0.5mm, dashed, color = red] (h1) edge node[left] {} (4);
	\path [->, line width = 0.5mm, dashed, color = red] (h1) edge node[left] {} (5);
	\path [->, line width = 0.5mm, dashed, color = red] (h1) edge node[left] {} (6);
	
	\path [->, line width = 0.5mm, color = blue] (1) edge node[left] {} (2);
	\path [->, line width = 0.5mm, color = blue] (2) edge node[left] {} (3);
	\path [->, line width = 0.5mm, color = blue] (3) edge node[left] {} (4);
	\path [->, line width = 0.5mm, color = blue] (4) edge node[left] {} (5);
	\path [->, line width = 0.5mm, color = blue, bend right] (4) edge node[left] {} (6);
\end{tikzpicture}
}
\subcaption{Latent-factor graph $G^\mathcal{L}$}
\end{subfigure}
\begin{subfigure}[t]{0.4\textwidth}
\scalebox{0.8}{
\begin{tikzpicture}
	\node[circle, draw, fill = c3, line width=0.4mm] (h1) at (-4,-0.8) {$h_1$};
	\node[circle, draw, line width=0.4mm] (1) at (-3, -2) {1};
	\node[circle, draw, line width=0.4mm] (2) at (-1.5, -2) {2};
	\node[circle, draw, line width=0.4mm] (3) at (0, -2) {3};
	\node[circle, draw, line width=0.4mm] (4) at (1.5, -2) {4};
	\node[circle, draw, line width=0.4mm] (5) at (3, -2) {5};
	\node[circle, draw, line width=0.4mm] (6) at (4.5, -2) {6};
	
	\path [->, line width = 0.5mm, dashed, color = red] (1) edge node[left] {} (h1);
	\path [->, line width = 0.5mm, dashed, color = red] (2) edge node[left] {} (h1);
	\path [->, line width = 0.5mm, dashed, color = red] (3) edge node[left] {} (h1);
	\path [->, line width = 0.5mm, dashed, color = red] (4) edge node[left] {} (h1);
	\path [->, line width = 0.5mm, dashed, color = red] (5) edge node[left] {} (h1);
	\path [->, line width = 0.5mm, dashed, color = red] (6) edge node[left] {} (h1);
	
	\path [->, line width = 0.5mm, color = blue] (2) edge node[left] {} (1);
	\path [->, line width = 0.5mm, color = blue] (3) edge node[left] {} (2);
	\path [->, line width = 0.5mm, color = blue] (4) edge node[left] {} (3);
	\path [->, line width = 0.5mm, color = blue] (5) edge node[left] {} (4);
	\path [->, line width = 0.5mm, color = blue, bend right] (6) edge node[left] {} (4);
	
\node[circle, draw, fill = c3, line width=0.4mm] (h1') at (-4,-2.8) {$h'_1$};
	\node[circle, draw, line width=0.4mm] (1') at (-3, -4) {1'};
	\node[circle, draw, line width=0.4mm] (2') at (-1.5, -4) {2'};
	\node[circle, draw, line width=0.4mm] (3') at (0, -4) {3'};
	\node[circle, draw, line width=0.4mm] (4') at (1.5, -4) {4'};
	\node[circle, draw, line width=0.4mm] (5') at (3, -4) {5'};
	\node[circle, draw, line width=0.4mm] (6') at (4.5, -4) {6'};
	
	\path [->, line width = 0.5mm, dashed, color = red] (h1') edge node[left] {} (1');
	\path [->, line width = 0.5mm, dashed, color = red] (h1') edge node[left] {} (2');
	\path [->, line width = 0.5mm, dashed, color = red] (h1') edge node[left] {} (3');
	\path [->, line width = 0.5mm, dashed, color = red] (h1') edge node[left] {} (4');
	\path [->, line width = 0.5mm, dashed, color = red] (h1') edge node[left] {} (5');
	\path [->, line width = 0.5mm, dashed, color = red] (h1') edge node[left] {} (6');
	
	\path [->, line width = 0.5mm, color = blue] (1') edge node[left] {} (2');
	\path [->, line width = 0.5mm, color = blue] (2') edge node[left] {} (3');
	\path [->, line width = 0.5mm, color = blue] (3') edge node[left] {} (4');
	\path [->, line width = 0.5mm, color = blue] (4') edge node[left] {} (5');
	\path [->, line width = 0.5mm, color = blue, bend right] (4') edge node[left] {} (6');
		\path [->, line width = 0.5mm, color = orange] (h1) edge node[left] {} (h1');
		\path [->, line width = 0.5mm, color = orange] (1) edge node[left] {} (1');
		\path [->, line width = 0.5mm, color = orange] (2) edge node[left] {} (2');
		\path [->, line width = 0.5mm, color = orange] (3) edge node[left] {} (3');
		\path [->, line width = 0.5mm, color = orange] (4) edge node[left] {} (4');
		\path [->, line width = 0.5mm, color = orange] (5) edge node[left] {} (5');
		\path [->, line width = 0.5mm, color = orange] (6) edge node[left] {} (6');
\end{tikzpicture}
}
\subcaption{Corresponding flow-graph $G^\mathcal{L}_{\textnormal{flow,det}}$}
\end{subfigure}	
\caption{Example latent-factor graph $G^\mathcal{L}$ and its corresponding flow graph $G^\mathcal{L}_{\textnormal{flow,det}}$ of which the former is not rationally identifiable via the LF-HTC but for which already some edges are identifiable via Theorem \ref{theorem_determinantal} (and all edges are rationally identifiable via an interplay of Theorems \ref{theorem_elfhtc} and \ref{theorem_determinantal}, for example).}
		\label{fig_3}
	\end{figure}
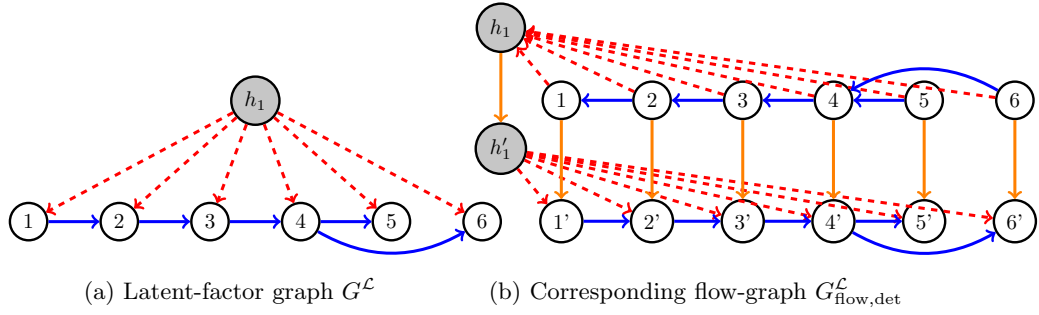
\begin{example}[Determinantal identification example]
Consider the latent-factor graph $G^\mathcal{L}$ given in Figure \ref{fig_3}a with $V=\{1,2,3,4,5,6\}$ and $\mathcal{L}=\{h_1\}$, and also consider the corresponding flow graph $G^\mathcal{L}_{\textnormal{flow,det}}$ in Figure \ref{fig_3}b.
This latent-factor graph $G^\mathcal{L}$ is not LF-HTC identifiable, which one can quickly see, for example, by using the R-package SEMID \citep{Rlanguage, barber2025package} in which the LF-HTC is implemented. However, one can already identify some edges in this graph using Theorem \ref{theorem_determinantal} as we now explain; as a remark, the entire graph $G^\mathcal{L}$ can in fact be identified by a combination of Theorems \ref{theorem_elfhtc} and \ref{theorem_determinantal}, for example, which we will not lay out in detail, however.

	\noindent\underline{$v=5$:} The only observed parent of $v=5$ is $w_0=4$. For the sets $S$ and $T$ take $S=\{2, 3, 4\}$ and $T=\{1, 2\}$. Then, $\lambda_{45}$ is rationally identifiable via
	\begin{align*}
	\lambda_{45} = \det\begin{pmatrix}
			\Sigma_{21} & \Sigma_{22} &\Sigma_{25}\\
			\Sigma_{31} & \Sigma_{32} & \Sigma_{35}\\
			\Sigma_{41} & \Sigma_{42} & \Sigma_{45}
	\end{pmatrix}\biggr/\det\begin{pmatrix}
			\Sigma_{21} & \Sigma_{22} &\Sigma_{24}\\
			\Sigma_{31} & \Sigma_{32} & \Sigma_{34}\\
			\Sigma_{41} & \Sigma_{42} & \Sigma_{44}
		\end{pmatrix}.
	\end{align*}
	
	\noindent\underline{$v=6$:} The only observed parent of $v=6$ is $w_0=4$. For the sets $S$ and $T$ take $S=\{2, 3, 4\}$ and $T=\{1, 2\}$. Then, $\lambda_{45}$ is rationally identifiable via
	\begin{align*}
		\lambda_{46} = \det\begin{pmatrix}
				\Sigma_{21} & \Sigma_{22} &\Sigma_{26}\\
				\Sigma_{31} & \Sigma_{32} & \Sigma_{36}\\
				\Sigma_{41} & \Sigma_{42} & \Sigma_{46}
		\end{pmatrix}\biggr/\det\begin{pmatrix}
				\Sigma_{21} & \Sigma_{22} &\Sigma_{24}\\
				\Sigma_{31} & \Sigma_{32} & \Sigma_{34}\\
				\Sigma_{41} & \Sigma_{42} & \Sigma_{44}
		\end{pmatrix}.
	\end{align*}
\end{example}
\vspace{0.25cm}

We now finish this section with a remark geared towards
a practical implementation.
\begin{remark}[Simplification]
\label{remark_deleted_edges} 
Suppose one wants to identify some edge $w_0\rightarrow v$. Then, it suffices to check Theorem \ref{theorem_determinantal} when deleting all already rationally identified edges $w_1'\rightarrow v',\ldots, w_m'\rightarrow v'$ at once. Deleting only subsets of $w_1'\rightarrow v',\ldots, w_m'\rightarrow v'$ and then checking Theorem \ref{theorem_determinantal} for all different subsets does not yield more direct causal effects that can be identified via Theorem \ref{theorem_determinantal}.

To see this fact, note that the assumptions in Theorem \ref{theorem_determinantal} when deleting all edges $w_1'\rightarrow v',\ldots, w_m'\rightarrow v'$ at once are always satisfied  if they are satisfied when only deleting some subset of $w_1'\rightarrow v',\ldots, w_m'\rightarrow v'$. This implication is true because fewer edges in a flow graph can only reduce (and not increase) the maximum flow.
\end{remark}

    \subsection{Recursively Checking Identifiability in Subgraphs}
\label{sec_recursive}
In this section, we introduce a meta identification result that is about checking different identification results (such as the eLF-HTC identification result from Theorem \ref{theorem_elfhtc} and the determinantal identification result from Theorem \ref{theorem_determinantal}) in subgraphs of the original latent-factor graph $G^\mathcal{L}$. These subgraphs are created by deleting certain edges from $G^\mathcal{L}$ that are rationally identifiable. Our meta approach is 
motivated from the deletion-of-edges procedure already discussed in Section \ref{sec_determinantal_results}, which itself was based on \citet{weihs2018determinantal}.
However, our meta approach is more general and we are not aware of it being discussed in the literature.

When checking identifiability of $G^\mathcal{L}$ using some subgraph $\overline{G}^\mathcal{L}$ of $G^\mathcal{L}$, one must use the covariance matrix $\overline{\Sigma}$ corresponding to that subgraph $\overline{G}^\mathcal{L}$. In principal, $\overline{\Sigma}$ might differ from the original covariance matrix $\Sigma$ corresponding to $G^\mathcal{L}$. However, some entries of $\overline{\Sigma}$ can be calculated using the original entries of $\Sigma$ as the following proposition asserts. We note that similar results about calculable covariances in subgraphs have already occurred in \citet{weihs2018determinantal}.
\begin{proposition}[Calculable covariances in subgraphs]
\label{allowed_covariances}
Suppose that for a vertex $v\in V$ the edges $w_1\rightarrow v,w_2\rightarrow v,\ldots, w_m\rightarrow v$ are rationally identifiable where $\{w_1,\ldots,w_m\}\subseteq \textnormal{pa}_V(v)$. Let $\overline{G}^\mathcal{L}$ be the latent-factor graph with all edges $w_1\rightarrow v,w_2\rightarrow v,\ldots, w_m\rightarrow v$ deleted from $G^\mathcal{L}$ and let $\overline{\Sigma}$ be the covariance matrix corresponding to $\overline{G}^\mathcal{L}$. Let $x, y\in V$ with $x,y$ not necessarily distinct. If
\begin{itemize}
    \item \textbf{Case 1:} $x,y\notin \textnormal{dec}_{G^\mathcal{L}}(v)$ and $x\neq v$ and $y\neq v$, or
        \item \textbf{Case 2:} $x, y\notin \textnormal{dec}_{G^\mathcal{L}}(v)$ and $x\neq v$ and $y=v$, or
        \item \textbf{Case 3:} $x, y\notin \textnormal{dec}_{G^\mathcal{L}}(v)$ and $x= v$ and $y\neq v$, 
\end{itemize}
 then the covariance $\overline{\Sigma}_{xy}$ is given by the rational function
\begin{align*}\overline{\Sigma}_{xy}=
\begin{cases}
    \Sigma_{xy},  \text{ in Case 1, }\\
    \Sigma_{xy}-\sum_{i=1}^m\Lambda_{w_iy}\Sigma_{xw_i}  \text{ in Case 2, }\\
        \Sigma_{xy}-\sum_{i=1}^m\Lambda_{w_ix}\Sigma_{yw_i}  \text{ in Case 3 }.
\end{cases}
\end{align*}
\end{proposition}

Proposition \ref{allowed_covariances} justifies applying any identification result to a subgraph $\overline{G}^\mathcal{L}$ of $G$ as long as the covariances $\overline{\Sigma}_{xy}$ required in that identification result can be calculated by Proposition \ref{allowed_covariances}. What remains to be shown is that rational identifiability of an edge in the subgraph $\overline{G}^\mathcal{L}$ implies rational identifiability of the same edge in the original graph $G^\mathcal{L}$. This implication is indeed true as the following proposition asserts.

\begin{proposition}[Rational identifiability in subgraph implies rational identifiability in the original graph]
\label{proposition_id_notion_subgraph}
Suppose that for a latent-factor graph $G^\mathcal{L} = (V\cup \mathcal{L}, D)$ and a vertex $v\in V$ the edges $w_1\rightarrow v,w_2\rightarrow v,\ldots, w_n\rightarrow v$ are rationally identifiable where $\{w_1,\ldots,w_m\}\subseteq \textnormal{pa}_V(v)$. Let $\overline{G}^\mathcal{L}$ be the latent-factor graph $G^\mathcal{L}$ with all edges $w_1\rightarrow v,w_2\rightarrow v,\ldots, w_m\rightarrow v$ from $G^\mathcal{L}$ deleted.
Let $w_0\rightarrow v_0$ with $w_0,v_0\in V$ be rationally identifiable in $\overline{G}^\mathcal{L}$ by just using covariances as specified by Proposition \ref{allowed_covariances}. Then, $w_0\rightarrow v_0$ is rationally identifiable in $G^\mathcal{L}$.
\end{proposition}

Propositions \ref{allowed_covariances} and \ref{proposition_id_notion_subgraph} thus give rise to the following procedure: Recursively delete all already rationally identified edges from a given latent-factor graph $G^\mathcal{L}$ and check identifiability in the resulting ever smaller subgraphs using any existing identification result. For each new recursion level, the set of allowed covariances, so the set of covariances that any identification result might use, is thereby calculated by applying Proposition \ref{allowed_covariances} with only the allowed covariances from the previous recursion level as inputs. This notion of ``allowedness" is formalized in the following definition.
\begin{definition}[Allowed covariance]
\label{allowed_def}
Consider a latent-factor graph $G^\mathcal{L}=(V\cup \mathcal{L},D)$. For a sequence of not necessarily distinct vertices $v_1,\ldots, v_n\in V$, consider the ordered sequence of pairwise disjoint edge sets $D_{v_1},\allowbreak\ldots,\allowbreak D_{v_n}\subseteq D_V$, referred to as \emph{deletion sequence}, such that for all $i\in\{1,\ldots, n\}$ every edge in $D_{v_i}$ points into $v_i$. Let $G^\mathcal{L}_{n}$ denote the latent-factor graph that is defined by deleting from $G^\mathcal{L}$ all the edges that are in $\bigcup_{i=1}^n D_{v_i}$.  Denote the covariance matrix corresponding to $G^\mathcal{L}_{n}$ by $\Sigma_{n}$. Consider $x,y\in V$ not necessarily distinct.

If $n=1$, then a covariance $(\Sigma_n)_{xy}$ is called allowed if Proposition \ref{allowed_covariances} applies. If $n>1$, then a covariance $(\Sigma_n)_{xy}$ is called allowed (with respect to the deletion sequence  $D_{v_1},\allowbreak\ldots,\allowbreak D_{v_n}$) if the following three conditions are satisfied:
\begin{itemize}
    \item[(i)] The conditions of one of the three cases from Proposition \ref{allowed_covariances} with $G^{\mathcal{L}}_{n-1}$ as input are satisfied,
    \item[(ii)] $(\Sigma_{n-1})_{xy}$ is allowed with respect to $G^{\mathcal{L}}_{n-1}$, and 
    \item[(iii)] additionally in Case 2 of Proposition \ref{allowed_covariances} all $(\Sigma_{n-1})_{xw}$ with $w\rightarrow v_n\in D_{v_n}$, and additionally in Case 3 of Proposition \ref{allowed_covariances} all $(\Sigma_{n-1})_{yw}$ with $w\rightarrow v_n\in D_{v_n}$ are allowed with respect to $G^{\mathcal{L}}_{n-1}$.
\end{itemize}
\end{definition}
As we show in the next proposition, the set of allowed covariances only depends on the union of all deleted edges and neither on their deletion order nor on the amount of edges deleted per step.
\begin{proposition}
\label{propo_deletion_procedure}
Let $G^\mathcal{L}=(V, D)$ be a latent-factor graph and let $D_{\textnormal{del}}\subseteq D_V$. Write $G^\mathcal{L}_{\textnormal{del}}$ for the subgraph of $G^\mathcal{L}$ with the edges from $D_{\textnormal{del}}$ removed. Then, every deletion sequence $D_{v_1},\ldots,D_{v_n}\subseteq D_V$ (as introduced in Definition \ref{allowed_def}) with $n\in \mathbb{N}$ and $\bigcup_{i=1}^nD_{v_i}=D_{\textnormal{del}}$ yields the same set of allowed covariances with respect to $G^\mathcal{L}_{\textnormal{del}}$.
\end{proposition}

Note that it might be the case that edges in subgraphs are not rationally identifiable while they are in the original supergraph even when the same covariances are allowed for the subgraph as for the supergraph---for example, consider the classical instrumental variables setup when removing the edge between the instrument and the treatment. Therefore, and due to Proposition \ref{propo_deletion_procedure}, in practice we create the subgraphs by recursively deleting one rationally identified edge at a time and not by deleting several rationally identified edges at once.
\addtocounter{example}{-1}
\begin{example}[continued]
Consider again the latent-factor graph $G^\mathcal{L}$ from Figure \ref{fig_3}a. As argued in Section \ref{sec_determinantal_results}, $\lambda_{45}$ and $\lambda_{46}$ are rationally identifiable via Theorem \ref{theorem_determinantal}. Now, delete the edge $4\rightarrow 5$ from $G^\mathcal{L}$. In the new graph $\overline{G}^\mathcal{L}$ with corresponding covariance matrix $\overline{\Sigma}$, we can further rationally identify $\lambda_{23}$ (first in $\overline{G}^\mathcal{L}$ and by Proposition \ref{proposition_id_notion_subgraph} then in $G^\mathcal{L}$) by applying Theorem \ref{theorem_determinantal} with $S=\{1, 2\}$ and $T=\{5\}$, namely
\begin{align*}
\label{eq_example_subgraph}
   	\lambda_{23} = \det\begin{pmatrix}
			\overline{\Sigma}_{15} & \overline{\Sigma}_{13}\\
			\overline{\Sigma}_{25} & \overline{\Sigma}_{23}
	\end{pmatrix}\biggr/\det\begin{pmatrix}
			\overline{\Sigma}_{15} & \overline{\Sigma}_{12}\\
			\overline{\Sigma}_{25} & \overline{\Sigma}_{22}
		\end{pmatrix}.\numberthis
\end{align*}
 Here, it holds that all occuring covariances in equation \eqref{eq_example_subgraph} are allowed and that $\overline{\Sigma}_{15}=\Sigma_{15}-\lambda_{45}\Sigma_{14}$, $\overline{\Sigma}_{13}=\Sigma_{13}$, $\overline{\Sigma}_{25}=\Sigma_{25}-\lambda_{45}\Sigma_{24}$, $\overline{\Sigma}_{23}=\Sigma_{23}$, $\overline{\Sigma}_{12}=\Sigma_{12}$, $\overline{\Sigma}_{22}=\Sigma_{22}$.
By applying Theorem \ref{theorem_determinantal} in a similar recursive fashion, we can identify all further remaining observed edges (which we, however, do not explicitly present).
\end{example}
\addtocounter{example}{5}
We finish this section by emphasizing that the recursive subgraph idea is \emph{not} already contained in the determinantal and eLF-HTC identification results: As we show in our numerical experiments in Section \ref{sec_appendix_numerical_experiments}, our recursive subgraph procedure yields strictly more identifiable cases both when combined with the determinantal and eLF-HTC identification results, respectively, and when combined with a combination of the determinantal and eLF-HTC identification results. For the determinantal identification result (Theorems \ref{theorem_determinantal}), a combination with the recursive subgraph procedure might yield more identifiability because for this combination rationally identifiable edges pointing into different vertices are allowed to be deleted instead of just edges pointing into the same vertex as in the original determinantal identification result. For the eLF-HTC identification result, deleting edges other than the ones pointing into the target vertex $v$ might yield smaller sets $W_z$ potentially making it easier to satisfy the requirements from the eLF-HTC identification result.

\section{Computation}
\label{section_computation}
In this section, we present an identification criterion/algorithm, namely Algorithm \ref{algo_combined_identification}, that combines the eLF-HTC identification result from Section  \ref{section_elfhtc},  the determinantal identification result from Section \ref{sec_determinantal_results}, and the recursive subgraph identification result from Section \ref{sec_recursive}.

This combined identification algorithm repeatedly applies these identification results and updates the set of all already rationally identified edges 
 $S_{\textnormal{edges}}$ along the way. If for some node $v\in V$ all incoming observed edges have been rationally identified, then $v$ is called ``solved" and in this case the combined identification algorithm adds $v$ to the set of solved nodes $S_{\textnormal{nodes}}$.

Hereby, computationally checking the determinantal identification result (Algorithm \ref{algo_determinantal}) is rather straightforward. It suffices to say that for a particular vertex $v\in V\setminus S_{\textnormal{nodes}}$ one just checks for all not identified edges $w_0\rightarrow v\in (\textnormal{pa}_V(v)\times \{v\})\setminus S_{\textnormal{edges}}$ and for all $S \subseteq V$ and $T\subseteq V\setminus \{v, w_0\}$ with $|S|=|T|+1$ whether the requirements from Theorem \ref{theorem_determinantal} are satisfied and whether the required covariances are allowed in the sense of Proposition \ref{allowed_covariances}. In particular, recall the suggested simplification from Remark \ref{remark_deleted_edges} in Section \ref{sec_determinantal_results}. The determinantal subprocedure then stops if either all $w_0\rightarrow v\in (\textnormal{pa}_V(v)\times \{v\})\setminus S_{\textnormal{edges}}$ have been identified, or when there is $w_0\rightarrow v\in (\textnormal{pa}_V(v)\times \{v\})\setminus S_{\textnormal{edges}}$ for which all possible sets $S$ and $T$ have been considered.

\begin{algorithm}
\caption{Combined identification algorithm}	
\begin{algorithmic}[1]
	\Require Latent-factor graph $G^{\mathcal{L}}=(V\cup \mathcal{L}, D)$; set of solved nodes $S_{\textnormal{nodes}}$ with $S_{\textnormal{nodes}}=\emptyset$ in the initial call; set of solved edges $S_{\textnormal{edges}}$ with $S_{\textnormal{edges}}=\emptyset$ in the initial call; set of allowed covariances $\Sigma_{\textnormal{allowed}}\subseteq \{\Sigma_{xy}:\;x, y\in V\}$ with equality in the initial call.
	\Repeat 
	\State $S_{\textnormal{edges,old}}\gets S_{\textnormal{edges}}$
	\For{$v\in V\setminus S_{\textnormal{nodes}}$}
	\State $S_{\textnormal{edges}}\gets \texttt{eLF-HTC subprocedure}(G^\mathcal{L}, v, S_{\textnormal{edges}}, S_{\textnormal{nodes}},\Sigma_{\textnormal{allowed}})$
	\If{all edges $w\rightarrow v\in \textnormal{pa}_V(v)\times \{v\}$ have been solved} 
	\State $S_{\textnormal{nodes}} \gets S_{\textnormal{nodes}}\cup \{v\}$
	\State \textbf{continue}
	\EndIf
	\State $S_{\textnormal{edges}}\gets \texttt{Determinantal subprocedure}(G^\mathcal{L}, v, S_{\textnormal{edges}},\Sigma_{\textnormal{allowed}})$
	\If{all edges $w\rightarrow v\in \textnormal{pa}_V(v)\times \{v\}$ have been solved} 
	\State $S_{\textnormal{nodes}} \gets S_{\textnormal{nodes}}\cup \{v\}$
	\EndIf
	
	\EndFor
\If{$S_{\textnormal{nodes}}=V$}
\State \textbf{break}
\EndIf
	\For{$w\rightarrow v\in S_{\textnormal{edges}}$}	\Comment{Recursive part}
	\State $\overline{G}^\mathcal{L}\leftarrow G^\mathcal{L}$ with $w\rightarrow  v$ removed
	\State $S_{\textnormal{edges,sub}}\leftarrow S_{\textnormal{edges}}\setminus \{w\rightarrow v\}$
	\State $\Sigma_{\textnormal{allowed,sub}}\leftarrow$ Proposition \ref{allowed_covariances} with $\Sigma_{\textnormal{allowed}}$ as input
	\State $S_{\textnormal{edges}}\leftarrow \texttt{Combined identification algorithm}(\overline{G}^\mathcal{L}, \Sigma_{\textnormal{allowed,sub}}, S_{\textnormal{nodes}}, S_{\textnormal{edges,sub}})$
	\If{all edges $\Tilde{w}\rightarrow v\in \textnormal{pa}_V(v)\times \{v\}$ have been solved} 
	\State $S_{\textnormal{nodes}} \gets S_{\textnormal{nodes}}\cup \{v\}$
	\EndIf
	\If{$S_{\textnormal{nodes}}=V$} 
\State \textbf{break}
\EndIf
	\EndFor
	\Until{$S_{\textnormal{nodes}}=V$ or $S_{\textnormal{edges}} = S_{\textnormal{edges,old}}$}
\State \Return{$S_{\textnormal{edges}}$}	
\end{algorithmic}
\label{algo_combined_identification}
\end{algorithm}

The algorithmic implementation of the eLF-HTC (Algorithm \ref{algo_elfhtc_sub}) requires more explanation: Similarly to the determinantal identification result, we also make use of max-flow computations to computationally check the eLF-HTC (see, e.g., Section 26 in \citealp{cormen2022introduction} for an introduction). These max-flow computations and accompanying results for the eLF-HTC are modifications of similar max-flow computations and results for the LF-HTC in \citet{barber2022half}, which were itself based on \citet{foygel2012half}, and there are some differences to the max-flow computations for the determinantal results from Section \ref{sec_determinantal_results}.

To begin with, suppose for a moment that we are given fixed sets $H,Z, W_v, (W_z)_{z\in Z}$ as introduced in the eLF-HTC such that $|Z|=|H|$, such that $Z\cap \{v\}=\emptyset$ and such that $Z^{(1)}\cap (W_Z\cup W_v)=\emptyset$, so such that all requirements in condition (i) of the eLF-HTC not including $Y$ are satisfied.
Because the second part of condition (ii) of the eLF-HTC is equivalent to $Y\cap \textnormal{ch}(\textnormal{pa}_\mathcal{L}(Z\cup \{v\})\setminus H)=\emptyset$, the set 
 $A=V\setminus (Z\cup \{v\}\cup \textnormal{ch}(\textnormal{pa}_\mathcal{L}(Z\cup \{v\})\setminus H))$ is the set of ``allowed" nodes that $Y$ is allowed to contain such that condition (ii) of the eLF-HTC is satisfied.

To check condition (iii) in the eLF-HTC, we define the flow graph $G^\mathcal{L}_{\textnormal{flow,eLF-HTC}}(v,\allowbreak A,\allowbreak Z,\allowbreak W_Z,\allowbreak W_v)=(V_f,\allowbreak D_f)$ corresponding to the latent-factor graph $G^\mathcal{L}$ having nodes $V_f=(A\cup \mathcal L)\cup (V'\cup \mathcal{L}')$, where $V'$ and $\mathcal{L}'$ are copies of $V$ and $\mathcal{L}$, respectively. The edge set $D_f$ is comprised of the edges
\begin{enumerate}[label=(\alph*)]
\item $a\rightarrow w$ if $a\in A$ and $w\rightarrow a\in D_{\mathcal{L}V}$,
\item $w\rightarrow w'$ for all $w\in A\cup \mathcal{L}$, and
\item $u'\rightarrow w'$ for all $u\rightarrow w\in D_{\mathcal{L}V}$ and for all $u\rightarrow w\in D_V$ such that $w\notin Z$.
\end{enumerate}
 For a flow graph $G^\mathcal{L}_{\textnormal{flow,eLF-HTC}}(v,\allowbreak A,\allowbreak Z,\allowbreak W_Z,\allowbreak W_v)$ corresponding to Figure \ref{fig_0}a, see Figure \ref{fig_1}.

We now turn $G^\mathcal{L}_{\textnormal{flow,eLF-HTC}}$ into a network by assigning to all edges capacity $\infty$ and to all nodes capacity $1$. In practice, $\infty$ can be replaced by $|W_v\cup Z\cup W_Z|$ as no flow can exceed $|W_v\cup Z\cup W_Z|$ in size.

\begin{algorithm}
\caption{Determinantal subprocedure}	
\begin{algorithmic}[1]
	\Require Latent-factor graph $G^{\mathcal{L}}=(V\cup \mathcal{L}, D)$, set of solved edges $S_{\textnormal{edges}}$, set of solved nodes  $S_{\textnormal{nodes}}$,  target vertex $v$,
	  set of allowed covariances $\Sigma_{\textnormal{allowed}}\subseteq \{\Sigma_{xy}:\;x, y\in V\}$.
	  \For{$w_0\in V$ with $w_0\rightarrow v\notin S_{\textnormal{edges}}$}
	  \For{$S\subseteq V$, $T\subseteq V\setminus \{v, w_0\}$ with $|S|=|T|+1$}
	  \State Construct $G^\mathcal{L}_{\textnormal{flow,det}}$ and $\overline{G}^\mathcal{L}_{\textnormal{flow,det}}$ by removing all edges of the form $\Tilde{w}\rightarrow v\in S_{\textnormal{edges}}$ from $G^\mathcal{L}_{\textnormal{flow,det}}$
	  \State Check conditions from Theorem \ref{theorem_determinantal}
	  \If{Theorem \ref{theorem_determinantal} applies and all covariances in $\Sigma_{S, T\cup \{v,w_0\}}$ and $\Sigma_{S, \Tilde{w}}$ with $\Tilde{w}\rightarrow v\in S_{\textnormal{edges}}$ are allowed}
	  \State Add $w_0\rightarrow v$ to $S_{\textnormal{edges}}$
	  \State \textbf{break}
	  \EndIf
	  \EndFor
	  \EndFor
\end{algorithmic}
\label{algo_determinantal}
\end{algorithm}

\begin{algorithm}
\caption{eLF-HTC subprocedure}	
\begin{algorithmic}[1]
	\Require Latent-factor graph $G^{\mathcal{L}}=(V\cup \mathcal{L}, D)$, set of solved edges $S_{\textnormal{edges}}$, set of solved nodes  $S_{\textnormal{nodes}}$,  target vertex $v$,
	  set of allowed covariances $\Sigma_{\textnormal{allowed}}\subseteq \{\Sigma_{xy}:\;x, y\in V\}$.
	 \Ensure Let $W_v$ contain the not-identified edges in $D_V$ into $v$.
		
			\For{$H\subseteq\mathcal{L}_{\geq 4}$}
				\For{$Z\subseteq \textnormal{ch}(H)\setminus \{v\}$ such that $|Z|=|H|$}
				\For{all $(W_z)_{z\in Z}$ where $W_z\subseteq\textnormal{pa}_V(z)$ is a superset of the observed parents for which the corresponding edge into $z$ has not been identified} 
				    \State $W_Z\gets\cup_{z\in Z}W_z$
				    \State $Z^{(1)}\gets\{z\in Z:\;W_z\subsetneq\textnormal{pa}_V(z)\}$ and $Z^{(2)}\gets \{z\in Z:\;W_z=\textnormal{pa}_V(z)\}$
							   
			\If{$Z^{(1)}\cap (W_Z\cup W_v)\neq \emptyset$}
			\State \textbf{next}
			\EndIf

					\State $A\gets V\setminus (Z\cup \{v\} \cup \textnormal{ch}(\textnormal{pa}_{\mathcal{L}}(Z\cup \{v\})\setminus H) \cup (\textnormal{htr}_H(Z\cup \{v\})\setminus  S_{\textnormal{nodes}}))$.
							\State Delete all elements from $A$ yielding covariances not in $\Sigma_{\textnormal{allowed}}$
						\If{Max-flow in $G_{\textnormal{flow,eLF-HTC}}(v, A, Z, W_Z, W_v)$ from $A$  to $W_v\cup Z\cup W_Z$ equals  $|W_v\cup Z\cup W_Z|$}
						\State Add all edges $p^v\rightarrow v$ with $p^v\in W_v\setminus (Z\cup W_Z)$ to $S_{\textnormal{edges}}$
						\State Update $W_v$ to contain the not-identified edges in $D_V$ into $v$
						\EndIf
						\If{$W_v = \emptyset$}	\State\textbf{break}
						\EndIf
					\EndFor
					\If{$W_v = \emptyset$}
			\State\textbf{break}
			\EndIf
					\EndFor
					\If{$W_v = \emptyset$}
		        	\State\textbf{break}
		        	\EndIf
			\EndFor
	\State\Return{$S_{\textnormal{edges}}$}
\end{algorithmic}
\label{algo_elfhtc_sub}
\end{algorithm}

It is now possible to prove the following theorem, which is similar to Theorem 5.1 in \citet{barber2022half} and which provides a formal connection between the eLF-HTC and the max-flow calculations in $G^\mathcal{L}_{\textnormal{flow,eLF-HTC}}(v, A, Z, W_Z, W_v)$. We want to emphasize at this point that Remark \ref{remark_flow_graph_extension} from Section \ref{sec_determinantal_results} also applies to this following theorem implying that the eLF-HTC max-flow calculations can be easily transferred to a standard  format for which there exist standard algorithms in computer science (see, e.g., \citealp{ford1956maximal} or Section 26 in \citealp{cormen2022introduction}).

		\begin{figure}
		\center
\scalebox{0.8}{
\begin{tikzpicture}
	\node[circle, draw, fill = c3, line width=0.4mm] (h1) at (-4,-0.8) {$h_1$};
	\node[circle, draw, line width=0.4mm] (1) at (-3, -2) {1};
	\node[circle, draw, line width=0.4mm] (2) at (-1.5, -2) {2};
	\node[circle, draw, line width=0.4mm] (3) at (0, -2) {3};
	\node[circle, draw, line width=0.4mm] (5) at (3, -2) {5};

	\path [->, line width = 0.5mm, dashed, color = red] (1) edge node[left] {} (h1);
	\path [->, line width = 0.5mm, dashed, color = red] (2) edge node[left] {} (h1);
	\path [->, line width = 0.5mm, dashed, color = red] (3) edge node[left] {} (h1);

	\path [->, line width = 0.5mm, dashed, color = red] (5) edge node[left] {} (h1);

\node[circle, draw, fill = c3, line width=0.4mm] (h1') at (-4,-2.8) {$h'_1$};
	\node[circle, draw, line width=0.4mm] (1') at (-3, -4) {1'};
	\node[circle, draw, line width=0.4mm] (2') at (-1.5, -4) {2'};
	\node[circle, draw, line width=0.4mm] (3') at (0, -4) {3'};
	\node[circle, draw, line width=0.4mm] (4') at (1.5, -4) {4'};
	\node[circle, draw, line width=0.4mm] (5') at (3, -4) {5'};
	\node[circle, draw, line width=0.4mm] (6') at (4.5, -4) {6'};
	
	\path [->, line width = 0.5mm, dashed, color = red] (h1') edge node[left] {} (1');
	\path [->, line width = 0.5mm, dashed, color = red] (h1') edge node[left] {} (2');
	\path [->, line width = 0.5mm, dashed, color = red] (h1') edge node[left] {} (3');
	\path [->, line width = 0.5mm, dashed, color = red] (h1') edge node[left] {} (4');
	\path [->, line width = 0.5mm, dashed, color = red] (h1') edge node[left] {} (5');
	\path [->, line width = 0.5mm, dashed, color = red] (h1') edge node[left] {} (6');
	
	\path [->, line width = 0.5mm, color = blue] (1') edge node[left] {} (2');
	\path [->, line width = 0.5mm, color = blue] (2') edge node[left] {} (3');
	\path [->, line width = 0.5mm, color = blue] (3') edge node[left] {} (4');
	\path [->, line width = 0.5mm, color = blue] (4') edge node[left] {} (5');
		\path [->, line width = 0.5mm, color = blue, bend right] (1') edge node[left] {} (5');

		\path [->, line width = 0.5mm, color = orange] (h1) edge node[left] {} (h1');
		\path [->, line width = 0.5mm, color = orange] (1) edge node[left] {} (1');
		\path [->, line width = 0.5mm, color = orange] (2) edge node[left] {} (2');
		\path [->, line width = 0.5mm, color = orange] (3) edge node[left] {} (3');
		\path [->, line width = 0.5mm, color = orange] (5) edge node[left] {} (5');

\end{tikzpicture}
}

\caption{Flow-graph $G^\mathcal{L}_{\textnormal{flow,eLF-HTC}}(v=4, A=\{1,2,3,5\}, Z=\{6\}, W_Z=\{5\}, W_v=\{3\})$ corresponding to Figure \ref{fig_0}a.}
		\label{fig_1}
	\end{figure}
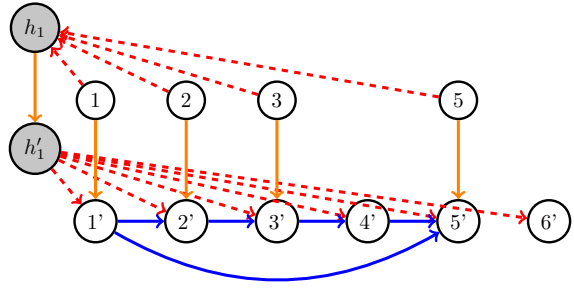 

\begin{theorem} [Checking the eLF-HTC with max-flow computations]
\label{theorem_existence_of_Y}
Let $G^{\mathcal{L}}=(V\cup \mathcal{L}, D)$ be a latent-factor graph. For a fixed node $v\in V$, a fixed set $H\subseteq \mathcal{L}$, a fixed set $W_v\subseteq \textnormal{pa}_V(v)$,
a fixed set $Z\subseteq \textnormal{ch}(H)\setminus \{v\}$ such that $|Z|=|H|$, and fixed sets $W_z\subseteq \textnormal{pa}_V(z)$ for all $z\in Z$ such that $Z^{(1)}\cap (W_Z\cup W_v)= \emptyset$, and for the set of allowed nodes $A=V\setminus (Z\cup \{v\} \cup \textnormal{ch}(\textnormal{pa}_\mathcal{L}(Z\cup \{v\})\setminus H))$ it holds:

The max-flow in $G_{\textnormal{flow,eLF-HTC}}(v, A, Z, W_Z, W_v)$ from $A$  to $W'_v\cup Z'\cup W'_Z$ equals  $|W_v\cup Z\cup W_Z|$ if and only if there exists $Y\subseteq A$ such that $(Y, Z, (W_z)_{z\in Z},  H)$ satisfies the eLF-HTC with respect to $(v, W_v)$.
\end{theorem}

Therefore, to check whether there exists $(Y, Z, (W_z)_{z\in Z},  H)$ and $W_v$ such that the eLF-HTC with respect to $(v, W_v)$ is satisfied, we iterate over all sets $H\subseteq \mathcal{L}$, $Z\subseteq \textnormal{ch}(H)\setminus \{v\}$, equate $W_v$ with the not-already-identified observed parents $w$ of $v$ (which suffices due to Remark \ref{remark_W_v}), and iterate over all supersets $W_z$ of the not-identified observed parents of $z$ for each $z\in Z$ (which is required due to Remark \ref{remark_W_z}). In order to incorporate the recursive identification idea from Section \ref{sec_recursive}, one further needs to delete elements from the set $A$ that yield covariances that are not allowed in the sense of Definition \ref{allowed_def} before applying Theorem \ref{theorem_existence_of_Y}. Also, to ensure that in practice all edges required to be rationally identifiable are already known to be rationally identifiable, we further remove from $A$ the set $(\textnormal{htr}_H(Z\cup \{v\})\setminus  S_{\textnormal{nodes}}))$.

Now, as the following proposition shows (which is again similar to Proposition 5.2 in \citet{barber2022half}), iterating over $H\subseteq \mathcal{L}$ can be simplified to iterating over $H\subseteq \mathcal{L}_{\geq 4}$ where $\mathcal{L}_{\geq 4}:=\{h\in \mathcal{L}:\;|\textnormal{ch}(h)|\geq 4\}$.
\begin{proposition}[Simplification of the eLF-HTC algorithm]
\label{propo_4_children}
Let $G^\mathcal{L}=(V\cup \mathcal{L}, D)$ be a latent-factor graph and let $v\in V$. If $(Y, Z, (W_z)_{z\in Z},  H)$ satisfies the eLF-HTC with respect to $(v, W_v)$ and there is a node $h\in H$ such that $|\textnormal{ch}(h)|\leq 3$, then there are strict subsets $\Tilde{Y}\subsetneq Y$ and $\Tilde{Z}\subsetneq Z$ such that $(\Tilde{Y}, \Tilde{Z}, (W_z)_{z\in \Tilde{z}},  \Tilde{H})$ with $\Tilde{H}=H\setminus \{h\}$ satisfies the eLF-HTC for $(v, W_v)$ as well.
\end{proposition}

We now briefly go back to the combined identification algorithm (Algorithm \ref{algo_combined_identification}) and discuss its computational complexity: Some parts of the combined identification algorithm do clearly not scale polynomially with increasing observed node set $V$ or latent node set $\mathcal{L}$. For example, already in the determinantal subprocedure the iteration over the sets $S$ and $T$ scales exponentially in the worst case because one essentially iterates over all potential subsets of $V$. Similarly, in line 1 of the eLF-HTC subprocedure, one iterates over different sets $H\subseteq \mathcal{L}_{\geq 4}$, and in line 3 of the eLF-HTC subprocedure, one iterates over the $W_z$'s which also both scales (at least) exponentially. Recursively deleting edges from a subgraph also does not scale polynomially as there are already exponentially many subgraphs.

However, as Theorem \ref{theorem_complexity} in the following states, the time complexity of the combined identification algorithm (Algorithm \ref{theorem_complexity}) still is much better than the computer algebra methods that typically have double exponential time complexity and which even for simple graphs do not finish in ``reasonable amounts of time" on a standard computer. Furthermore, as we also show in Theorem \ref{theorem_complexity}, if one is willing to make some simplifying algorithmic assumptions which still yield a sufficient condition for rational identifiability, one is in fact also able to achieve a polynomial time complexity. For example and as we also numerically illustrate in Section \ref{sec_appendix_numerical_experiments}, simplifications of Algorithm \ref{algo_combined_identification} not including or simplifying the determinantal subprocedure shave off significant computation time while still being effective.

In the following, let $r\leq |D_V|/2$ be the number of reciprocal edge pairs in $D_V$, that is, the number of $v,w\in V$ for which $v\rightarrow w \in D_V$ and $w\rightarrow v\in D_V$.

\begin{theorem} [Computational Complexity]
\label{theorem_complexity}
Without any simplifications, Algorithm \ref{algo_combined_identification} has a time complexity of $\mathcal{O}(|V|^{4|V|^2+3}2^{|\mathcal{L}|}2^{|V|^2+|V|}(|V|+|\mathcal{L}|+r)^3)$.
If
 \begin{itemize}
 \item one restricts the number of to-be-searched pairs of sets $S$ and $T$ in determinantal subprocedure to $c_S$,
     \item only allows for sets $|H|\leq c_H$ in line 1 of the eLF-HTC subprocedure,
     \item one only considers the case where each $W_z$ equals the already identified observed parents of $z$ in line 3 of the eLF-HTC subprocedure, and
     \item bounds the number of edges that can be deleted from $G^\mathcal{L}$ for the recursive part in lines 14--25 of Algorithm \ref{algo_combined_identification} by $c_{\textnormal{rec}}$, 
 \end{itemize}
 then, this simplified version of Algorithm \ref{algo_combined_identification} has a time complexity of $\mathcal{O}(c_S|\mathcal{L}|^{c_H}|V|^{4c_{\textnormal{rec}}+c_H+3}\allowbreak(|V|+|\mathcal{L}|+r)^3)$.
\end{theorem}
In practice, it is thus advisable to start checking identifiability of a given latent-factor graph $G^\mathcal{L}$ with the not-simplified version of Algorithm \ref{algo_combined_identification}. If this not-simplified version of Algorithm \ref{algo_combined_identification} does not finish in a reasonable amount of time, then one might try with some or all of the previously mentioned simplifications. We remark that our provided implementation allows for such simplifications. We investigate several such simplifications in terms of runtime and number of identifiable graphs in Section \ref{sec_appendix_numerical_experiments}.

\section{Numerical Experiments}
\label{sec_numerical_experiments}
		\begin{figure}
		\center
		\begin{subfigure}[t]{0.4\textwidth}
		\scalebox{0.8}{
	\begin{tikzpicture}
	\node[circle, draw, fill = c3, line width=0.4mm] (h1) at (0.75,0) {$h_1$};
	\node[circle, draw, line width=0.4mm] (1) at (-3, -2) {\color{white}1\color{black}};
	\node[circle, draw, line width=0.4mm] (2) at (-1.5, -2) {\color{white}2\color{black}};
	\node[circle, draw, line width=0.4mm] (3) at (0, -2) {\color{white}3\color{black}};
	\node[circle, draw, line width=0.4mm] (4) at (1.5, -2) {\color{white}4\color{black}};
	\node[circle, draw, line width=0.4mm] (5) at (3, -2) {\color{white}5\color{black}};
	\node[circle, draw, line width=0.4mm] (6) at (4.5, -2) {\color{white}6\color{black}};
	
	\path [->, line width = 0.5mm, dashed, color = red] (h1) edge node[left] {} (1);
	\path [->, line width = 0.5mm, dashed, color = red] (h1) edge node[left] {} (2);
	\path [->, line width = 0.5mm, dashed, color = red] (h1) edge node[left] {} (3);
	\path [->, line width = 0.5mm, dashed, color = red] (h1) edge node[left] {} (4);
	\path [->, line width = 0.5mm, dashed, color = red] (h1) edge node[left] {} (5);
	\path [->, line width = 0.5mm, dashed, color = red] (h1) edge node[left] {} (6);

\end{tikzpicture}
}
\subcaption{Latent structure of unlabeled latent-factor graphs for Table \ref{table1}}
\end{subfigure}
\begin{subfigure}[t]{0.4\textwidth}
\scalebox{0.8}{
	\begin{tikzpicture}
	\node[circle, draw, fill = c3, line width=0.4mm] (h1) at (-0.75,0) {$h_1$};
	\node[circle, draw, fill = c3, line width=0.4mm] (h2) at (3,0) {$h_2$};
	\node[circle, draw, line width=0.4mm] (1) at (-3, -2) {\color{white}1\color{black}};
	\node[circle, draw, line width=0.4mm] (2) at (-1.5, -2) {\color{white}2\color{black}};
	\node[circle, draw, line width=0.4mm] (3) at (0, -2) {\color{white}3\color{black}};
	\node[circle, draw, line width=0.4mm] (4) at (1.5, -2) {\color{white}4\color{black}};
	\node[circle, draw, line width=0.4mm] (5) at (3, -2) {\color{white}5\color{black}};
	\node[circle, draw, line width=0.4mm] (6) at (4.5, -2) {\color{white}6\color{black}};
	
	\path [->, line width = 0.5mm, dashed, color = red] (h1) edge node[left] {} (1);
	\path [->, line width = 0.5mm, dashed, color = red] (h1) edge node[left] {} (2);
	\path [->, line width = 0.5mm, dashed, color = red] (h1) edge node[left] {} (3);
	\path [->, line width = 0.5mm, dashed, color = red] (h1) edge node[left] {} (4);
		\path [->, line width = 0.5mm, dashed, color = red] (h2) edge node[left] {} (4);
	\path [->, line width = 0.5mm, dashed, color = red] (h2) edge node[left] {} (5);
	\path [->, line width = 0.5mm, dashed, color = red] (h2) edge node[left] {} (6);

\end{tikzpicture}
}
\subcaption{Latent structure of unlabeled latent-factor graphs for Table \ref{table2}}
\end{subfigure}	
\caption{Considered latent structures for Section \ref{sec_numerical_experiments} and Section \ref{sec_appendix_numerical_experiments}.}
		\label{fig_3.5}
	\end{figure}
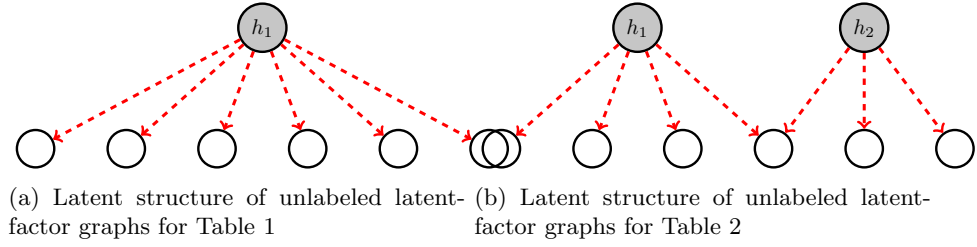
In this section, we numerically 
compare our criterion to the LF-HTC, which is, to our knowledge, the state-of-the-art criterion for rational identifiability in latent-factor graphs. 

We consider the same simulation setup as in \citet{barber2022half}. 
In the first setting (see Figure \ref{fig_3.5}a), there is one latent variable and 6 observed variables and there is a directed edge from the latent variable to each observed variable. For the different subcases $|D_V|\in\{0,1,\ldots, 9\}$, we then consider all latent-factor graphs (up to topological equivalence) with that respective number of observed edges and with that just mentioned latent structure. We then present the number of rationally identifiable latent-factor graphs for the LF-HTC and for the combined identification algorithm (Algorithm \ref{algo_combined_identification}) in Table \ref{table1}. In this table, we also present the number of rationally identifiable graphs which one can obtain using computer algebra, see Section C in \citet{barber2022supplement}.\footnote{\label{footnote_computer_algebra}While the approach discussed in Section C in \citet{barber2022supplement} states with reasonable accuracy whether a given latent-factor graph is rationally identifiable or not, this approach is still not able to generate rational formulas. In contrast, the LF-HTC and our identification results are able to produce such identification formulas.}   In Table \ref{table1_app} in Section \ref{sec_appendix_numerical_experiments}, we consider further amalgamations of the eLF-HTC, determinantal and recursive subgraph identification approaches, to illustrate the individual contributions each identification result makes. In Figure \ref{fig_runtime} in Section \ref{sec_appendix_numerical_experiments}, we also present the runtimes of these different amalgations.

We then also consider a second setting (see Figure \ref{fig_3.5}b), where there are again 6 observed variables but now 2 latent variables. The first latent variable points to 4 of the observed variables, the second latent variable points to the remaining two observed variables and to one of the other 4 observed variables. Then, similarly as in the first setting, we consider for different $|D_V|$ the number of identifiable cases. We present a condensation of the results for the second setting in Table \ref{table2} and in more detail in Table \ref{table2_app} and Figure \ref{fig_runtime} in Section \ref{sec_appendix_numerical_experiments}.

For both numerical settings, one can see that our combined identification algorithm identifies significantly more latent-factor graphs, especially, when $|D_V|$ is rather large. We emphasize that we can identify approximately 98\% of all identifiable graphs in both of these two settings, in comparison to 25\% and 71\% for the LF-HTC, respectively. For dense settings, the improvement is even stronger. For example, consider the case $|D_V|=9$ in Table \ref{table1}, where the LF-HTC identifies less than 1\% of all identifiable graphs compared to 96\% for our combined approach.
In the aforementioned tables in Section \ref{sec_appendix_numerical_experiments}, one can also see that each of the three discussed identification results, so the eLF-HTC, determinantal and recursive subgraph identification results, yield a strict improvement when compared to just using the other two results. In addition, one can see that algorithmic simplifications based on Theorem \ref{theorem_complexity} still identify a large number of graphs.
Moreover, one can see that all three identification results combined are better than just taking the union of the individual results \emph{after} applying them which hints at synergies between the different results. 
Regarding runtime, one can further see that particularly the determinantal subprocedure is responsible for longer runtimes; when simplifying the determinantal subprocedure or when not using it at all, runtimes are significantly shorter.

\definecolor{c100}{rgb}{0, 0.5, 0}
\definecolor{c99}{rgb}{0, 0.49, 0}
\definecolor{c98}{rgb}{0, 0.48, 0}
\definecolor{c97}{rgb}{0, 0.47, 0}
\definecolor{c96}{rgb}{0, 0.46, 0}
\definecolor{c93}{rgb}{0, 0.43, 0}
\definecolor{c87}{rgb}{0, 0.37, 0}
\definecolor{c84}{rgb}{0, 0.34, 0}
\definecolor{c78}{rgb}{0, 0.28, 0}
\definecolor{c71}{rgb}{0, 0.21, 0}
\definecolor{c65}{rgb}{0, 0.15, 0}
\definecolor{c63}{rgb}{0, 0.13, 0}
\definecolor{c31}{rgb}{0.19, 0, 0}
\definecolor{c25}{rgb}{0.25, 0, 0}
\definecolor{c7}{rgb}{0.43, 0, 0}
\definecolor{c0}{rgb}{0.5, 0, 0}
\begin{table}
\begin{center}
\resizebox{0.7\textwidth}{!}{%
	\begin{tabular}{ ccccc }
		\hline
		 $|D_V|$ & Total & Rationally &  LF-HTC  & Det + eLF-HTC + rec \\
		\hline
			0 & 1 & 1 & 1 \footnotesize\color{c100}(100\%) & 1 \footnotesize\color{c100}(100\%) \\
		1 & 1 & 1 & 1 \footnotesize\color{c100}(100\%) & 1 \footnotesize\color{c100}(100\%) \\
		2 & 4 & 4 & 4 \footnotesize\color{c100}(100\%) & 4 \footnotesize\color{c100}(100\%) \\
	3 & 13 & 13 & 13 \footnotesize\color{c100}(100\%) & 13 \footnotesize\color{c100}(100\%) \\
	4 & 51 & 51 & 50 \color{c98}\footnotesize($\approx$98\%) & 51 \footnotesize\color{c100}(100\%) \\
	5 & 163 & 159 & 134 \color{c84}\footnotesize($\approx$84\%) & 159 \footnotesize\color{c100}(100\%) \\
	6 & 407 & 398 & 250 \color{c63}\footnotesize($\approx$63\%) & 398 \color{c100}\footnotesize(100\%) \\
	7 & 796 & 747 & 234 \color{c31}\footnotesize($\approx$31\%) & 743 \color{c99}\footnotesize($\approx$99\%) \\
	8 & 1169 & 956 & 64 \color{c7}\footnotesize($\approx$7\%) & 938 \color{c98}\footnotesize($\approx$98\%) \\
	9 & 1291 & 631 & 4 \color{c0}\footnotesize($<$1\%) & 606 \color{c96}\footnotesize($\approx$96\%) \\
	Total & 3896 & 2961 & 755 \color{c25}\footnotesize($\approx$25\%) & 2914 \color{c98}\footnotesize($\approx$98\%)
	\end{tabular}
	}
\end{center}
\caption{Number of total, rationally identifiable, LF-HTC identifiable, and Algorithm \ref{algo_combined_identification} identifiable unlabeled DAGs (up to topological equivalence) for the latent structure presented in Figure \ref{fig_3.5}a. The percentages are relative to the number of rationally identifiable latent-factor graphs.}
\label{table1}
\end{table}
\begin{table}
	\begin{center}
	\resizebox{0.7\textwidth}{!}{%
		\begin{tabular}{ ccccc }
		\hline
		 $|D_V|$ & Total & Rationally &  LF-HTC  & Det + eLF-HTC + rec\\
		\hline
	0 & 1 & 1 & 1 \color{c100}\footnotesize(100\%)& 1 \color{c100}\footnotesize(100\%)\\
	1 & 8 & 6 & 6 \color{c100}\footnotesize(100\%) & 6 \color{c100}\footnotesize(100\%)\\
	2 & 63 & 45 & 43 \color{c96}\footnotesize($\approx$96\%)& 45 \color{c100}\footnotesize(100\%)\\
	3 & 391 & 255 & 236 \color{c93}\footnotesize($\approx$93\%)& 255 \color{c100}\footnotesize(100\%)\\
	4 & 1983 & 1171 & 1018 \color{c87}\footnotesize($\approx$87\%)& 1168 \color{c100}\footnotesize($\approx$100\%)\\
	5 & 7570 & 3898 & 3028 \color{c78}\footnotesize($\approx$78\%)& 3850 \color{c99}\footnotesize($\approx$99\%)\\
	6 & 21,029 & 8960 & 5861 \color{c65}\footnotesize($\approx$65\%) & 8675 \color{c97}\footnotesize($\approx$97\%)\\
	Total & 31,045 & 14,336 & 10,193 \color{c71}\footnotesize($\approx$71\%) & 14,000 \color{c98}\footnotesize($\approx$98\%)\\
	\end{tabular}
	}
	\end{center}
	\caption{Number of total, rationally identifiable, LF-HTC identifiable, and Algorithm \ref{algo_combined_identification} identifiable DAGs (up to topological equivalence) for the latent structure presented in Figure \ref{fig_3.5}b. The percentages are relative to the number of rationally identifiable latent-factor graphs.}
	\label{table2}
\end{table}

\section{Conclusion}
\label{sec_conclusion}
In this paper, we have provided a new identification criterion that strictly improves all existing criteria that are known to us. As we have shown, combining all these extensions into one algorithm yields significantly more latent-factor graphs that one can identify. For the considered latent-factor graphs, a combination of our presented graphical criteria is able to rationally identify nearly all rationally identifiable graphs. To the best of our knowledge, our results are the most general graphical results for latent-factor graphs. From a computational point of view, these graphical results also  scale better than classical computer-algebra-based methods, and under some simplifying assumptions we even achieve polynomial time complexity. Besides, a practitioner with a latent-factor graph at hand might also directly apply our provided implementation or implement the provided pseudo-code themselves. 

In the future, it will be interesting to find (rational) formulas for identifying the remaining rationally identifiable graphs (which is typically not possible thus far with computer algebra methods due to the very long runtimes). Furthermore, one could combine our presented ideas with the recent ideas from \citet{sturma2025trek} who provide identification results for more general latent variables structures where the latent variables are not restricted to be source nodes in the graph. 
One could also derive more elaborate necessary results than the rather simple necessary results already derived in \citet{barber2022half}.
Besides, one could study finite-sample properties of different valid estimators for a given direct causal effect and study some notion of ``optimal estimation" (similar in spirit to \citet{runge2021necessary} or \citet{henckel2022graphical, henckel2024graphical}).

\section*{Acknowledgements and Disclosure of Funding}
T.H.\ is also affiliated with Technische Universität Berlin as a registered PhD-student. J.R.\ received funding from the European Research Council (ERC) Starting Grant CausalEarth
under the European Union’s Horizon 2020 research and
innovation program (Grant Agreement No.\ 948112). M.D. has received funding from the European Research Council (ERC) under the European Union’s Horizon 2020 research and innovation programme (grant agreement No.\ 883818).


\newpage

\appendix

\section{Appendix to Section \ref{sec_preliminaries}}
\label{app_id_notions}
In Section \ref{sec_terminologies_implications}, we introduce the terminology from \citet{barber2022half} of
rational identifiability and compare it to our Definition~\ref{def:rat-id}. 
In Section \ref{app_further_exp}, we then give more detailed explanations on why the Zariski closures of the involved parameter spaces are irreducible.

\subsection{Comparison of Different Definitions of Identifiability}
\label{sec_terminologies_implications}

The notion of rational identifiability from \citet{barber2022half} does not directly operate on the parameters $(\Lambda, \Gamma,\Omega_{\textnormal{diag}}, \mathcal{V}_L)$, but rather on $(\Lambda, \Omega_{\textnormal{diag}}+\Gamma^T\mathcal{V}_L\Gamma)$.
To formalize their approach, define the map
\begin{align*}
    \tau&:\mathbb{R}^{D_{\mathcal{L}V}}\times \textnormal{diag}^+_d\times \textnormal{
    diag}^+_\ell\rightarrow \textnormal{PD}(d),\\
    &(\Gamma,\Omega_{\textnormal{diag}}, \mathcal{V}_L)\mapsto \Omega_{\textnormal{diag}} + \Gamma^T\mathcal{V}_L \Gamma
\end{align*}
and let $\Theta_{2022}:= \mathbb{R}^{D_V}_{\textnormal{reg}}\times \textnormal{Im}(\tau)$, where the index 2022 indicates the publishing year. It holds that $\overline{\Theta_{2022}}$ is irreducible because $\Theta_{2022}$ is the polynomial image of an open set; 
for a more detailed explanation of this fact see Section \ref{app_further_exp}.
Note that \citet{barber2022half} assume without loss of generality that $\mathcal{V}_L=I_\ell$. However, their identification results also apply to the slightly more general setting considering arbitrary $\mathcal{V}_L$, as the image of the map $\tau$ does not change. 
Next, define the map
\begin{align*}
    \phi_{G^{\mathcal{L}}}&:\Theta_{2022} \longrightarrow \textnormal{PD}(d),\\
    &(\Lambda, \Omega)\mapsto (I_d-\Lambda)^{-T}\Omega(I_d-\Lambda)^{-1}.
\end{align*}
The definition of rational identifiability from \citet{barber2022half} is now as follows. Since the definition operates directly on $\Theta_{2022}$, we use the term  \emph{$\Theta_{2022}$-rational identifiability} to differentiate from Definition~\ref{def:rat-id}.

\begin{definition}
\mbox{ }
\begin{enumerate}[label=(\alph*)]
    \item     The latent-factor graph $G^\mathcal{L}$ is $\Theta_{2022}$-rationally identifiable if there exists a proper algebraic subset $A_{2022}\subsetneq \overline{\Theta_{2022}}$ and a rational map $\psi_{2022}:\textnormal{PD}(d)\rightarrow \mathbb{R}^{D_V}_{\textnormal{reg}}\times \textnormal{PD}(d)$ such that $(\psi_{2022} \circ \phi_{G^\mathcal{L}})(\Lambda,\Omega)=(\Lambda, \Omega)$ for all $(\Lambda,\Omega)\in \Theta_{2022}\setminus A_{2022}$.
    \item The direct causal effect $\lambda_{vw}$, or also simply the edge $v\rightarrow w\in D_V$, or if the vertex $w$ is clear from the context simply the parent $v$, is $\Theta_{2022}$-rationally identifiable if there exists a proper algebraic subset $A_{2022}\subsetneq \overline{\Theta_{2022}}$ and a rational map $\psi_{2022}:\textnormal{PD}(d)\rightarrow \mathbb{R}$ such that ($\psi_{2022} \circ\phi_{G^\mathcal{L}})(\Lambda, \Omega)=\lambda_{vw}$ for all $(\Lambda,\Omega)\in \Theta_{2022}\setminus A_{2022}$.
\end{enumerate}
\end{definition}

Let $\textnormal{diag}_d$ and $\textnormal{diag}_\ell$ denote the set of all $d\times d$- and $\ell\times\ell$-diagonal matrices, respectively. To state our result relating the two definitions of rational identifiability, we need to extend $\tau$ to a mapping from $\mathbb{R}^m\rightarrow \mathbb{R}^n$ where $m:= |D_{\mathcal{L}V}|+d+\ell$ and where $n:=d(d+1)/2$ (and where we slightly abuse notation and identify $\mathbb{R}^{D_{\mathcal{L}V}}\times \textnormal{diag}_d\times \textnormal{
    diag}_\ell$ with $\mathbb{R}^m$ and $\textnormal{PD}(d)$ as a subset of $\mathbb{R}^n$). We denote this extension by $\overline{\tau}$ and define it in a trivial way by $\overline{\tau}: \mathbb{R}^m\rightarrow \mathbb{R}^n, \;(\Gamma,\Omega_{\textnormal{diag}},\mathcal{V}_L)\mapsto \Omega_{\textnormal{diag}} + \Gamma^T \mathcal{V}_L\Gamma$. Note that $\overline{\tau}$ restricted to $\Theta_{2022}$ equals $\tau$. We also use the function $(\textnormal{id},\tau):\Theta\rightarrow \Theta_{2022}$, for which the first $|D_V|$-components are the identity and the last $|D_{\mathcal{L}V}|+d+\ell$-components are given by $\tau$. The function $(\textnormal{id},\tau)$ is important as $\chi_{G^\mathcal{L}}(\Lambda,\Gamma,\Omega_{\textnormal{diag}},\mathcal{V}_L)=(\phi_{G^\mathcal{L}}\circ(\textnormal{id},\tau))(\Lambda, \Gamma,\Omega_{\textnormal{diag}},\mathcal{V}_L)$ for all $(\Lambda,\Gamma,\Omega_{\textnormal{diag}},\mathcal{V}_L)\in \Theta$. Finally, we consider the extension
$ \overline{(\textnormal{id},\tau)}:= (\textnormal{id},\overline{\tau})$.

\begin{proposition}
\label{proposition_notions_identifiability}

Let $G^\mathcal{L}= (V\cup \mathcal{L},D)$ be a latent-factor graph. If $G^\mathcal{L}$ is $\Theta_{2022}$-rationally identifiable, then it is rationally identifiable. In particular, if $A_{2022} \subseteq \overline{\Theta_{2022}}$ is a proper algebraic subset, then $A_\Theta:=\overline{(\textnormal{id},\tau)}^{-1}(A_{2022}) \subseteq \overline{\Theta}$ is a proper algebraic subset. 

Vice versa, if $G^\mathcal{L}$ is rationally identifiable, then it is $\Theta_{2022}$-rationally identifiable if the proper algebraic subset $A_\Theta \subseteq \overline{\Theta}$ is given by  $A_\Theta=   \{(\Lambda, \Gamma, \Omega_{\textnormal{diag}},\mathcal{V}_L)\in \overline{\Theta}:\;f(\Lambda,\allowbreak \overline{\tau}(\Gamma, \Omega_{\textnormal{diag}},\allowbreak\mathcal{V}_L))=0\textnormal{ for all } f \in S\}$, where $S$ is a finite set of polynomials in the entries of $\Lambda$ and $\Omega$.

\end{proposition}

\begin{proof}
First note that from a standard topological fact about continuous functions (e.g., Theorem 2.9 in \citealp{armstrong1983basic}), it follows that
\begin{align*}
\label{subset_relation_image}
    \overline{(\textnormal{id},\tau)}(\overline{\Theta})\subseteq \overline{\overline{(\textnormal{id},\tau)}(\Theta)}=\overline{(\textnormal{id},\tau)(\Theta)}=\overline{\Theta_{2022}}.\numberthis
\end{align*} 
Also, note that $\overline{(\textnormal{id},\tau)}$ restricted to $\Theta$ equals $(\textnormal{id},\tau)$.

Now, to show the first direction, suppose that $A_{2022} \subseteq \overline{\Theta_{2022}}$ is a proper algebraic subset.

Define $A_\Theta:=\overline{(\textnormal{id},\tau)}^{-1}(A_{2022})$ and observe that $A_\Theta\subseteq \overline{\Theta}$ due to \eqref{subset_relation_image}. 
Since preimages of algebraic sets under polynomial maps are again algebraic sets, it 
follows that $A_\Theta$ is an \emph{algebraic} subset of $\overline{\Theta}$.

It remains to show the properness of $A_\Theta$. First, suppose that $\Theta_{2022}\subseteq A_{2022}$. Then, $\overline{\Theta_{2022}}\subseteq \overline{A_{2022}}=A_{2022}\subsetneq \overline{\Theta_{2022}}$, which is a contradiction, and thus, $\Theta_{2022} \not\subseteq A_{2022}$, which implies that there exists
$(\Lambda,\Omega)\in \Theta_{2022} \setminus A_{2022}$. For this $(\Lambda,\Omega)$ it holds that 
$\overline{(\textnormal{id},\tau)}^{-1}(\Lambda, \Omega)=(\textnormal{id},\tau)^{-1}(\Lambda, \Omega)\neq \emptyset$, thus we can take $(\Lambda, \Gamma, \Omega_{\textnormal{diag}},\mathcal{V}_L)\in \overline{(\textnormal{id},\tau)}^{-1}(\Lambda, \Omega)\subseteq \overline{(\textnormal{id},\tau)}^{-1}(\Theta_{2022})= (\textnormal{id},\tau)^{-1}(\Theta_{2022})=\Theta$. Observe that if $(\Lambda, \Gamma, \Omega_{\textnormal{diag}},\mathcal{V}_L)$ were in $A_{\Theta}$, then, by the definition of $A_{\Theta}$, the tuple $(\Lambda, \Omega)$ would be in $A_{2022}$, which is a contradiction. Thus, we conclude that 
$(\Lambda, \Gamma, \Omega_{\textnormal{diag}},\mathcal{V}_L)\notin A_\Theta$, which implies that $A_\Theta$ is a \emph{proper} algebraic subset of $\overline{\Theta}$.

If $G^\mathcal{L}$ is $\Theta_{2022}$-rationally identifiable, then there exists a proper algebraic subset $A_{2022}\subsetneq \overline{\Theta_{2022}}$ and a rational map $\psi_{2022}:\textnormal{PD}(d)\rightarrow \mathbb{R}^{D_V}_{\textnormal{reg}}\times \textnormal{PD}(d)$ such that $(\psi_{2022} \circ \phi_{G^\mathcal{L}})(\Lambda,\Omega)=(\Lambda, \Omega)$ for all $(\Lambda,\Omega)\in \Theta_{2022}\setminus A_{2022}$.
Define $\psi$ by $\psi:=\pi\circ\psi_{2022}$ where $\pi:\mathbb{R}^{D_V}_{\textnormal{reg}}\times \textnormal{PD}(d)\rightarrow \mathbb{R}^{D_V}_{\textnormal{reg}}$ is the projection onto $\mathbb{R}^{D_V}_{\textnormal{reg}}$. Also, recall that $\chi_{G^\mathcal{L}}(\Lambda,\Gamma,\Omega_{\textnormal{diag}},\mathcal{V}_L)=(\phi_{G^\mathcal{L}}\circ(\textnormal{id},\tau))(\Lambda, \Gamma,\Omega_{\textnormal{diag}},\mathcal{V}_L)$ for all $(\Lambda,\Gamma,\Omega_{\textnormal{diag}},\mathcal{V}_L)\in \Theta$. Since we have already seen that $A_\Theta:=\overline{(\textnormal{id},\tau)}^{-1}(A_{2022})$ is a proper algebraic subset of $\overline{\Theta}$, we conclude that, for all $(\Lambda, \Gamma, \Omega_{\textnormal{diag}},\mathcal{V}_L)\in \Theta\setminus A_\Theta$,  it holds that
\begin{align*}
    (\psi\circ \chi_{G^\mathcal{L}})(\Lambda,\Gamma,\Omega_{\textnormal{diag}},\mathcal{V}_L) &=(\pi\circ\psi_{2022}\circ\phi_{G^\mathcal{L}}\circ(\textnormal{id},\tau))(\Lambda, \Gamma,\Omega_{\textnormal{diag}},\mathcal{V}_L)\\
    &=(\pi\circ\psi_{2022}\circ \phi_{G^\mathcal{L}})(\Lambda,\allowbreak \tau(\Gamma,\allowbreak\Omega_{\textnormal{diag}},\allowbreak\mathcal{V}_L)) \\
    &=\pi(\Lambda,\allowbreak \tau(\Gamma,\allowbreak\Omega_{\textnormal{diag}},\allowbreak\mathcal{V}_L))=\Lambda.
\end{align*}

For the other direction of the proposition, suppose first that $A_\Theta \subsetneq \overline{\Theta}$ is a proper algebraic subset that is defined by
\begin{align*}
    A_\Theta     &=\{(\Lambda, \Gamma, \Omega_{\textnormal{diag}},\mathcal{V}_L)\in \overline{\Theta}:\;f(\Lambda, \overline{\tau}(\Gamma, \Omega_{\textnormal{diag}},\mathcal{V}_L))=0\textnormal{ for all } f \in S\}\\
   &= \{(\Lambda, \Gamma, \Omega_{\textnormal{diag}},\mathcal{V}_L)\in \overline{\Theta}:\;\Tilde{f}(\Lambda, \Gamma, \Omega_{\textnormal{diag}},\mathcal{V}_L)=0\textnormal{ for all } \Tilde{f} \in \Tilde{S}\}
\end{align*}
for some sets of polynomials $S$ and
$\Tilde{S}=\{f\circ \overline{(\textnormal{id},\tau)}:\;f\in S\}$. 
Note that $\Tilde{S}$ is a finite set of polynomials because $\overline{(\textnormal{id},\tau)}$ is a polynomial. Now, define
\begin{equation} \label{eq:def-of-A}
\begin{aligned}
   A_{2022}&:=\{(\Lambda, \Omega)\in\overline{\Theta_{2022}}:\;f(\Lambda, \Omega)=0 \textnormal{ for all } f \in S \} \\
   &=\{(\Lambda, \Omega)\in\mathbb{R}^{|D_V| + d(d+1)/2}:\;f(\Lambda, \Omega)=0 \textnormal{ for all } f \in S \}\cap \overline{\Theta_{2022}},
\end{aligned}
\end{equation}
where we in the last line slightly abuse notation and interpret $(\Lambda,\Omega)$ as an element of $\mathbb{R}^{|D_V| + d(d+1)/2}$.
Because $\overline{\Theta_{2022}}$ is an algebraic set by definition and because  the intersection of two algebraic sets is an algebraic set, $A_{2022}$ is an algebraic set. It remains to show the properness of $A_{2022}$. Because $A_\Theta$ is proper, it follows that there exists $(\Lambda, \Gamma, \Omega_{\textnormal{diag}}, \mathcal{V}_L)\in \overline{\Theta}\setminus A_\Theta$, that is, there exists $(\Lambda, \Gamma, \Omega_{\textnormal{diag}}, \mathcal{V}_L)\in \overline{\Theta}$ and an $\Tilde{f}\in \Tilde{S}$ such that $\Tilde{f}(\Lambda, \Gamma, \Omega_{\textnormal{diag}}, \mathcal{V}_L)\neq 0$. By construction of $\Tilde{S}$, it now follows that $f(\Lambda, \overline{\tau}(\Gamma, \Omega_{\textnormal{diag}}, \mathcal{V}_L))\neq 0$ for at least one $f\in S$, and thus, it follows that $(\Lambda, \overline{\tau}(\Gamma, \Omega_{\textnormal{diag}}, \mathcal{V}_L))\notin A_{2022}$. Since $(\Lambda, \Gamma, \Omega_{\textnormal{diag}}, \mathcal{V}_L)\in \overline{\Theta}$ and because of Equation~\eqref{subset_relation_image} it also holds that $(\Lambda, \overline{\tau}(\Gamma, \Omega_{\textnormal{diag}}, \mathcal{V}_L))\in \overline{\Theta_{2022}}$, we have that $A_{2022}$ is a proper subset of $\overline{\Theta_{2022}}$.

Now, suppose that $G^\mathcal{L}$ is rationally identifiable, that is, there exists a proper algebraic subset $A_\Theta\subsetneq \overline{\Theta}$ and a rational map $\psi:\textnormal{PD}(d)\rightarrow \mathbb{R}^{D_V}_{\textnormal{reg}}$ such that $(\psi\circ \chi_{G^\mathcal{L}})(\Lambda,\Gamma,\Omega_{\textnormal{diag}},\mathcal{V}_L)=\Lambda$ for all $(\Lambda,\Gamma, \Omega_{\textnormal{diag}},\mathcal{V}_L)\in \Theta\setminus A_\Theta$. Moreover, suppose that $A_\Theta$ is of the required form, which implies that $A_{2022} \subseteq \overline{\Theta_{2022}}$ defined as in~\eqref{eq:def-of-A} is a proper algebraic subset. Consider a parameter pair $(\Lambda, \Omega)\in \Theta_{2022}\setminus A_{2022}$ and a corresponding triple $(\Gamma, \Omega_{\textnormal{diag}}, \mathcal{V}_L)\in \tau^{-1}(\Omega)$ in the preimage. 
Since $(\Lambda, \Gamma, \Omega_{\textnormal{diag}}, \mathcal{V}_L)\in \Theta\setminus A_\Theta$, it holds that 
 \begin{align*}
     (\psi\circ \phi_{G^\mathcal{L}})(\Lambda, \Omega) &= (\psi \circ \phi_{G^\mathcal{L}} \circ (\textnormal{id},\tau))(\Lambda, \Gamma,\Omega_{\textnormal{diag}}, \mathcal{V}_L)\\
     &=(\psi \circ \chi_{G^\mathcal{L}})(\Lambda, \Gamma,\Omega_{\textnormal{diag}}, \mathcal{V}_L)\\
     &= \Lambda.
 \end{align*}
 Setting $\psi_{2022}(\Sigma):=(\psi(\Sigma), (I_d-\psi(\Sigma))^{T}	\Sigma (I_d-\psi(\Sigma)))$, we conclude that 
 \begin{align*}
   &(\psi_{2022}\circ \phi_{G^\mathcal{L}})(\Lambda, \Omega)\\
   &= (\psi( \phi_{G^\mathcal{L}}(\Lambda, \Omega)), (I_d-\psi(\phi_{G^\mathcal{L}}(\Lambda, \Omega)))^{T}	\phi_{G^\mathcal{L}}(\Lambda, \Omega) (I_d-\psi(\phi_{G^\mathcal{L}}(\Lambda, \Omega))))\\
   &=(\Lambda, (I_d-\Lambda)^{T}	\phi_{G^\mathcal{L}}(\Lambda, \Omega) (I_d-\Lambda))\\
   &= (\Lambda, \Omega),
 \end{align*}
 where we used Equation~\eqref{eq_Sigma} for the second to last equality.
\end{proof}

We now state a widely applicable condition under which both directions of Proposition \ref{proposition_notions_identifiability} hold.
\begin{corollary} \label{cor_implications}
Let $G^\mathcal{L}= (V\cup \mathcal{L},D)$ be a latent-factor graph. Suppose that $G^\mathcal{L}$ is rationally identifiable and the proper algebraic subset $A_{\Theta}$ 
can be rewritten in terms of polynomials of the observed covariance matrix $\Sigma(\Lambda, \Gamma, \Omega_{\textnormal{diag}},\mathcal{V}_L) := (I_d-\Lambda)^{-T}\overline{\tau}(\Gamma,\Omega_{\textnormal{diag}},\mathcal{V}_L)(I_d-\Lambda)^{-1}$, that is,   
\[
    A_\Theta:=\{(\Lambda, \Gamma, \Omega_{\textnormal{diag}},\mathcal{V}_L)\in \overline{\Theta}:\;f(\Sigma(\Lambda, \Gamma, \Omega_{\textnormal{diag}},\mathcal{V}_L))=0\textnormal{ for all } f \in S\}
\] 
for some set of polynomials $S\subseteq \mathbb{R}[x_1,\ldots,x_{d(d+1)/2}]$ (where we again slightly abuse notation and identify $\Sigma(\Lambda, \Gamma, \Omega_{\textnormal{diag}},\mathcal{V}_L)$ as an element of $\mathbb{R}^{d(d+1)/2}$). Then, $G^\mathcal{L}$ is $\Theta_{2022}$-rationally identifiable. 
\end{corollary}
\begin{proof}
Suppose that $G^\mathcal{L}$ is rationally identifiable, that is, there exists a proper algebraic subset $A_\Theta\subsetneq \overline{\Theta}$ and a rational map $\psi:\textnormal{PD}(d)\rightarrow \mathbb{R}^{D_V}_{\textnormal{reg}}$ such that $(\psi\circ \chi_{G^\mathcal{L}})(\Lambda,\Gamma,\Omega_{\textnormal{diag}},\mathcal{V}_L)=\Lambda$ for all $(\Lambda,\Gamma, \Omega_{\textnormal{diag}},\mathcal{V}_L)\in \Theta\setminus A_\Theta$. By assumption, we can rewrite $A_\Theta$ in terms of \emph{rational} functions of $\Lambda$ and $\overline{\tau}(\Gamma, \Omega_{\textnormal{diag}},\mathcal{V}_L)$, that is, as $A_\Theta=\{(\Lambda, \Gamma, \Omega_{\textnormal{diag}},\mathcal{V}_L)\in \overline{\Theta}:\;g(\Lambda, \overline{\tau}(\Gamma, \Omega_{\textnormal{diag}},\mathcal{V}_L))=0\textnormal{ for all } g \in \Tilde{S}\}$ where $\Tilde{S}:=\{g\equiv f\circ \Sigma\}$. It remains to show that we can replace the rational functions $g \in \Tilde{S}$ by polynomial functions. For every $g\in \Tilde{S}$ consider the representation $g=g_{\textnormal{num}}/g_{\textnormal{denom}}$ for some polynomials $g_{\textnormal{num}}$ and $g_{\textnormal{denom}}$, and define $S':=\{h \equiv g_{\textnormal{num}}:\; g\in \Tilde{S}\}$. Then, the set 
    \begin{align*}
         B_\Theta&:
         =\{(\Lambda, \overline{\tau}(\Gamma, \Omega_{\textnormal{diag}},\mathcal{V}_L))\in \overline{\Theta}:\;h(\Lambda, \overline{\tau}(\Gamma, \Omega_{\textnormal{diag}},\mathcal{V}_L))=0 \textnormal{ for all } h \in S'\}
    \end{align*}
is a proper algebraic subset of $\overline{\Theta}$ that contains $A_{\Theta}$. Hence, it holds that $(\psi\circ \chi_{G^\mathcal{L}})(\Lambda,\Gamma,\allowbreak\Omega_{\textnormal{diag}},\allowbreak\mathcal{V}_L)=\Lambda$ for all $(\Lambda,\Gamma, \Omega_{\textnormal{diag}},\mathcal{V}_L)\in \Theta\setminus B_\Theta$, and we conclude by   
Proposition~\ref{proposition_notions_identifiability} that $G^\mathcal{L}$ is $\Theta_{2022}$-rationally identifiable.
\end{proof}

By Corollary~\ref{cor_implications}, all discussed identification results from our main paper, that is, Theorems \ref{theorem_lfhtc}, \ref{theorem_elfhtc} and \ref{theorem_determinantal}  and Proposition \ref{proposition_id_notion_subgraph},  are also true when replacing rational identifiability with $\Theta_{2022}$-rational identifiability and vice versa. This follows directly from inspecting the proofs.

\subsection{Further Explanations on the Irreducibility of $\overline{\Theta}$ and $\overline{\Theta_{2022}}$}
\label{app_further_exp}
We first recall a standard fact about Cartesian products and Zariski closures: For two sets $X\subseteq \mathbb{R}^{n_x}$ and $Y\subseteq \mathbb{R}^{n_y}$ for some $n_x, n_y\in \{1,2,3,\ldots\}$, we have $\overline{X\times Y} = \overline{X} \times \overline{Y}$. The ``$\subseteq$" follows because $\overline{X} \times \overline{Y}$ is an algebraic set (see, for example, exercise 15d in Section 2 in \citet{cox2015ideals}) containing $X\times Y$. Vice versa for the ``$\supseteq$", write $\overline{X\times Y}=V(S)$. Then, for every $f\in S$, it holds that $f(x,y)=0$ for all $(x,y)\in X\times Y$. Now, define for all $y\in Y$ and $f\in S$ the polynomial $g_{y,f}(x):=f(x, y)$. Clearly, for all $x\in X$, $y\in Y$ and $f\in S$ it holds that $g_{y,f}(x)=0$, and hence, $\{x\in \mathbb{R}^{n_x}:\;g_{y,f}(x)=0,\;y\in Y,\;f\in S\}$ is a closed set containing $X$ which implies that $\overline{X}\subseteq \{x\in \mathbb{R}^{n_x}:\;g_{y,f}(x)=0,\;y\in Y,\;f\in S\}$ because $\overline{X}$ is the smallest closed set containing $X$ by definition. 
Thus, $g_{y,f}(x)=0$ for all $x\in \overline{X},y\in Y$ and $f\in S$, which implies $f(x, y)=0$ for all $x\in \overline{X}, y\in Y$ and $f\in S$. Analogously, defining $h_{x,f}(y):=f(x,y)$ for all $x\in \overline{X}, y\in Y$ and $f\in S$, we conclude that $f(x,y)=0$ for all $x\in \overline{X}, y\in \overline{Y}$ and $f\in S$ which implies that $(x,y)\in V(S)=\overline{X\times Y}$ for all $(x,y)\in \overline{X}\times \overline{Y}$. \\

 Now we show the irreducibility of $\overline{\Theta}$. For that, we write $\textnormal{diag}_d$ and $\textnormal{diag}_\ell$ to denote the set of all $d\times d$- and $\ell\times\ell$-diagonal matrices, respectively. As already mentioned in the main paper, it holds that $\overline{\Theta}\cong\mathbb{R}^{m_\Theta}$. This fact is true firstly because $\overline{\Theta} =\overline{\mathbb{R}^{D_V}_{\textnormal{reg}}\times \mathbb{R}^{D_{\mathcal{L}V}}\times \textnormal{diag}^+_d\times \textnormal{diag}^+_\ell}= \overline{\mathbb{R}^{D_V}_{\textnormal{reg}}}\times \overline{\mathbb{R}^{D_{\mathcal{L}V}}}\times \overline{\textnormal{diag}^+_d}\times \overline{\textnormal{diag}^+_\ell}$ from the fact about Cartesian products in the previous paragraph. Secondly, we have $\overline{\mathbb{R}^{D_V}_{\textnormal{reg}}}=\mathbb{R}^{D_V}\cong \mathbb{R}^{|D_V|}$ and $\overline{\textnormal{diag}^+_d}=\textnormal{diag}_d\cong\mathbb{R}^d$ and $\textnormal{diag}^+_\ell=\textnormal{diag}_\ell\cong\mathbb{R}^\ell$ as the Zariski closure of a non-empty open set (in the classical Euclidean topology) is the entire space because a polynomial that is zero on open balls (in the classical Euclidean topology) must be zero everywhere as the polynomial restricted to every line intersecting this ball is a one-dimensional polynomial and as every such intersection contains infinitely many points and in one dimension the Zariski closed sets are exactly all the sets with finitely many points except for the real numbers itself. We now explain (the also rather standard fact) that $\mathbb{R}^{m_\Theta}$ is irreducible (which then immediately implies that $\overline{\Theta}\cong\mathbb{R}^{m_\Theta}$ is irreducible):  Let $\mathbb{R}^{m_\Theta}\subseteq A_1\cup A_2$ for two proper algebraic subsets $A_1, A_2$ of $\mathbb{R}^{m_\Theta}$. Then, we can write $A_1=V(S_1)$ and $A_2=V(S_2)$ for some set of polynomials $S_1,S_2\subseteq \mathbb{R}[x_1,\ldots, x_{m_\Theta}]$. Hence, we can write $\mathbb{R}^{m_\Theta}=A_1\cup A_2 = V(S_1\cdot S_2)$ where $S_1S_2=\{f_1f_2:\;f_1 \in S_1,\; f_2\in S_2\}$. Therefore, $f_1f_2\equiv 0$ for all $f_1\in S_1$ and $f_2\in S_2$, which implies that $f_1\equiv 0$ for all $f_1\in S_1$ or $f_2\equiv 0$ for all $f_2\in S_2$. Therefore, either $V(S_1)=\mathbb{R}^{m_\Theta}$ or $V(S_2)=\mathbb{R}^{m_\Theta}$, contradiction. \\

Next, we show the irreducibility of $\overline{\Theta_{2022}}$ defined in Section \ref{sec_terminologies_implications}. Suppose $\overline{\Theta_{2022}}\subseteq C_1\cup C_2$, where $C_1$ and $C_2$ are (in slight abuse of notation) Zariski closed subsets of $\mathbb{R}^{m_{2022}}$, where $m_{2022}:=|D_V| + d(d+1)/2$, and where $C_1$ and $C_2$ are proper subsets of $\overline{\Theta_{2022}}$. Then, $\Theta_{2022}\subseteq \overline{\Theta_{2022}} \subseteq C_1 \cup C_2$. Now, recall the definition of $\overline{(\textnormal{id},\tau)}$ in Section \ref{sec_terminologies_implications}.
Because $\Theta_{2022} =(\textnormal{id},\tau)(\Theta)= \overline{(\textnormal{id},\tau)}(\Theta)$, it follows that $\Theta\subseteq \overline{(\textnormal{id},\tau)}^{-1}(\overline{(\textnormal{id},\tau)}(\Theta))= \overline{(\textnormal{id},\tau)}^{-1}(\Theta_{2022})\subseteq\overline{(\textnormal{id},\tau)}^{-1}(C_1\cup C_2)=\overline{(\textnormal{id},\tau)}^{-1}(C_1)\cup\overline{(\textnormal{id},\tau)}^{-1}(C_2)$. Because $\overline{(\textnormal{id},\tau)}$ is a polynomial and thus continuous with respect to the Zariski topology, it follows that $\overline{(\textnormal{id},\tau)}^{-1}(C_1)$ and $\overline{(\textnormal{id},\tau)}^{-1}(C_2)$ are also Zariski closed with respect to the Zariski topology on $\overline{\Theta}\cong\mathbb{R}^{m_\Theta}$. Therefore, $\overline{\Theta}\subseteq \overline{(\textnormal{id},\tau)}^{-1}(C_1)\cup\overline{(\textnormal{id},\tau)}^{-1}(C_2)$. As argued in the previous paragraph,
$\overline{\Theta}\cong\mathbb{R}^{m_\Theta}$ is irreducible, and thus,
without loss of generality,  $\overline{(\textnormal{id},\tau)}^{-1}(C_1) = \overline{\Theta}\cong\mathbb{R}^{m_\Theta}$. Hence, we have that $C_1=\overline{(\textnormal{id},\tau)}(\mathbb{R}^{m_\Theta})\supseteq \overline{(\textnormal{id},\tau)}(\Theta) =(\textnormal{id},\tau)(\Theta)= \Theta_{2022}$, which implies $C_1\supseteq \overline{\Theta_{2022}}$ because $C_1$ is Zariski closed by assumption, contradiction.

\section{Appendix to Section \ref{section_elfhtc}}
In Section \ref{sec_proof_elfhtc}, we provide the proof of the eLF-HTC identification result (Theorem \ref{theorem_elfhtc}) from Section \ref{section_elfhtc}. Afterwards, in Sections \ref{sec_claim2}
and \ref{sec_modified_lemma_2}, we provide statements and proofs of lemmas
occuring in Section \ref{sec_proof_elfhtc}. Then, in Section \ref{sec_extended_example_elfhtc}, we provide the rational identification formulas for Example \ref{example_elfhtc}. Finally, in Section \ref{sec_rat_two_proxies}, we provide the rational identification formula for Example \ref{ex_two_proxies}.

\subsection{Proof of the eLF-HTC Identification Result}
\label{sec_proof_elfhtc}

\begin{proof}
We first show rational identifiability with respect to the definition from \citet{barber2022half}, so with respect to $ \Theta_{2022}$, which we introduce in Section \ref{sec_terminologies_implications}.  Proposition \ref{proposition_notions_identifiability} then implies rational identifiability with respect to our notion.

We divide the proof of Theorem \ref{theorem_elfhtc} into three steps: In \textbf{Step 1}, we introduce the linear equation system which one can uniquely solve in order to obtain the coefficients $\Lambda_{p^vv}$ with $p^v\in W_v\setminus (Z^{(2)}\cup W_Z)$ that Theorem \ref{theorem_elfhtc} claims to identify. In Step 1, we also introduce some required notation and give some general remarks. In \textbf{Step 2}, we then prove that the claimed unique solution of this linear equation system in Step 1 indeed is a solution for generic $(\Lambda, \Omega)\in \Theta_{2022}$, that is, for all $(\Lambda, \Omega)\in \Theta_{2022}\setminus A_{2022}$ for some proper algebraic subset $A_{2022}\subsetneq \overline{\Theta_{2022}}$. In \textbf{Step 3}, we then prove that this solution is unique for generic $(\Lambda, \Omega)\in \Theta_{2022}$.

\noindent\underline{\textbf{Step 1 (introducing the linear equation system of interest):}}

Regarding notation: To ease notation, for a matrix $M$, we write $M_{rs}$ with $r,s\in V$ to indicate the entry of $M$ where the row corresponds to $r$ and the column corresponds to $s$, similarly for a vector $q$ we write $q_r$ to denote the component corresponding to $r$ (what ``corresponds" means will be clear from context). Similarly, we write $M_{R, S}$ with $R=\{r_1,\ldots,r_m\},S=\{s_1,\ldots,s_n\}\subseteq V$ with $m$ not necessarily equal to $n$ to indicate the submatrix of $M$ whose $(i,j)$-th entry with $i\in\{1,\ldots, m\}$ and $j\in \{1,\ldots, n\}$ is $M_{r_is_j}$, similarly for a vector $q_R$. When $R=\{r\}$, we sometimes also write $M_{r,S}$. When writing down a matrix-vector product $M_{R, S}q_S$, we implicitly assume that the $s$-columns of $M_{R,S}$ match the corresponding $s$-components of $q_S$.

Now, in order to prove Theorem \ref{theorem_elfhtc}, we construct a linear equation system of the form
	\begin{align}
	\label{eq_lin}
		d = (A\;\;\; B\;\;\;C) \cdot \begin{pmatrix}
			\alpha\\
			\beta\\
			\gamma
		\end{pmatrix},
	\end{align}
	where $\alpha$ has components 
	\begin{align}
	\label{def_alplha}
	    \alpha_{p^v}=\Lambda_{p^vv}\text{ for }p^v\in W_v\setminus (Z^{(2)}\cup W_Z)
	\end{align}
	and where we introduce the ``byproducts" $\beta$ and $\gamma$ which are not relevant for the remainder of this work in Step 2.

	  For defining the respective rows of $A$, $B$, $C$ and $d$, we distinguish between two cases:
	  
	  	\noindent \textbf{Case 1:} $y\not \in \textnormal{htr}_H(Z\cup\{v\})$: 
	In this case, define
	\begin{align*}
		d_y&:=\Sigma_{yv} - \sum_{p^v\in \textnormal{pa}_V(v)\setminus W_v}\Sigma_{yp^v}\Lambda_{p^vv},\\
		A_{yp^v}&:=\Sigma_{yp^v},\;\;\;p^v\in W_v\setminus (Z^{(2)}\cup W_Z),\\
		B_{yz}&:= \Sigma_{yz} - \sum_{p^z\in \textnormal{pa}_V(z)\setminus W_z}\Sigma_{yp^z}\Lambda_{p^zz},\;\;\;z\in Z^{(1)},\text{ and }\\
	C_{yw}&:=\Sigma_{yw},\;\;\;w\in Z^{(2)}\cup W_Z.
	\end{align*}
Note that $d_y$ and each $B_{yz}$ is rationally identifiable as all occuring covariances are between observed variables and as all occuring parameters from the matrix $\Lambda$ are rationally identifiable by the (additional) assumptions in Theorem \ref{theorem_elfhtc}. Similarly, each $A_{yp^v}$ and each $C_{yw}$ is trivially rationally identifiable from $\Sigma$.

	\noindent \textbf{Case 2:} $y \in \textnormal{htr}_H(Z\cup\{v\})$:

		In this case, define
	\begin{align*}
		d_y&:=\Sigma_{yv}-\sum_{p^y\in \textnormal{pa}_V(y)}\Lambda_{p^yy}\Sigma_{p^yv}-\sum_{p^v\in \textnormal{pa}_{V}(v)\setminus W_v}\left(\Sigma_{yp^v}-\sum_{p^y\in \textnormal{pa}_V(y)}\Lambda_{p^yy}\Sigma_{p^yp^v}\right)\Lambda_{p^vv},\\
		A_{yp^v} &:= \Sigma_{yp^v}-\sum_{p^y\in \textnormal{pa}_V(y)}\Lambda_{p^yy}\Sigma_{p^yp^v},\;\;\;p^v\in W_v\setminus (Z^{(2)}\cup W_Z),\\
		B_{yz}&:=\Sigma_{yz}-\sum_{p^y\in \textnormal{pa}_V(y)}\Lambda_{p^yy}\Sigma_{p^yz}-\sum_{p^z\in \textnormal{pa}_{V}(z)\setminus W_z}\left(\Sigma_{yp^z}-\sum_{p^y\in \textnormal{pa}_V(y)}\Lambda_{p^yy}\Sigma_{p^yp^z}\right)\Lambda_{p^zz},\;\;z\in Z^{(1)},\\
		C_{yw}&:=	\Sigma_{yw}-\sum_{p^y\in \textnormal{pa}_V(y)}\Lambda_{p^yy}\Sigma_{p^yw},\;\;\;w\in Z^{(2)}\cup W_Z.
	\end{align*}
 Note that again every $d_y$, $A_{yp^v}$, $B_{yz}$ and $C_{yw}$ is rationally identifiable as all occuring covariances are between observed variables and as all occuring parameters are rationally identifiable by assumption.

	\noindent\underline{ \textbf{Step 2 (proposed solution is indeed a solution):}}
	
Write $\textnormal{pa}_V(v)=\{p^v_1,\allowbreak\ldots,\allowbreak p^v_n\}$ and $Z=\{z_1,\allowbreak\ldots,\allowbreak z_r\}$ and $Y=\{y_1,\ldots,y_{k}\}$ where $n:=|\textnormal{pa}_V(v)|$ and $r:=|H|$ and $k:=|W_v\cup W_Z\cup Z|$. Define $n_1:=|W_v\setminus (Z^{(2)}\cup W_Z)|$ and without loss of generality, assume that $p^v_1,\ldots,p^v_{n_1}\in W_v\setminus (Z^{(2)}\cup W_Z)$ and that $p^v_{n_1+1},\ldots,p^v_n\in \textnormal{pa}_V(v)\setminus (W_v\setminus (Z^{(2)}\cup W_Z))$. Define $r_1:=|Z^{(1)}|$ and also assume without loss of generality that $z_1,z_2\ldots,z_{r_1}\in Z^{(1)}$ and $z_{r_1+1},\ldots, z_r\in Z^{(2)}$.
Now, define the matrix $A'\in \mathbb{R}^{k\times n}$ and $B'\in \mathbb{R}^{k\times r}$ and a vector $c'\in \mathbb{R}^{k}$ by
\begin{align*}
    A'_{yp^v} := \begin{cases}
        [(I_d-\Lambda)^T\Sigma]_{yp^v} &\text{ if } y\in \textnormal{htr}_H(Z\cup\{v\}),\\
        \Sigma_{yp^v} &\text{ if } y\not\in \textnormal{htr}_H(Z\cup\{v\})
    \end{cases}
\end{align*}
for all $y\in Y$ and $p^v\in \textnormal{pa}_V(v)$, and
\begin{align*}
    B'_{yz}:= \begin{cases}
        [(I_d-\Lambda)^T\Sigma(I_d-\Lambda)]_{yz} &\text{ if } y\in \textnormal{htr}_H(Z\cup\{v\}),\\
        [\Sigma(I_d-\Lambda)]_{yz} &\text{ if } y\not\in \textnormal{htr}_H(Z\cup\{v\})
    \end{cases}
\end{align*}
for all $y\in Y$ and $z\in Z$,	and
	\begin{align*}
	    c'_y:= \begin{cases}
        [(I_d-\Lambda)^T\Sigma]_{yv} &\text{ if } y\in \textnormal{htr}_H(Z\cup\{v\}),\\
        \Sigma_{yv} &\text{ if } y\not\in \textnormal{htr}_H(Z\cup\{v\})
    \end{cases}
	\end{align*}
	for all $y\in Y$.
Next, similarly as in the proof of the LF-HTC in \citet{barber2022half}, note that from condition (iii) in the eLF-HTC it follows that there exists $Y_Z\subseteq Y$ such that there is a system of latent-factor half-treks with no-sided intersection from $Y_Z$ to $Z$. Every latent-factor half-trek in this system of treks thereby has the form $y\leftarrow h\rightarrow z$ for $y\in Y$, $z\in Z$ and $h\in H$. From Proposition 3.4 in \citet{sullivant2010trek}, it now follows that $\det(\Omega_{Y_Z, Z})\neq 0$ for generic $(\Lambda, \Omega)\in \Theta_{2022}$ and hence that $\Omega_{Y, Z}$ has full column rank for generic $(\Lambda, \Omega)\in \Theta_{2022}$.

We now invoke Claim 2 from \citet{barber2022half} which we more formally state in Section \ref{sec_claim2}  and which implies that there exists $\psi \in \mathbb{R}^r$ for generic $(\Lambda, \Omega)\in \Theta_{2022}$ such that
\begin{align*}
\label{eq_claim_2}
    \begin{pmatrix}
        A' & B'
    \end{pmatrix} \cdot
    \begin{pmatrix}
        \Lambda_{\textnormal{pa}_V(v) ,v}\\
        \psi
    \end{pmatrix} = c'.\numberthis
\end{align*}
Now, note that $A'_{yp^v}=A_{yp^v}$ for all $y\in Y$ and $p^v\in W_v\setminus (Z^{(2)}\cup W_Z)$. Similarly, for all $y\in Y$ and $z\in Z$,
\begin{align*}
  B'_{yz}= \begin{cases}
  B_{yz}-\sum_{w\in W_{z}}C_{yw}\Lambda_{wz}  & \text{if } z\in Z^{(1)} \\
   C_{yz}-\sum_{w\in W_{z}}C_{yw}\Lambda_{wz} & \text{if } z\in Z^{(2)}.
  \end{cases}
\end{align*}
Therefore, equation \eqref{eq_claim_2} implies that  
\begin{align*}
\label{eq_claim_2_1}
    &\begin{pmatrix}
        A  & A'_{Y,\textnormal{pa}_V(v)\setminus (W_v\setminus (Z^{(2)}\cup W_Z))}& (B_{Yz} - \sum_{w\in W_z}C_{Yw}\Lambda_{wz})_{z\in Z^{(1)}}& (C_{Yz} - \sum_{w\in W_z}C_{Yw}\Lambda_{wz})_{z\in Z^{(2)}} 
    \end{pmatrix}\\
    &\cdot
    \begin{pmatrix}
        \Lambda_{W_v\setminus (Z^{(2)}\cup W_Z), v}\\
        \Lambda_{\textnormal{pa}_V(v)\setminus(W_v\setminus (Z^{(2)}\cup W_Z)), v}\\
        \psi_{Z^{(1)}}\\
        \psi_{Z^{(2)}}
    \end{pmatrix} = c'.\numberthis
\end{align*}

Now, let $\alpha$ be as in equation \eqref{def_alplha} and define $\beta := \psi_{Z^{(1)}}$ and $\gamma$ for all $w\in Z^{(2)}\cup W_Z$ by $\gamma_w:=\psi_w1_{w\in Z^{(2)}}+\Lambda_{w,v}1_{w\in (\textnormal{pa}_V(v)\setminus (W_v\setminus (Z^{(2)}\cup W_z)))}-\sum_{z\in Z:w\in W_z}\Lambda_{wz}\psi_z$. Note that (in particular) for $w\in (\textnormal{pa}_V(v)\setminus (W_v\setminus (Z^{(2)}\cup W_z)))\cap (Z^{(2)}\cup W_z)$ it holds that $A'_{yw}=C_{yw}$. Also note that $(\textnormal{pa}_V(v)\setminus (W_v\setminus (Z^{(2)}\cup W_Z)))\setminus (Z^{(2)}\cup W_Z)=\textnormal{pa}_V(v)\setminus W_v$. Besides, recall that $Z^{(1)}\cap (W_Z\cup W_v\cup Z^{(2)})=\emptyset$ by condition (i) in the eLF-HTC and by definition. Using these three facts, and rearranging \eqref{eq_claim_2_1} then yields
\begin{align*}
\label{eq_claim_2_2}
    \begin{pmatrix}
        A  & B & C 
    \end{pmatrix}   \cdot
    \begin{pmatrix}
        \alpha\\
        \beta\\
        \gamma 
    \end{pmatrix}
    = c'-       A'_{Y,\textnormal{pa}_V(v)\setminus W_v}\cdot  \Lambda_{W_v\setminus \textnormal{pa}_V(v), v}=d.\numberthis
\end{align*}

\noindent\underline{\textbf{Step 3 (unique solvability):}}

	Similarly to \citet{barber2022half}, we will set some parameters equal to zero and then show that the considered matrix $(A\;\;\;B\;\;\;C)$ is invertible for generic choices of the remaining free parameters of $(\Lambda, \Omega)\in \Theta_{2022}$. Doing so then implies that the considered matrix $(A\;\;\;B\;\;\;C)$ is invertible for generic choices of all parameters of $(\Lambda, \Omega)\in \Theta_{2022}$ (see Proposition \ref{proposition_id_notion_subgraph} in Section \ref{sec_recursive} for a formal proof of (a modification of) this fact).
	
	In particular, we set $\Lambda_{p^zz}=0$ for all $(p^z,z)$ with $z\in Z^{(1)}$ and  $p^z\in \textnormal{pa}_V(z)\setminus W_z$.
	Under this modification and for both Case 1 and Case 2, the $A$- and $C$-matrices do not change.
	However, for the $B$-matrix we get in Case 1 $B_{yz}=\Sigma_{yz}$ with $z\in Z^{(1)}$ and in Case 2 $B_{yz}=\Sigma_{yz}-\sum_{p^y\in \textnormal{pa}_V(y)}\Lambda_{p^yy}\Sigma_{p^yz}$ with $z\in Z^{(1)}$.

	By definition, the matrices $A$ and $C$ do not share any matrix-entries (where the entries are interpreted as indeterminates). Under the previously described modification, by assumption (i) in the eLF-HTC implying that $Z^{(1)}\cap (W_Z\cup W_v)=\emptyset$ and also the trivial fact
	$Z^{(1)}\cap Z^{(2)} = \emptyset$ which follows by definition, it follows that also the $B$-matrix does not share any matrix entry with either the $A$- or $C$-matrix (where again the entries are interpreted as indeterminates).
	
	Note that under the modification from the beginning of Step 3, if $Z^{(1)}\cap (W_Z\cup W_v)\neq \emptyset$ from assumption (i) in the eLF-HTC was violated, then either the $B$- and $C$-matrix or the $A$- and $B$-matrix would share an entire column.

We now define a new graph $\hat{G}$ as follows: Start with the induced subgraph $\hat{G}=(V,D_V)$ of the original graph $G^\mathcal{L}$ to the observed vertices $V$. Then, remove all edges $p^z\rightarrow z$ in $\hat{G}$ with $z\in Z^{(1)}$ and $p^z\in \textnormal{pa}_V(z)\setminus W_z$. Next, whenever two observed vertices  $w_1,w_2\in V$ are confounded by some $h\in \mathcal{L}$ in $G^\mathcal{L}$, so whenever $\Omega_{w_1w_2}\neq 0$ (where $\Omega_{w_1w_2}$ is again interpreted as an indeterminant), add the bidirected edge $w_1\leftrightarrow w_2$ to $\hat{G}$.

 Now, define $\hat{\Lambda}$  by starting with $\Lambda$ and then setting $\Lambda_{p^zz}=0$ for all $(p^z,z)$ with $z\in Z^{(1)}$ and  $p^z\in \textnormal{pa}_V(z)\setminus W_z$, so exactly as explained at the beginning of Step 3, and define $\hat{\Omega}:=\Omega$. 
 Note that the LSEM in the latent-projection framework
given by the direct causal effect matrix $\hat{\Lambda}$ and the covariance matrix for the dependent error terms between the observed variables $\hat{\Omega}$ induces the mixed graph $\hat{G}$ by construction. Also note that $\Sigma =(I_d-\hat{\Lambda})^{-T}\hat{\Omega}(I_d-\hat{\Lambda})^{-1}$ under this modification (which follows from the same reasoning as for equation \eqref{eq_Sigma} in Section \ref{sec_intro}).

Also, by definition of the new graph $\hat{G}$, the selected system of latent-factor half-treks from $Y$ to $W_v\cup W_Z\cup Z$ induces a system of half-treks in $\hat{G}$. (A half-trek in a mixed graph either is a directed path or a bidirected edge followed by a directed path, see \citet{foygel2012half} for a formal definition).
In this system of half-treks in $\hat{G}$, any latent-factor half-trek in $G^\mathcal{L}$ that begins with $y\leftarrow h\rightarrow w$ has these two edges replaced in its corresponding half-trek in $\hat{G}$ with the bidirected edge $y\leftrightarrow w$. 
Because the left-side of a latent-factor half-trek consists of at most an edge linking latent variables and observed variables, none of the latent-factor half-treks in $G^\mathcal{L}$ ending in $(W_v\cup W_Z)\setminus Z$ uses one of the to-be-deleted edges into $Z^{(1)}$ as these to-be-deleted edges must be on the right side of the respective latent-factor half-treks and hence, there would be a sided-intersection with the latent-factor half-treks ending in $Z$.  Furthermore, by assumption, the latent-factor half-treks ending in $Z$ also do not use such deleted edges into $Z^{(1)}$ as they have the form $y\leftarrow h\rightarrow z$ for $y\in Y$ and $z\in Z$ in $G^\mathcal{L}$ and thus the form $y\leftrightarrow z$ in $\hat{G}$. Therefore, this construction is possible, and by this construction, the resulting system of half-treks in $\hat{G}$ has also no sided intersection.

Under this modification, we can now rewrite the row of $(A\;\;\; B\;\;\;C)$ indexed by $y\in Y$ as
\begin{align*}
    &(A\;\;\; B\;\;\;C)_{y}\\
    &=
    \begin{cases}
        \Sigma_{y,W_v\setminus (Z^{(2)}\cup W_Z)}, \Sigma_{y,Z^{(1)}},\Sigma_{y,Z^{(2)}\cup W_Z}\hspace{4.8cm}  \text{ if }y\notin \textnormal{htr}_H(Z\cup \{v\})\\
        [(I_d-\hat{\Lambda})^T\Sigma]_{y,W_v\setminus (Z^{(2)}\cup W_Z)}, [(I_d-\hat{\Lambda})^T\Sigma]_{y,Z^{(1)}}, [(I_d-\hat{\Lambda})^T\Sigma]_{y,Z^{(2)}\cup W_Z} \hspace{2cm} \text{ else}.
    \end{cases}
\end{align*}

As argued in the previous paragraph, there is a system of half-treks in $\hat{G}$ from $Y$ to $W_v\cup Z^{(1)}\cup Z^{(2)}\cup W_Z$ with no sided intersection. We can now apply a modification of Lemma 2 from \citet{foygel2012half} (which we state as Lemma \ref{lemma_modified_lemma_2} in Section \ref{sec_modified_lemma_2}) to $(A\;\;\;B\;\;\;C)$ to get invertibility of $(A\;\;\; B\;\;\;C)$ for generic $(\Lambda, \Omega)\in \Theta_{2022}$.

\end{proof}
\subsection{Statement of Claim 2 from \citet{barber2022half}}
\label{sec_claim2}
\begin{proposition}[Claim 2 from \citet{barber2022half}]
Consider the latent-factor graph setting as introduced in Sections \ref{sec_intro} and \ref{sec_preliminaries}. Furthermore, assume condition (ii) from the eLF-HTC (or equivalently condition (ii) from the LF-HTC) and $|Z|=|H|$ and that $\Omega_{Y,Z}$ has full column rank. Then, there exists $\psi \in \mathbb{R}^r$ such that equation \eqref{eq_claim_2} in Section \ref{sec_proof_elfhtc} is true for generic $(\Lambda, \Omega)\in \Theta_{2022}$, that is, for all $(\Lambda, \Omega)\in \Theta_{2022}\setminus A_{2022}$ where $A_{2022}$ is a proper algebraic subset of $\overline{\Theta_{2022}}$.
\end{proposition}
\begin{proof}

In the original proof of \citet{barber2022half}, replace the definition of $\Omega_h=(\Gamma_h)^T\Gamma_h$ for all $h\in \mathcal{L}$, whereby $\Gamma_h$ indicates the row of $\Gamma$ indexed by $h$, by $\Omega_h=(\Gamma_h)^T\allowbreak(\mathcal{V}_L)_{hh}\allowbreak\Gamma_h$ to incorporate our slightly more general setting of an arbitrary latent covariance matrix. The rest of the proof is exactly as in \citet{barber2022half}, we do not restate it.
\end{proof}

\subsection{Statement of Modified Version of Lemma 2 from \citet{foygel2012half}}
\label{sec_modified_lemma_2}
In this section, we consider mixed graphs  $G=(V,D, B)$ as introduced in \citet{foygel2012half} with vertex set $V$ and $|V|=d$, directed edge set $D$ and bidirected edge set $B$. These mixed graphs represent LSEMs with observed variables indexed by $V$ and latent confounding represented by dependent error terms with covariance matrix $\Omega$. (For a more detailed introduction, see \citet{foygel2012half}).
We now have the following modification of Lemma 2 in \citet{foygel2012half}.
\begin{lemma}
\label{lemma_modified_lemma_2}
Consider a mixed graph $G=(V,D, B)$.
 Let $Y, P\subseteq V$ with $|Y|=|P|$, consider a partition $Y=Y_1\dot{\cup}Y_2$ of $Y$, and define the matrix $\mathbf{A}$ for $y\in Y$ and $p\in P$ by 
\begin{align*}
    \mathbf{A}_{yp}:=\begin{cases}
        [(I_d-\Lambda)^T\Sigma]_{yp}, & \text{ if } y\in Y_1\\
        \Sigma_{yp}, & \text{ if } y\in Y_2.
    \end{cases}
\end{align*}
If there exists a system of half-treks from $Y$ to $P$ with no sided intersection, then $\mathbf{A}$ is invertible for generic $(\Lambda, \Omega)\in \Theta_{2022}$, that is, for all $(\Lambda, \Omega)\in \Theta_{2022}\setminus A_{2022}$ where $A_{2022}$ is a proper algebraic subset of $\overline{\Theta_{2022}}$.
\end{lemma}
\begin{proof}

Let $\mathcal{H}(y, p)$ denote the set of all half-treks from $y\in Y$ to $p\in P$ in $G$, and let $\pi(\Lambda, \Omega)$ denote the trek monomial of a trek $\pi$ in the mixed graph $G$.
Note that we can write
\begin{align*}
       \mathbf{A}_{yp}:=\begin{cases}
        \sum_{\pi \in \mathcal{H}(y, p)}\pi(\Lambda, \Omega), & \text{ if } y\in Y_1\\
         \sum_{\pi \in \mathcal{T}(y, p)}\pi(\Lambda, \Omega), & \text{ if } y\in Y_2,
    \end{cases} 
\end{align*}
which is exactly the starting point in the original proof of Lemma 2 in \citet{foygel2012half}. In the proof of Lemma 2 in \citet{foygel2012half}, the specific choices of $Y_1$ and $Y_2$ never occur from this point on in the proof, as do the further assumptions in the statement of Lemma 2 in \citet{foygel2012half}. Our proof from here on is hence exactly as the original proof of Lemma 2 from \citet{foygel2012half}, we do not restate it.
\end{proof}

\subsection{Rational Formulas for Identification for Example \ref{example_elfhtc}}
\label{sec_extended_example_elfhtc}
In the following, we write $(x_1,\dots)^T=A^{-1}b$ if $x_1$ equals the first component of the vector $x=(x_1,\dots)^T$ that uniquely solves the linear equation system $Ax=b$. With this notation, the rational formulas for the presented choices of $W_v$ and $(Y, Z, (W_z)_{z\in Z}, H)$ are
\begin{align*}
	&\begin{pmatrix}
		\lambda_{34} \\
		\vdots
	\end{pmatrix} =
	\begin{pmatrix}
		\Sigma_{13} & \Sigma_{16} & \Sigma_{15} \\
		\Sigma_{23} & \Sigma_{26} & \Sigma_{25} \\
		\Sigma_{33} & \Sigma_{36} & \Sigma_{35}
	\end{pmatrix}^{-1}
\cdot 
\begin{pmatrix}
	\Sigma_{14} \\
	\Sigma_{24} \\
	\Sigma_{34}
\end{pmatrix},\\
	&\begin{pmatrix}
		\lambda_{56} \\
		\vdots
	\end{pmatrix} =
	\begin{pmatrix}
		\Sigma_{15} & \Sigma_{14} & \Sigma_{13} \\
		\Sigma_{25} & \Sigma_{24} & \Sigma_{23} \\
		\Sigma_{35} & \Sigma_{34} & \Sigma_{33}
	\end{pmatrix}^{-1}
	\cdot 
	\begin{pmatrix}
		\Sigma_{16} \\
		\Sigma_{26} \\
		\Sigma_{36}
	\end{pmatrix},\\
	&\begin{pmatrix}
		\lambda_{12} \\
		\vdots
	\end{pmatrix} =
	\begin{pmatrix}
		\Sigma_{11} & \Sigma_{14}-\Lambda_{34}\Sigma_{13}\\
		\Sigma_{61}-\Lambda_{56}\Sigma_{51} & \Sigma_{64}-\Lambda_{56}\Sigma_{54}-(\Sigma_{63}-\Lambda_{56}\Sigma_{53})\Lambda_{34}
	\end{pmatrix}^{-1}
	\cdot 
	\begin{pmatrix}
		\Sigma_{12} \\
		\Sigma_{62} - \Lambda_{56}\Sigma_{52}
	\end{pmatrix},\\
	&\begin{pmatrix}
		\lambda_{23} \\
		\vdots
	\end{pmatrix} =
	\begin{pmatrix}
		\Sigma_{22} - \Lambda_{12}\Sigma_{12} & \Sigma_{21}-\Lambda_{12}\Sigma_{11}\\
	\Sigma_{42} - \Lambda_{34}\Sigma_{32} & \Sigma_{41}-\Lambda_{34}\Sigma_{31}
	\end{pmatrix}^{-1}
	\cdot 
	\begin{pmatrix}
		\Sigma_{23} - \Lambda_{12}\Sigma_{13}\\
		\Sigma_{43} - \Lambda_{34}\Sigma_{33}
	\end{pmatrix}, \textnormal{ and }\\
&(	\lambda_{15},
		\lambda_{45},\ldots)^T\\
		&=
	\begin{pmatrix}
		\Sigma_{11} & \Sigma_{14} & \Sigma_{12} - \Lambda_{12}\Sigma_{11}\\
		\Sigma_{31} - \Lambda_{23}\Sigma_{21} & \Sigma_{34} - \Lambda_{23}\Sigma_{24} & \Sigma_{32} - \Lambda_{23}\Sigma_{22} - (\Sigma_{31} - \Lambda_{23}\Sigma_{21})\Lambda_{12}\\
		\Sigma_{41} - \Lambda_{34}\Sigma_{31} & \Sigma_{44} - \Lambda_{34}\Sigma_{34} & \Sigma_{42} - \Lambda_{34}\Sigma_{32} - (\Sigma_{41}-\Lambda_{34}\Sigma_{31})\Sigma_{12}
	\end{pmatrix}^{-1}
\hspace{-0.2cm}	\cdot \hspace{-0.1cm}
	\begin{pmatrix}
		\Sigma_{15}\\
		\Sigma_{35} - \Lambda_{23}\Sigma_{25}\\
		\Sigma_{45} - \Lambda_{34}\Sigma_{35}
	\end{pmatrix}.
\end{align*}
\subsection{Rational Formulas for Identification for Example \ref{ex_two_proxies}}
\label{sec_rat_two_proxies}
Using the same notation as in Section \ref{sec_extended_example_elfhtc}, the rational formula for the presented choices of $W_v$ and $(Y, Z, (W_z)_{z\in Z}, H)$ is
\begin{align*}
&\begin{pmatrix}
\lambda_{23}\\
\vdots
\end{pmatrix}
=\begin{pmatrix}
\Sigma_{12} & \Sigma_{14}\\
\Sigma_{22} & \Sigma_{24}
\end{pmatrix}^{-1}
\begin{pmatrix}
\Sigma_{13}\\
\Sigma_{23}
\end{pmatrix}.\end{align*}

\section{Appendix to Section \ref{sec_further_complementary_results}}
We begin this section by proving Theorem \ref{theorem_determinantal} in Section \ref{sec_proof_determinantal}. We then show Proposition \ref{allowed_covariances} in Section \ref{sec_recursive_proof_1}. In Section \ref{sec_recursive_proof_2}, we then prove Proposition \ref{proposition_id_notion_subgraph}. Finally, we prove Proposition \ref{propo_deletion_procedure} in Section \ref{sec_deletion_procedure}. 

\subsection{Proof of Theorem \ref{theorem_determinantal}}
\label{sec_proof_determinantal}
\begin{proof}
To prove the result, we apply Theorem 3.8 from \citet{weihs2018determinantal} to the graph $G^\mathcal{L}$ (which does not contain bidirected edges) and only allow for sets $S, T$ containing observed nodes. As we will explain, 
rational identifiability of $\lambda_{w_0v}$ as in Theorem 3.8 from \citet{weihs2018determinantal} then implies rational identifiability of $\lambda_{w_0v}$ with respect to our employed notion.
    
    For defining rational identifiability as in \citet{weihs2018determinantal}, first note that because
    	\begin{align*}
		\begin{pmatrix}
			L\\
			X
		\end{pmatrix} = 
		\begin{pmatrix}
				0 & 0 \\
				\Gamma^T & \Lambda^T
		\end{pmatrix}
		\cdot 
		\begin{pmatrix}
			L \\
			X
		\end{pmatrix} +
		\begin{pmatrix}
			L\\
			\epsilon
		\end{pmatrix},
	\end{align*}
	and hence that
	    	\begin{align*}
	\begin{pmatrix}
			L\\
			X
		\end{pmatrix} =
			\left(I_{\ell + d} - \begin{pmatrix}
				0 & \Gamma \\
				0 & \Lambda
		\end{pmatrix}\right)^{-T}
		\begin{pmatrix}
			L\\
			\epsilon
		\end{pmatrix},
	\end{align*}
	where invertibility holds for all $(\Lambda, \Gamma) \in \mathbb{R}^{D_V}_{\textnormal{reg}}\times \mathbb{R}^{D_{\mathcal{L}V}}$ as argued further down below,
     the vector $(L,X)^T$ has covariance matrix
    \begin{align*}
   \Sigma^{L, X}&= \chi^{L, X}_{G^\mathcal{L}}\left(\begin{pmatrix}
        0 &\Gamma \\
        0 &\Lambda
    \end{pmatrix}, \Omega_{\textnormal{diag}}, \mathcal{V}_L\right)\\
   &:= \begin{pmatrix}
			\mathcal{V}_L & \textnormal{Cov}(L, X) \\
			\textnormal{Cov}(X, L) & \Sigma
		\end{pmatrix}\\
		&=
    \left(I_{\ell+d}-\begin{pmatrix}
        0 &\Gamma \\
        0 &\Lambda
    \end{pmatrix}\right)^{-T}\underbrace{\begin{pmatrix}
        \mathcal{V}_L & 0\\
        0 & \Omega_{\textnormal{diag}}
    \end{pmatrix}}_{=\textnormal{Cov}(L, \epsilon)}
    \left(I_{\ell+d}-\begin{pmatrix}
        0 &\Gamma \\
        0 &\Lambda
    \end{pmatrix}\right)^{-1}\\
    &=\begin{pmatrix}
       I_\ell &-\Gamma \\
        0 &I_d-\Lambda
    \end{pmatrix}^{-T}\begin{pmatrix}
        \mathcal{V}_L & 0\\
        0 & \Omega_{\textnormal{diag}}
    \end{pmatrix}
    \begin{pmatrix}
        I_\ell &-\Gamma \\
        0 &I_d-\Lambda
    \end{pmatrix}^{-1}\\
     &=\begin{pmatrix}
        I_\ell &\Gamma(I_d-\Lambda)^{-1} \\
        0 &(I_d-\Lambda)^{-1}
    \end{pmatrix}^{T}\begin{pmatrix}
      \mathcal{V}_L & 0\\
        0 & \Omega_{\textnormal{diag}}
    \end{pmatrix}
\begin{pmatrix}
        I_\ell &\Gamma(I_d-\Lambda)^{-1} \\
        0 &(I_d-\Lambda)^{-1}
    \end{pmatrix}\\
    &=\begin{pmatrix}
        \mathcal{V}_L & \mathcal{V}_L\Gamma(I_d-\Lambda)^{-1}\\
        (I_d-\Lambda)^{-T}\Gamma^T\mathcal{V}_L & (I_d-\Lambda)^{-T}(\Omega_{\textnormal{diag}}+\Gamma^T\mathcal{V}_L\Gamma)(I_d-\Lambda)^{-1}
    \end{pmatrix}.
\end{align*}
Here, $ \chi^{L, X}_{G^\mathcal{L}}$ has domain $\Theta^{L, X}:=\mathbb{R}^D_{\textnormal{reg}}\times \textnormal{diag}^+_d\times \textnormal{diag}^+_\ell$, where $\textnormal{diag}^+_d\times \textnormal{diag}^+_\ell$ equals the set of positive definite covariance matrices as defined in \citet{weihs2018determinantal} corresponding to the bidirected edges between the components of $(L, X)$ because there are no bidirected edges in our case; and where furthermore, $\mathbb{R}^D_{\textnormal{reg}}$ is the space of all real-valued matrices which have the same number of rows and columns and the same zero structure as
\begin{align*}
    \begin{pmatrix}
        0 &\Gamma \\
        0 &\Lambda
    \end{pmatrix}
\end{align*}
and for which 
\begin{align*}
    \left(I_{\ell+d}-\begin{pmatrix}
        0 &\Gamma \\
        0 &\Lambda
    \end{pmatrix}\right)
\end{align*}
is invertible. Note that because 
\begin{align*}
    \det\left(I_{\ell+d}-\begin{pmatrix}
        0 &\Gamma \\
        0 &\Lambda
    \end{pmatrix}\right)=\det(I_{l})\cdot \det(I_d-\Lambda),
\end{align*}
it follows that $\mathbb{R}^D_{\textnormal{reg}}\cong\mathbb{R}^{D_V}_{\textnormal{reg}}\times \mathbb{R}^{D_{\mathcal{L}V}}$ and hence that 
$\Theta^{L, X}\cong \Theta$. Therefore and by definition of $\chi^{L, X}_{G^\mathcal{L}}$, we have
\begin{align*}
    \chi_{G^\mathcal{L}}(\Lambda, \Gamma, \Omega_{\textnormal{diag}}, \mathcal{V}_L)=  \left(\chi^{L, X}_{G^\mathcal{L}}\left(\begin{pmatrix}
        0 &\Gamma \\
        0 &\Lambda
    \end{pmatrix}, \Omega_{\textnormal{diag}},\mathcal{V}_L\right)\right)_{4,4}
\end{align*}
 for all $(\Lambda, \Gamma,\Omega_{\textnormal{diag}}, \mathcal{V}_L)\in \mathbb{R}^{D_V}_{\textnormal{reg}}\times \mathbb{R}^{D_{\mathcal{L}V}}\times \textnormal{diag}^+_d\times \textnormal{diag}^+_\ell=\Theta$, and vice versa 
   for all 
\begin{align*}
   \left(\begin{pmatrix}
        0 &\Gamma \\
        0 &\Lambda
    \end{pmatrix}, \Omega_{\textnormal{diag}},\mathcal{V}_L\right)\in \Theta^{L, X}.
\end{align*}

\citet{weihs2018determinantal} now say that some edge $x\rightarrow y\in D_V$ is rationally identifiable if there exists a rational function $\psi^{L, X}$ such that 
\begin{align*}
   (\psi^{L, X}\circ \chi^{L, X}_{G^\mathcal{L}})\left(\begin{pmatrix}
        0 &\Gamma \\
        0 &\Lambda
    \end{pmatrix}, \Omega_{\textnormal{diag}}, \mathcal{V}_L\right)=\lambda_{xy}
\end{align*}
for all 
\begin{align*}
    \left(\begin{pmatrix}
        0 &\Gamma \\
        0 &\Lambda
    \end{pmatrix}, \Omega_{\textnormal{diag}}, \mathcal{V}_L\right)\in\Theta^{L, X}\setminus A^{L, X}_\Theta
\end{align*}
 for some proper algebraic subset $A^{L, X}_\Theta$ of $\Theta^{L, X}$. It holds that our notion of rational identifiability implies the notion of rational identifiability from \citet{weihs2018determinantal}: To see this fact, let $\psi$ the rational function in our employed notion of rational identifiability, let $A_\Theta$ be the occuring algebraic set, and define $\psi^{L, X}$ by
\begin{align*}
    \psi^{L, X}\left(\begin{pmatrix}
			\mathcal{V}_L & \textnormal{Cov}(L, X) \\
			\textnormal{Cov}(X, L) & \Sigma
		\end{pmatrix}\right) := \psi(\Sigma).
\end{align*}
It then immediately follows that for all 
\begin{align*}
    \left(\begin{pmatrix}
        0 &\Gamma \\
        0 &\Lambda
    \end{pmatrix}, \Omega_{\textnormal{diag}}, \mathcal{V}_L\right)\in\Theta^{L, X}\setminus A^{L, X}_\Theta,
\end{align*}
where $A^{L, X}_\Theta$ is the isomorphic correspondent of $A_\Theta$, that
\begin{align*}
  (\psi^{L, X}\circ  \chi^{L, X}_{G^\mathcal{L}})\left(\begin{pmatrix}
        0 &\Gamma \\
        0 &\Lambda
    \end{pmatrix}, \Omega_{\textnormal{diag}},\mathcal{V}_L\right) &=(\psi\circ(\chi^{L, X}_{G^\mathcal{L}})_{4,4})\left(\begin{pmatrix}
        0 &\Gamma \\
        0 &\Lambda
    \end{pmatrix}, \Omega_{\textnormal{diag}},\mathcal{V}_L\right)\\ &=(\psi\circ\chi_{G^\mathcal{L}})(\Lambda, \Gamma, \Omega_{\textnormal{diag}}, \mathcal{V}_L) = \lambda_{xy}.
\end{align*}
Therefore, the assumption in our theorem that $\lambda_{w_1v},\ldots, \lambda_{w_mv}$ are rationally identifiable (according to our notion) implies rational identifiable of these same coefficients as in \citet{weihs2018determinantal}, so Theorem 3.8 from \citet{weihs2018determinantal} applies.

Theorem 3.8 from \citet{weihs2018determinantal} now yields rational identifiability with respect to the notion of rational identifiability from \citet{weihs2018determinantal}. Writing $A^{L, X}_{\Theta, 1},\ldots, A^{L, X}_{\Theta, m}$ for the isormorphic correspondents of the algebraic sets $A_{\Theta, 1},\ldots, A_{\Theta, m}$ occuring in the rational identifiability assumptions of the edges $w_1\rightarrow v,\ldots, w_m\rightarrow v$ with respect to our employed terminology,  there hence exists a rational function $\psi^{L, X}$ and a proper algebraic subset $A^{L, X}_\Theta\subseteq \Theta^{L, X}$,
where $A^{L, X}_\Theta$ is given by
\begin{align*}
    A^{L, X}_\Theta&:=\underbrace{\left\{\left(\begin{pmatrix}
        0 &\Gamma \\
        0 &\Lambda
    \end{pmatrix}, \Omega_{\textnormal{diag}},\mathcal{V}_L\right)\in \Theta^{L, X}:\;\det(\Sigma_{S, T\cup \{w_0\}})=0\right\}}_{=:A^{L, X}_{\Theta, 0}}
    \cup \bigcup_{i=1}^m A^{L, X}_{\Theta, i},
\end{align*}
such that 
\begin{align*}
 (\psi^{L, X}\circ \chi^{L, X}_{G^\mathcal{L}})\left(\begin{pmatrix}
        0 &\Gamma \\
        0 &\Lambda
    \end{pmatrix}, \Omega_{\textnormal{diag}}, \mathcal{V}_L\right) = \lambda_{w_0v}   
\end{align*}
    for all 
\begin{align*}
   \left(\begin{pmatrix}
        0 &\Gamma \\
        0 &\Lambda
    \end{pmatrix}, \Omega_{\textnormal{diag}},\mathcal{V}_L\right)\in \Theta^{L, X}\setminus A^{L,X}_\Theta.
\end{align*}
In particular, it holds that
\begin{align*}
     (\psi^{L, X}\circ \chi^{L, X}_{G^\mathcal{L}})\left(\begin{pmatrix}
        0 &\Gamma \\
        0 &\Lambda
    \end{pmatrix}, \Omega_{\textnormal{diag}}, \mathcal{V}_L\right)= \frac{\det{(\Sigma_{S, T\cup \{v\}})} - \sum_{i=1}^m\lambda_{w_iv}\det{(\Sigma_{S, T\cup \{w_i\}})}}{\det{(\Sigma_{S, T\cup \{w_0\}})}},
\end{align*}
where the $\lambda_{w_iv}$'s are rational functions in terms of $\Sigma$ on $\Theta^{L, X}\setminus A^{L, X}_\Theta$ and hence, $\psi^{L, X}$ is indeed a rational function.

The stated rational formula $\psi^{L, X}$ in Theorem 3.8 of \citet{weihs2018determinantal}, however, just contains entries of $\Sigma$ (where we use the assumption that the coefficients $\lambda_{w_1v},\ldots, \lambda_{w_mv}$ are rationally identifiable according to our notion which implies that $\lambda_{w_1v},\ldots, \lambda_{w_mv}$ are also rationally identifiable according to the notion of \citet{weihs2018determinantal} by just entries of $\Sigma$), and so in particular, the occuring rational function $\psi^{L, X}$ can be rewritten to a function $\psi$ that just depends on the arguments of $\Sigma$, that is, we can write
\begin{align*}
    \psi(\Sigma) = \psi^{L, X}\left(\begin{pmatrix}
			\mathcal{V}_L & \textnormal{Cov}(L, X) \\
			\textnormal{Cov}(X, L) & \Sigma
		\end{pmatrix}\right)
\end{align*}
which implies that
\begin{align*}
   &(\psi\circ\chi_{G^\mathcal{L}})\left(\Lambda,\Gamma, \Omega_{\textnormal{diag}}, \mathcal{V}_L\right)\\
   &=(\psi\circ(\chi^{L, X}_{G^\mathcal{L}})_{4, 4})\left(\begin{pmatrix}
        0 &\Gamma \\
        0 &\Lambda
    \end{pmatrix}, \Omega_{\textnormal{diag}}, \mathcal{V}_L\right)\\
    &= (\psi^{L, X}\circ \chi^{L, X}_{G^\mathcal{L}})\left(\begin{pmatrix}
        0 &\Gamma \\
        0 &\Lambda
    \end{pmatrix}, \Omega_{\textnormal{diag}}, \mathcal{V}_L\right)=\lambda_{w_0v}
\end{align*}
  for all 
\begin{align*}
   \left(\begin{pmatrix}
        0 &\Gamma \\
        0 &\Lambda
    \end{pmatrix}, \Omega_{\textnormal{diag}},\mathcal{V}_L\right)\in \Theta^{L, X}\setminus A^{L,X}_\Theta.
\end{align*}
respectively, for all $(\Lambda, \Gamma, \Omega_{\textnormal{diag}}, \mathcal{V}_L)\in \Theta\setminus A_\Theta$ where $A_\Theta$ is the isomorphic correspondent of $A^{L, X}_\Theta$.
\end{proof}

\subsection{Proof of Proposition \ref{allowed_covariances}}
\label{sec_recursive_proof_1}
\
\begin{proof}
\underline{Case 1:} In this case, no trek from $x$ to $y$ uses one of the deleted edges, as otherwise, either $x$ or $y$ would either be a descendant of $v$ in $G^\mathcal{L}$ or would equal $v$, contradiction.
 Therefore,
 every trek from $x$ to $y$ in $\overline{G}^\mathcal{L}$ occurs in $G^\mathcal{L}$, and vice versa, every trek from $x$ to $y$ in $G^\mathcal{L}$ occurs in $\overline{G}^\mathcal{L}$. From the trek rule (see, e.g., Section 2 in \citet{sullivant2010trek} or Section 2 in \citet{foygel2012half}) then follows the result.

\underline{Case 2:} 
In this case, every trek from $x$ to $y=v$ in $G^\mathcal{L}$ ends with an edge $p^v\rightarrow v$ with $p^v\in \textnormal{pa}(v)$; otherwise, either $x=v$ or $x\in \textnormal{dec}_{G^\mathcal{L}}(v)$. 
Also, in every trek from $x$ to $y=v$ in $G^\mathcal{L}$ an edge of the type $p^v\rightarrow v$ with $p^v\in \textnormal{pa}(v)$ can only occur on the right-side of the trek, as otherwise $x\in \textnormal{dec}_{G^\mathcal{L}}(v)$.
Furthermore, in every trek from $x$ to $y=v$ in $G^\mathcal{L}$ an edge of the type $p^v\rightarrow v$ with $p^v\in \textnormal{pa}(v)$ appears at most (and hence exactly) once, as otherwise $v=y\in\textnormal{dec}_{G^\mathcal{L}}(v)$ (as all these edges of the type $p^v\rightarrow v$ with $p^v\in \textnormal{pa}(v)$ must appear on the right-side of the trek as previously argued), contradiction.

Therefore,
the set of all the treks in $\overline{G}^\mathcal{L}$ from $x$ to $y$ equals the set of all the treks from $x$ to $y$ in $G^\mathcal{L}$ minus the set of all the treks in $G^\mathcal{L}$ from $x$ to $y$ that end via any of the deleted edges $w_i\rightarrow v=y$. From the trek rule  (see, e.g., Section 2 in \citet{sullivant2010trek} or Section 2 in \citet{foygel2012half}), it follows that $\Sigma_{xw_i}$ equals the sum of all trek monomials over all treks from $x$ to $w_i$ and it thus also follows that $\Sigma_{xw_i}\Lambda_{w_iy}$ equals the sum of all trek monomials over all treks from $x$ to $v=y$ in $G$ having as last edge $w_i\rightarrow v=y$. From the trek rule, it also follows that $\Sigma_{xy}$ equals the sum of all trek monomials over all treks from $x$ to $v=y$. Therefore, Case 2 follows.

\underline{Case 3:} Switch roles of $x$ and $y$ in the proof of Case 2.
\end{proof}

\subsection{Proof of Proposition \ref{proposition_id_notion_subgraph}}
\label{sec_recursive_proof_2}

\begin{proof}
Write $\Theta_{\textnormal{sub}}:=\mathbb{R}^{D^{\textnormal{sub}}_V}_{\textnormal{reg}}\times \mathbb{R}^{D_{\mathcal{L}V}}\times \textnormal{diag}^+_d\times \textnormal{diag}^+_\ell$, where $D^{\textnormal{sub}}_V:=D_V\setminus \{w_1\rightarrow v,\ldots, w_m\rightarrow v\}$.
Let $A_{\textnormal{sub}}\subsetneq \overline{\Theta_{\textnormal{sub}}}$ be the proper algebraic subset for which rational identifiability of $\Lambda_{w_0v_0}$ with respect to $\overline{G}^\mathcal{L}$ holds. For all $i\in\{1,\ldots,m\}$, let $A_i\subsetneq \overline{\Theta}$ be the proper algebraic subset for the definition of rational identifiability of $w_i\rightarrow v$ with respect to $G^\mathcal{L}$. In the following, because $ \overline{\Theta_{\textnormal{sub}}}\cong \mathbb{R}^{m_{\Theta_{\textnormal{sub}}}}$, where $m_{\Theta_{\textnormal{sub}}}=|D^{\textnormal{sub}}_V|+|D_{\mathcal{L}V}|+d+\ell=|D_V|-m+|D_{\mathcal{L}V}|+d+\ell$, and because $\overline{\Theta}\cong \mathbb{R}^m$, we in slight abuse of notation write $\mathbb{R}^m\times B$ for any subset $B\subseteq \overline{\Theta_{\textnormal{sub}}}$ to indicate the subset of $\overline{\Theta}$ for which the coefficients corresponding to $w_1\rightarrow v_1,\ldots, w_m\rightarrow v_m$ can be chosen from the real-numbers and for which the remaining coefficients corresponding to $\overline{G}^\mathcal{L}$ must be in $B$.
We now construct a new proper algebraic subset $A\subsetneq \overline{\Theta}\cong\mathbb{R}^{m_\Theta}$ for rational identifiability of $\Lambda_{w_0v_0}$ with respect to $G^\mathcal{L}$ by
\begin{align*}
    A:=(\mathbb{R}^m\times A_{\textnormal{sub}}) \cup \bigcup_{i=1}^mA_i.
\end{align*}
 Note that $\mathbb{R}^m\times A_{\textnormal{sub}}$ is still an algebraic set as the Cartesian product of algebraic sets is an algebraic set (see the beginning of Section \ref{app_further_exp} in the Appendix for a proof of this fact). Also, note that $\mathbb{R}^m\times A_{\textnormal{sub}}$ is proper because $A_{\textnormal{sub}}$ is proper. Because the finite union of algebraic sets is also an algebraic set, it follows that $A$ is an algebraic subset of $\overline{\Theta}\cong\mathbb{R}^{m_\Theta}$.
Regarding properness of $A$: The set $A$ is a finite union of proper algebraic subsets of the irreducible set $\overline{\Theta}$ and hence, $A$ has Lebesgue measure zero (as mentioned in Section \ref{sec_generic_id}, see, e.g., the lemma in \citet{okamoto1973distinctness}). However, $\overline{\Theta}\cong\mathbb{R}^{m_\Theta}$ does not have Lebesgue measure zero.
Hence, $A$ must be a proper subset of $\overline{\Theta}$.

Now, write $\chi_{\overline{G}^\mathcal{L}}:\Theta_{\textnormal{sub}}\rightarrow \textnormal{PD}(d)$ for the function mapping parameters in $\Theta_{\textnormal{sub}}$ to observed covariance matrices with respect to $\overline{G}^\mathcal{L}$. Write $\chi_{\overline{G}^\mathcal{L},\textnormal{allowed}}:=\pi\circ \chi_{\overline{G}^\mathcal{L}}$ where $\pi$ is the projection function onto the covariances that can be calculated via Proposition \ref{allowed_covariances}. Note that from Proposition \ref{allowed_covariances} it follows that we can write $\chi_{\overline{G}^\mathcal{L},\textnormal{allowed}}(\Lambda_{\textnormal{sub}}, \Gamma, \Omega_{\textnormal{diag}}, \mathcal{V}_L)$, where $\Lambda_{\textnormal{sub}}$ is $\Lambda$ with the entries $\Lambda_{w_1v},\ldots,\Lambda_{w_mv}$ set to zero, as a polynomial and hence rational function in terms of $\chi_{G^\mathcal{L}}(\Lambda, \Gamma, \Omega_{\textnormal{diag}}, \mathcal{V}_L)$ and $\Lambda_{w_1v},\ldots,\Lambda_{w_mv}$ for all $(\Lambda, \Gamma, \Omega_{\textnormal{diag}},\allowbreak \mathcal{V}_L)\allowbreak\in \Theta\setminus A$. That is, for all $(\Lambda, \Gamma, \Omega_{\textnormal{diag}}, \mathcal{V}_L)\in \Theta\setminus A$ we can write $\chi_{\overline{G}^\mathcal{L},\textnormal{allowed}}(\Lambda_{\textnormal{sub}}, \Gamma, \Omega_{\textnormal{diag}}, \mathcal{V}_L)=\gamma(\chi_{G^\mathcal{L}}(\Lambda, \Gamma, \Omega_{\textnormal{diag}}, \mathcal{V}_L),\allowbreak\Lambda_{w_1v},\allowbreak\ldots,\Lambda_{w_mv})$ for a rational function $\gamma$.

 Let  $\psi_{\textnormal{sub}}$ be the rational function such that $(\psi_{\textnormal{sub}}\circ \chi_{\overline{G}^\mathcal{L}})(\Lambda_{\textnormal{sub}}, \Gamma, \Omega_{\textnormal{diag}},\mathcal{V}_L)=\Lambda_{w_0v_0}$ for all  $(\Lambda_{\textnormal{sub}}, \Gamma, \Omega_{\textnormal{diag}},\mathcal{V}_L)\in \Theta_{\textnormal{sub}}\setminus A_{\textnormal{sub}}$. 
As one is just using  covariances that can be calculated by Proposition \ref{allowed_covariances} for rational identifiability of $\Lambda_{w_0v}$ with respect to $\overline{G}^\mathcal{L}$, we have for some other rational function $\psi_{\textnormal{sub},\textnormal{allowed}}$ that
\begin{align*}
    &(\psi_{\textnormal{sub},\textnormal{allowed}}\circ\chi_{\overline{G}^\mathcal{L},\textnormal{allowed}})(\Lambda_{\textnormal{sub}}, \Gamma, \Omega_{\textnormal{diag}},\mathcal{V}_L)\\
   &=(\psi_{\textnormal{sub}}\circ \chi_{\overline{G}^\mathcal{L}})(\Lambda_{\textnormal{sub}}, \Gamma, \Omega_{\textnormal{diag}},\mathcal{V}_L)
\end{align*}
 for all $(\Lambda_{\textnormal{sub}}, \Gamma, \Omega_{\textnormal{diag}},\mathcal{V}_L)\in\Theta_{\textnormal{sub}}\setminus A_{\textnormal{sub}}$ and hence that
\begin{align*}
&(\psi_{\textnormal{sub},\textnormal{allowed}}(\gamma(\chi_{G^\mathcal{L}}(\Lambda, \Gamma, \Omega_{\textnormal{diag}}, \mathcal{V}_L),\Lambda_{w_1v},\ldots,\Lambda_{w_mv}))\\
&(\psi_{\textnormal{sub},\textnormal{allowed}}\circ\chi_{\overline{G}^\mathcal{L},\textnormal{allowed}})(\Lambda_{\textnormal{sub}}, \Gamma, \Omega_{\textnormal{diag}},\mathcal{V}_L)\\
   &=(\psi_{\textnormal{sub}}\circ \chi_{\overline{G}^\mathcal{L}})(\Lambda_{\textnormal{sub}}, \Gamma, \Omega_{\textnormal{diag}},\mathcal{V}_L)\\
    &=\Lambda_{w_0v_0},
\end{align*}
for all $(\Lambda, \Gamma, \Omega_{\textnormal{diag}},\mathcal{V}_L)\in\Theta\setminus A$. Each $\Lambda_{w_iv}$ is itself rationally identifiable in $G^\mathcal{L}$, and thus, there exists $\psi_i$ for each $i\in\{1,\ldots,m\}$ such that
$(\psi_i\circ \chi_{G^\mathcal{L}})(\Lambda, \Gamma, \Omega_{\textnormal{diag}},\mathcal{V}_L)=\Lambda_{w_iv}$ for all $(\Lambda, \Gamma, \Omega_{\textnormal{diag}},\mathcal{V}_L) \in \Theta\setminus A_i$ and hence in particular for all $(\Lambda, \Gamma, \Omega_{\textnormal{diag}},\mathcal{V}_L) \in \Theta\setminus A$. Therefore, we have for all $(\Lambda, \Gamma, \Omega_{\textnormal{diag}},\mathcal{V}_L) \in \Theta\setminus A$
\begin{align*}
   \Lambda_{w_0v_0} &= (\psi_{\textnormal{sub},\textnormal{allowed}}(\gamma(\chi_{G^\mathcal{L}}(\Lambda, \Gamma, \Omega_{\textnormal{diag}}, \mathcal{V}_L),\Lambda_{w_1v},\ldots,\Lambda_{w_mv}))\\&
   =(\psi_{\textnormal{sub},\textnormal{allowed}}(\gamma(\chi_{G^\mathcal{L}}(\Lambda_{\textnormal{sub}}, \Gamma, \Omega_{\textnormal{diag}},\mathcal{V}_L),\\&
   \hspace{3cm}(\psi_1\circ \chi_{G^\mathcal{L}})(\Lambda, \Gamma, \Omega_{\textnormal{diag}},\mathcal{V}_L),\ldots,\\&
   \hspace{3cm}(\psi_m\circ \chi_{G^\mathcal{L}})(\Lambda, \Gamma, \Omega_{\textnormal{diag}},\mathcal{V}_L))\\
   &=(\psi_{\textnormal{sub},\textnormal{allowed}}\circ\gamma\circ( \chi_{G^\mathcal{L}}, (\psi_1\circ \chi_{G^\mathcal{L}}),\ldots,(\psi_m\circ \chi_{G^\mathcal{L}})))(\Lambda, \Gamma, \Omega_{\textnormal{diag}},\mathcal{V}_L)
   \\
   &=(\psi_{\textnormal{sub},\textnormal{allowed}}\circ\gamma\circ( (\textnormal{id}\circ\chi_{G^\mathcal{L}}), (\psi_1\circ \chi_{G^\mathcal{L}}),\ldots,(\psi_m\circ \chi_{G^\mathcal{L}})))(\Lambda, \Gamma, \Omega_{\textnormal{diag}},\mathcal{V}_L)
   \\
      &=(\psi_{\textnormal{sub},\textnormal{allowed}}\circ\gamma\circ( \textnormal{id}, \psi_1,\ldots,\psi_m)\circ \chi_{G^\mathcal{L}})(\Lambda, \Gamma, \Omega_{\textnormal{diag}},\mathcal{V}_L)
   \\
   &=(\psi\circ \chi_{G^\mathcal{L}})(\Lambda, \Gamma, \Omega_{\textnormal{diag}},\mathcal{V}_L),
\end{align*}
where $\psi:=\psi_{\textnormal{sub},\textnormal{allowed}}\circ\gamma\circ( \textnormal{id}, \psi_1,\ldots,\psi_m)$ is a rational function which implies the result.

\textbf{Regarding the notion of rational identifiability from \citet{barber2022half}:} The previous proof also works for the notion of rational identifiability from \citet{barber2022half}. Just replace $\mathbb{R}^{D_{\mathcal{L}V}}\times \textnormal{diag}^+_d\times \textnormal{diag}^+_\ell$ by $\textnormal{im}(\tau)$; the remainder of the proof is then exactly the same as only the first part of the Cartesian product---the part corresponding to $\Lambda$---is used and connected with $\Lambda_{\textnormal{sub}}$ in the remainder of the proof.
\end{proof}

\subsection{Proof of Proposition \ref{propo_deletion_procedure}}
\label{sec_deletion_procedure}
For $x, y\in V\cup \mathcal{L}$, write $x\preceq y$ if $x\notin \textnormal{dec}_{G^\mathcal{L}}(y)$. We now argue that $(\Sigma_n)_{xy}$ being allowed with respect to the deletion sequence $D_{v_1},\ldots, D_{v_n}$ is equivalent to $x,y\preceq v_1,\ldots, v_n$ (part 1) and ($\{x,y\}\cap \{v_1,\ldots,v_n\}=\emptyset$ or ($\{x,y\}\cap \{v_1,\ldots,v_n\}\neq \emptyset$ and $x\neq y$)) (part 2). This equivalence then implies the result because the right-hand side of this equivalence does not depend on the order of the $v_1,\ldots, v_n$ and also only depends on the unique elements of $v_1,\ldots, v_n$, that is, on $\{v_1,\ldots, v_n\}$, and thus, the equivalent condition on the right-hand side is the same for every other choice $v'_1,\ldots,v'_{n'}\in V$ for which $\bigcup_{i=1}^{n'}D_{v'_i}=D_{\textnormal{del}}$ (as $\{v_1,\ldots, v_n\}=\{v'_1,\ldots, v'_{n'}\}$ for every other such deletion sequence).

For the ``$\Rightarrow $"-direction, part 2 immediately follows from the equality conditions in Proposition \ref{allowed_covariances}. From the repeated application of Proposition \ref{allowed_covariances}, it also follows that $x,y\notin \textnormal{dec}_{G^\mathcal{L}_{i-1}}(v_i)$ for all $i\in \{1,\ldots, n\}$ (where $G^\mathcal{L}_0:=G^\mathcal{L}$). Now, suppose without loss of generality that $i\in \{1,\ldots,n\}$ is the \emph{minimal} index for which $x\in \textnormal{dec}_{G^\mathcal{L}}(v_i)$, which would contradict part 1. Suppose first $i>1$. Then, there is a nontrivial directed path in $G^\mathcal{L}$ from $v_i$ to $x$ that has to use one of the deleted edges from $D_{v_1},\ldots, D_{v_{i-1}}$, otherwise, $x\in \textnormal{dec}_{G^\mathcal{L}_{i-1}}(v_i)$. However, then, $x\in \textnormal{dec}_{G^\mathcal{L}}(\{v_1,\ldots, v_{i-1}\})$ which is a contradiction to the minimality of $i$. If $i=1$, then $G^\mathcal{L}_{i-1}=G^\mathcal{L}$ anyhow, so this case is trivial. For the `` $\Leftarrow$"-direction: Note that $x,y\preceq v_1,\ldots,v_n$ and $\{x,y\}\cap \{v_1,\ldots,v_n\}=\emptyset$ immediately implies that Case 1 of Proposition \ref{allowed_covariances} applies to each deletion (as deleting edges from $G^\mathcal{L}$ can only reduce the number of descendants). Similarly, if $x,y\preceq v_1,\ldots,v_n$  and ($\{x,y\}\cap \{v_1,\ldots,v_n\}\neq \emptyset$ and $x\neq y$), then one of the three cases from Proposition \ref{allowed_covariances} at each deletion step applies. In addition,
for all $w\in \textnormal{pa}(\{x, y\})$, it holds that $w\preceq v_1,\ldots, v_n$ and $w\neq v_1,\ldots, v_n$, otherwise $x$ or $y$ is in $ \textnormal{dec}_{G^{\mathcal{L}}}(\{v_1,\ldots, v_n\})$, contradiction. Therefore, for all $w\in \textnormal{pa}_V(y)$ and $w'\in \textnormal{pa}_V(x)$, one of the three cases from Proposition \ref{allowed_covariances} also applies at each deletion step $i\in \{1,\ldots, n\}$ to the covariances $(\Sigma_i)_{xw}$ and $(\Sigma_i)_{yw'}$, and Case 1 applies to $(\Sigma_i)_{ww'}$, thus making all further required covariances for Proposition \ref{allowed_covariances} allowed at each step $i$.

\section{Appendix to Section \ref{section_computation}}
\label{app_sec_computation}
We begin this section by proving  Theorem \ref{theorem_existence_of_Y} in Section \ref{sec_numerical_proof_1}. In Section \ref{sec_numerical_proof_2}, we then prove Proposition \ref{propo_4_children}. Finally, in Section \ref{sec_numerical_proof_3}, we prove Theorem \ref{theorem_complexity}.

\subsection{Proof of Theorem \ref{theorem_existence_of_Y}}
\label{sec_numerical_proof_1}
\begin{proof}
The set of allowed nodes $A$ is already defined such that condition (ii) is satisfied for every $Y\subseteq A$. Furthermore, the remaining assumptions
 in Theorem \ref{theorem_existence_of_Y} ensure that condition (i) of the eLF-HTC is satisfied except for $|Y|=|W_v\cup Z\cup W_Z|$. Thus, Theorem \ref{theorem_existence_of_Y} really is a statement about condition (iii) of the eLF-HTC (plus the just mentioned remaining part of condition (i)). Now, the set $W_v\cup W_Z$ in condition (iii) of the eLF-HTC is $\textnormal{pa}_V(v)$ in condition (iii) of the LF-HTC. Otherwise, condition (iii) in the eLF-HTC and condition (iii) in the LF-HTC are equal. 
 Now, similarly as in the LF-HTC paper, introduce a source node and a sink node in order to transform our multi-source and multi-sink max-flow problem for the eLF-HTC into an equivalent single-source and single-sink max-flow problem with the same max-flow by adding a source node $s$ which points to all vertices in $A$ and by adding a sink node $t$ such that from each vertex in $W_v'\cup W_z'\cup Z'$ there is an edge to $t$, and such that $s$ and $t$ have capacity $\infty$. Next, the set $W_v\cup W_Z$ (and $W'_v\cup W'_Z$) in this equivalent flow network for the eLF-HTC is $\textnormal{pa}_V(v)$ in the flow network for the LF-HTC (which is introduced in \citet{barber2022half}). Otherwise, the (transformed) flow network constructions are equal. Said differently, for both the eLF-HTC and the LF-HTC, it is possible to replace $W_v\cup W_Z$ respectively $\textnormal{pa}_V(v)$ (and the corresponding primed versions) for the sake of this proof by any arbitrary placeholder sets since the precise definitions of these sets are neither used for condition (iii) (and the remainder of condition (i)) 
 nor are they used for the corresponding (transformed) flow networks. Thus, conditions (iii) in the eLF-HTC and the LF-HTC and the corresponding (transformed) flow networks looked at in an isolated manner are the same. Therefore, the result for condition (iii) for the eLF-HTC follows from the same proof as the proof of Theorem 5.1 in \citet{barber2022half}. The remainder of condition (i) stating that $|Y|=|W_v\cup Z\cup W_Z|$ then follows automatically as otherwise, no such system of latent-factor half-treks from $Y$ to $W_v\cup Z\cup W_Z$ can exist.
\end{proof}

\subsection{Proof of Proposition \ref{propo_4_children}}
\label{sec_numerical_proof_2}
\begin{proof}
The proof is very similar to the proof of Proposition 5.2 in \citet{barber2022half}. First note that by assumption there exists $h\in H$ so $|H|=|Z|\geq 1$. Now, suppose $(Y,\allowbreak Z,\allowbreak (W_z)_{z\in Z},\allowbreak  H)$ satisfies the eLF-HTC with respect to $(v, W_v)$. Because there exists a system of latent-factor half-treks from $Y$ to $W_v\cup W_Z\cup Z$ such that for each $z\in Z$ the latent-factor half-trek ending in $z$ has the form $y\leftarrow h\rightarrow z$ for some $y\in Y$ and because $|H|=|Z|\geq 1$, it follows that $|\textnormal{ch}(h)|> 1$ for all $h\in H$. 

Next, let $h\in H$ such that $|\textnormal{ch}(h)|\in \{2,3\}$. Then, in the system of latent-factor half-treks for condition (iii) there is a latent-factor half-trek of the form $y\leftarrow h\rightarrow z$ for some $y\in Y$ and some $z\in Z$. Now, take $\Tilde{Y}=Y\setminus \{y\}$,  $\Tilde{Z}=Z\setminus \{z\}$ and $\Tilde{H}=H\setminus \{h\}$. Then, $(\Tilde{Y}, \Tilde{Z}, (W_{\Tilde{z}})_{\Tilde{z}\in \Tilde{Z}},  \Tilde{H})$  clearly satisfies conditions (i) and (iii) from the eLF-HTC and that $\Tilde{Y}\cap (\Tilde{Z}\cup \{v\})=\emptyset$ (the first part of condition (ii)).

In order to show that the remainder of condition (ii) from the eLF-HTC also holds, we must show that $h\notin \textnormal{pa}_{\mathcal{L}}(\Tilde{Y})\cap (\textnormal{pa}_{\mathcal{L}}(\Tilde{Z})\cup \{v\})$. As conditions (ii) in the eLF-HTC and the LF-HTC are exactly the same, from here on, the proof is exactly the same as the proof of Proposition 5.2; for completeness, we restate it:
If $|\textnormal{ch}(h)|=2$, then the vertices $y$ and $z$ are the only two children of $h$. Therefore, $h$ can clearly not be a parent of $\Tilde{Y}$ or $\Tilde{Z}\cup v$, that is, $h\notin \textnormal{pa}_\mathcal{L}(\Tilde{Y})\cap \textnormal{pa}_\mathcal{L}(\Tilde{Z}\cup \{v\})$. In case $|\textnormal{ch}(h)|=3$, then there might exist one child $w\in \textnormal{ch}(h)\setminus \{y,z\}$. Because as previously argued $\Tilde{Y}\cap (\Tilde{Z}\cup \{v\})=\emptyset$, it follows that $w$ cannot be an element of both $\Tilde{Y}$ and $\Tilde{Z}\cup \{v\}$ simultaneously. Therefore, $h\notin \textnormal{pa}_\mathcal{L}(\Tilde{Y})\cap \textnormal{pa}_\mathcal{L}(\Tilde{Z}\cup \{v\})$ in the case $|\textnormal{ch}(h)|=3$ as well. Therefore, condition (ii) and hence the eLF-HTC are satisfied for $(\Tilde{Y}, \Tilde{Z}, (W_{\Tilde{z}})_{\Tilde{z}\in \Tilde{Z}},  \Tilde{H})$.
\end{proof}
\subsection{Proof of Theorem \ref{theorem_complexity}}
\label{sec_numerical_proof_3}
\begin{proof}
This proof bears some similarity to the proof of Theorem 5.3 from \citet{barber2022half} which itself is similar to the proof from \citet{foygel2012half}.

We start by analyzing the complexity of the determinantal subprocedure. The multi-source multi-sink determinantal flow graph and corresponding max-flow problem  can be (and is in practice) transformed  to an equivalent single-source single-sink max-flow problem by adding a source node $s$ pointing to all previous sources, and a sink node $t$ such that all previous sinks point to $t$ and by giving $s$ and $t$ capacity $\infty$ and the newly introduced edges capacity $1$. For this equivalent flow-graph, each max-flow computation has a computational complexity of at most
$\mathcal{O}((|V^f|+r^f)^3)$ (e.g., \citet{foygel2012half} or Section 26 in \citealp{cormen2022introduction}) where $r^f$ is the number of reciprocal edge pairs in the original (or equivalently transformed) determinantal flow-graph. 
Because $|V^f|=2(|V|+|\mathcal{L}|)$ and because in the transformed determinantal flow graph the number of reciprocal edge pairs is two times the number of reciprocal edge pairs in the original latent-factor graph, so $r^f=2r$, it follows that the determinantal-flow-graph calculation has a time complexity of at most $O((|V|+|\mathcal{L}|+r)^3)$. 

 In line 2 of the determinantal subprocedure, we iterate (at most) over all sets $S\subseteq V$ and sets $T\subseteq V\setminus \{v,w_0\}$ where $|S|=|T|+1=k$ for all possible sizes $k=1,\ldots,|V|-1$. Given that the number of subsets of $V$ equals $2^{|V|}$, it follows that there are at most $\mathcal{O}(2^{|V|})$ individual checks of Theorem \ref{theorem_determinantal} for one $w_0\rightarrow v$; if one is willing to only iterate over $c_{S}$-many subsets, then there are at most $\mathcal{O}(c_S)$ individual checks of Theorem \ref{theorem_determinantal} for one $w_0\rightarrow v$. As there are at most $|V|$-many observed parents of $v$, there are at most $\mathcal{O}(|V|2^{|V|})$ respectively $\mathcal{O}(|V|c_S)$ individual checks of Theorem \ref{theorem_determinantal} for a fixed target vertex $v\in V$, so for one call of the determinantal subprocedure.

Finding descendants in a graph (condition (i) in Theorem \ref{theorem_determinantal}) is of complexity at most  $\mathcal{O}(|V|^2)$ (by, for example, doing a breadth-first search (e.g., Section 22.2 in \citealp{cormen2022introduction}]) and thus, finding descendants is not of higher complexity than the max-flow computations. Therefore, the determinantal subprocedure has without any simplifying assumptions a time complexity of at most $\mathcal{O}(|V|2^{|V|}(|V|+|\mathcal{L}|+r)^3)$ and if one is just iterating over $c_S$-many subsets, a time complexity of at most $\mathcal{O}(|V|c_S(|V|+|\mathcal{L}|+r)^3)$.

Now to the eLF-HTC subprocedure: For the for-loop starting in line 1 we iterate over all sets $H\subseteq \mathcal{L}_{\geq 4}\subseteq \mathcal{L}$. In general, there are $\mathcal{O}(2^{|\mathcal{L}|})$ such iterations. However, as already argued in \citet{barber2022half}, under the assumption $|H|\leq c_H$, this number reduces to
\begin{align}
\label{eq_c_H_assumption}
    \sum_{i=0}^{c_H} {|\mathcal{L}_{\geq 4}|\choose i} =\mathcal{O}(|\mathcal{L}|^{c_H}).
\end{align}
In line 2 of the eLF-HTC subprocedure we then iterate over all $Z\subseteq \textnormal{ch}(H)\setminus \{v\}\subseteq V$ such that $|Z|=|H|$. It follows from an analogous calculation as in equation \eqref{eq_c_H_assumption} that this for-loop in line 2 has at most $\mathcal{O}(2^{|V|})$ iterations under no simplifying assumptions and $\mathcal{O}(|V|^{c_H})$ iterations under the simplifying assumptions from the theorem. In general, the for-loop over the $(W_z)_{z\in Z}$ has at most $\mathcal{O}(2^{|V|^2})$ iterations as each $W_z$ has at most $O(2^{|V|})$ supersets and as $|Z|\leq |V|$.
If one just takes as the $W_z$ the set of already rationally identified observed parents of $z$, then this for-loop in line 3 ``disappears" and hence has $\mathcal{O}(1)$ iterations.

Therefore, there are at most $\mathcal{O}(2^{|\mathcal{L}|}2^{|V|^2+|V|})$ max-flow calculations in each call of the eLF-HTC subprocedure when no simplifying assumptions are made. With the simplifying assumptions from the theorem, 
there are at most $\mathcal{O}(|\mathcal{L}|^{c_H}|V|^{c_H})$ max-flow calculations in each call of the eLF-HTC subprocedure. From an analogous argument as for the determinantal subprocedure, each max-flow computation has time complexity at most $\mathcal{O}((|V|+|\mathcal{L}|+r)^3)$. Each $\textnormal{htr}_H$-calculation can be done by, e.g., a breadth-first search yielding a  complexity of at most $\mathcal{O}(|V|+|\mathcal{L}|+|D|)$ (e.g., Section 22.2 in \citealp{cormen2022introduction}) and thus of at most $\mathcal{O}(|\mathcal{L}||V|^2)$ (and hence not of higher complexity than the max-flow computations) because $|D|=|D_V|+|D_{\mathcal{L}V}|\leq |V|^2+|\mathcal{L}||V|\leq 2|\mathcal{L}||V|^2$ and because $|\mathcal{L}|\leq |\mathcal{L}||V|^2$. The other required basic set operations  and finding parents is not of higher time complexity. Hence, each eLF-HTC subprocedure call has a time complexity of at most $\mathcal{O}(2^{|\mathcal{L}|}2^{|V|^2+|V|}(|V|+|\mathcal{L}|+r)^3)$ when no simplifying assumptions are made and $\mathcal{O}(|\mathcal{L}|^{c_H}|V|^{c_H}(|V|+|\mathcal{L}|+r)^3)$ with the simplifying assumptions.

Together, one call of the eLF-HTC subprocedure and one call of the determinantal subprocedure thus have a time complexity of $\mathcal{O}(2^{|\mathcal{L}|}2^{|V|^2+|V|}(|V|+|\mathcal{L}|+r)^3)$ with no simplifying assumptions and $\mathcal{O}(c_S|\mathcal{L}|^{c_H}|V|^{c_H}(|V|+|\mathcal{L}|+r)^3)$ with the simplifying assumptions.

Now to the combined algorithm (Algorithm \ref{algo_combined_identification}): Note that we run
the repeat-loop in each call of Algorithm \ref{algo_combined_identification} at most $|V|^2$ times.
To see this fact, note that this repeat loop only does another repetition if a further edge is added to $S_{\textnormal{edges}}$, otherwise the algorithm terminates. Therefore, after at most $|V|^2$-many repetitions of the repeat-loop in line 1 of Algorithm \ref{algo_combined_identification}, either all edges have been added to $S_{\textnormal{edges}}$ or Algorithm \ref{algo_combined_identification} has terminated.

Next, the for-loop in line 3 of Algorithm \ref{algo_combined_identification} is executed at most $|V|$-times within each repeat-loop iteration as there are at most $|V|$ unsolved nodes.
At the last level in the recursion tree, the for-loop in line 17 is not executed at all. Without this execution of the for-loop in line 17, the combined identification algorithm thus has a complexity of at most $\mathcal{O}(|V|^32^{|\mathcal{L}|}2^{|V|^2+|V|}(|V|+|\mathcal{L}|+r)^3)$ when no simplifying assumptions are made and $\mathcal{O}(c_S|\mathcal{L}|^{c_H}|V|^{c_H+3}(|V|+|\mathcal{L}|+r)^3)$ with the simplifying assumptions. Now, for any recursion level above, note that the complexity(-bound) of the for-loop in line 3 is negligible relative to the complexity(-bound) of calling another instance of the combined algorithm. Writing $A_0$ for the initial call of the combined identification algorithm and $A_i$ to represent any call with $i$ edges deleted from $G^\mathcal{L}$, we have that
\begin{align*}
    \textnormal{Time\_Complexity\_Bound}(A_i)\leq |V|^4\textnormal{Time\_Complexity\_Bound}(A_{i+1})
\end{align*}
as the repeat-loop in each call of Algorithm \ref{algo_combined_identification} is executed at most $|V|^2$-times and as the for-loop in line 17 of Algorithm \ref{algo_combined_identification} is called at most $|D_V|\leq |V|^2$-times in each repeat-loop-iteration. As at most $|D_V|\leq |V|^2$-edges can be deleted from $G^\mathcal{L}$, we obtain that Algorithm $1$ has a time complexity of
$\mathcal{O}(|V|^{4|V|^2+3}2^{|\mathcal{L}|}2^{|V|^2+|V|}(|V|+|\mathcal{L}|+r)^3)$ with no simplifying assumptions and $\mathcal{O}(c_S\allowbreak|\mathcal{L}|^{c_H}\allowbreak|V|^{4c_{\textnormal{rec}}+c_H+3}\allowbreak(|V|+|\mathcal{L}|+r)^3)$ with the simplifying assumptions.

\end{proof}
\section{Appendix to Section \ref{sec_numerical_experiments}}
\label{sec_appendix_numerical_experiments}
In this section, we present extensions of Tables \ref{table1} and \ref{table2} from Section \ref{sec_numerical_experiments}. In particular, we consider different combinations of the individual identification results in order to investigate each's contribution. The "No $W_z$-loop"-column is the algorithm including the determinantal and eLF-HTC subprocedure and the recursive subgraph result with the additional simplification for the for-loop over the $W_z$'s as explained in Theorem \ref{theorem_complexity} in Section \ref{section_computation}. Similarly, the $\leq 10$-, $\leq 100$- and $\leq 500$-columns refer to the algorithm including the determinantal and eLF-HTC subprocedure and the recursive subgraph result with the additional simplification of limiting the number of pairs $S, T$ in the determinantal subprocedure as explained in Section \ref{section_computation}. In Figure \ref{fig_runtime}, we also present the median runtimes for each latent structure and for each $|D_V|$. 
In Section \ref{sec_numerical_experiments}, we discuss some findings regarding these tables and figures.
\begin{figure}
    \centering
    \includegraphics[scale=0.15, angle=0]{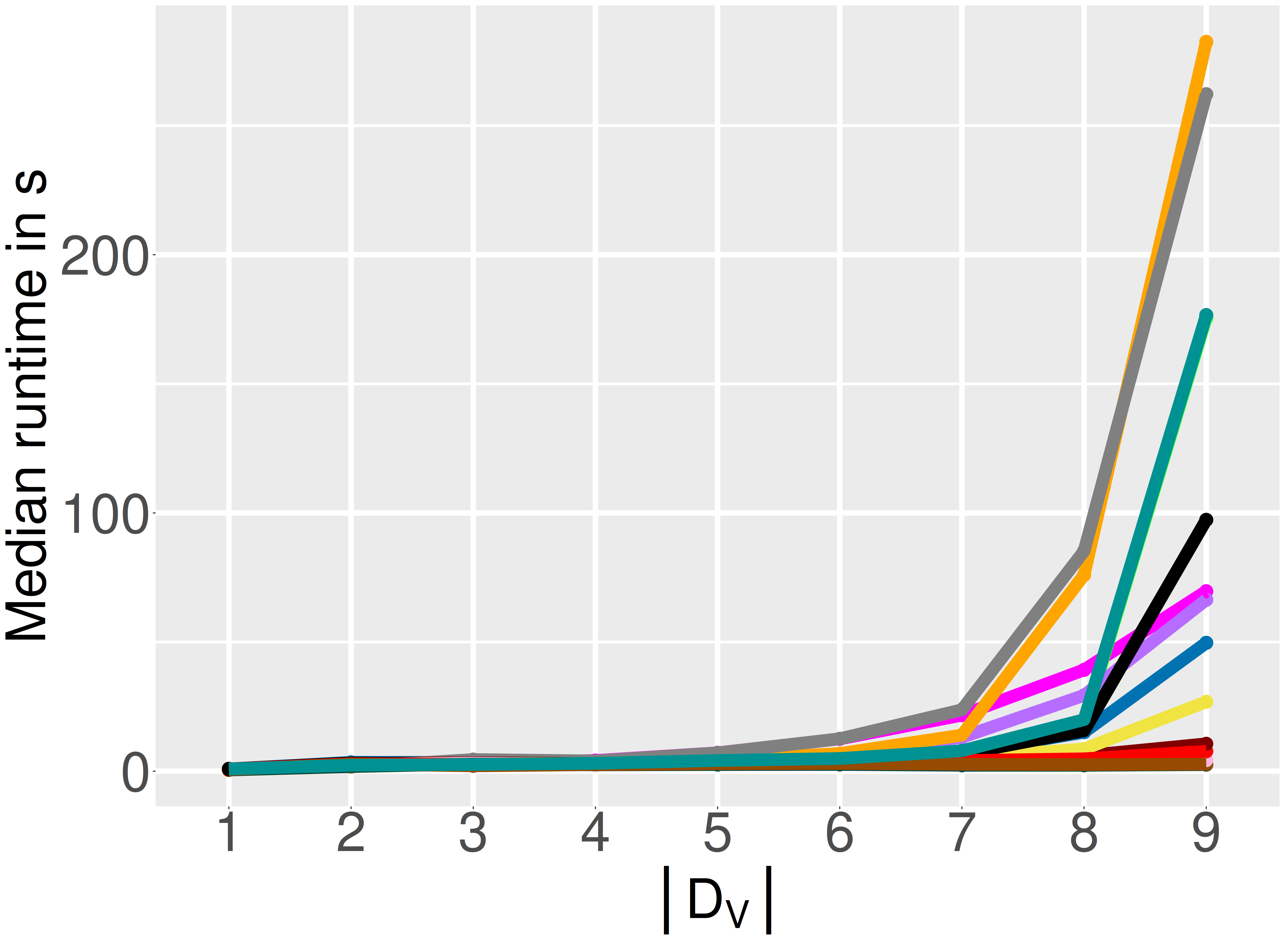}
        \includegraphics[scale=0.15, angle=0]{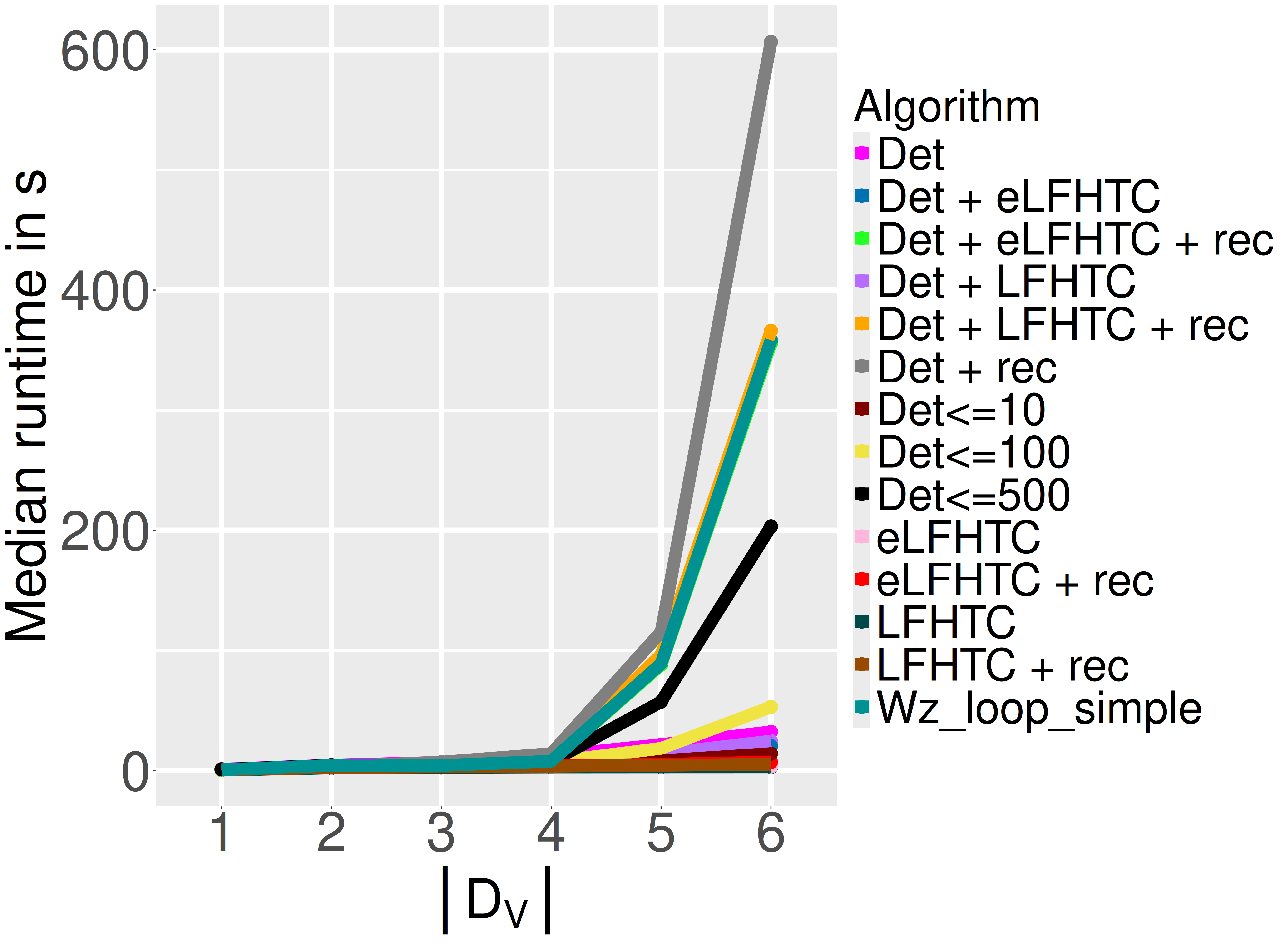}
    \caption{Median runtimes: 1 latent node (left), 2 latent nodes (right). For 
    better distinguishability: For the left plot and  $|D_V|=9$, the lines from top to bottom are: Det+LFHTC+rec, Det+rec, Wz\_loop\_simple, Det+eLFHTC+rec (essentially located under the Wz\_loop\_simple-line), Det$\leq$500, Det, Det+LFHTC, Det+eLFHTC, Det$\leq $100, Det$\leq$10, eLFHTC+rec, eLFHTC, LFHTC+rec, LFHTC. For the right plot and $|D_V|=6$, the lines from top to bottom are: Det+rec, Det+LFHTC+rec, Wz\_loop\_simple, Det+eLFHTC+rec (essentially located under the Wz\_loop\_simple-line), Det$\leq$500, Det$\leq $100, Det+LFHTC, Det$\leq$10, eLFHTC+rec, LFHTC+rec, eLFHTC, LFHTC.}
    \label{fig_runtime}
\end{figure}

\begin{table}[h]
\begin{center}
\resizebox{10cm}{!}{%
	\begin{tabular}{ cccccccc }
		\hline
		 $|D_V|$ & Total & Rationally & Det & Det + rec& LF-HTC & LF-HTC + rec & Det + LF-HTC\\
		\hline
		0 & 1 & 1 & 1 & 1 & 1 & 1 & 1  \\
		1 & 1 & 1 & 1 & 1 & 1 & 1 & 1 \\
		2 & 4 & 4 & 4 & 4 & 4 & 4 & 4 \\
	    3 & 13 & 13 & 13 & 13 & 13 & 13 & 13   \\
	    4 & 51 & 51 & 45 & 51 & 50 & 50 & 51  \\
		5 & 163 & 159 & 104 & 159 & 134 & 134 & 159 \\
		6 & 407 & 398 & 168 & 397 & 250 & 250 & 398 \\
		7 & 796 & 747 & 168 & 737 & 234 & 234 & 739 \\
		8 & 1169 & 956 & 89 & 911 & 64 & 64 & 882 \\
		9 &  1291 & 631 & 11 & 550 & 4 & 4 & 434\\
	\end{tabular}}\\
\resizebox{14cm}{!}{%
	\begin{tabular}{ cccccccc }
		\hline
		 $|D_V|$ & Total & Rationally & Det + LF-HTC + rec & eLF-HTC& eLF-HTC + rec & Det + eLF-HTC & Det + eLF-HTC + rec \\
		\hline
		0 & 1 & 1  & 1 & 1 & 1 & 1 & 1  \\
		1 & 1 & 1  & 1 & 1 & 1 & 1 & 1 \\
		2 & 4 & 4 & 4 & 4 & 4 & 4 & 4 \\
	3 & 13 & 13 & 13 & 13 & 13 & 13 & 13  \\
	4 & 51 & 51 & 51 & 51 & 51 & 51 & 51 \\
		5 & 163 & 159 & 159 & 154 & 154 & 159 & 159 \\
		6 & 407 & 398 & 398 & 372 & 372 & 398 & 398 \\
		7 & 796 & 747 & 742 & 642 & 649 & 743 & 743 \\
		8 & 1169 & 956 & 926 & 669 & 714 & 922 & 938 \\
		9 &  1291 & 631 & 566 & 300 & 377 & 556 & 606 \\
	\end{tabular}}\\	
	\resizebox{7cm}{!}{%
		\begin{tabular}{ ccccccc }
		\hline
		 $|D_V|$ & Total & Rationally & No $W_z$-loop & $\leq 10$ &  $\leq 100$ &  $\leq 500$\\
		\hline
		0 & 1 & 1 & 1 & 1 & 1 & 1\\
		1 & 1 & 1 & 1 & 1 & 1 & 1\\
		2 & 4 & 4 & 4 & 4 & 4 & 4\\
	3 & 13 & 13 & 13 & 13 & 13 & 13\\
	4 & 51 & 51 & 51 & 51 & 51 & 51\\
		5 & 163 & 159 & 159 & 154 & 157 & 159\\
		6 & 407 & 398 & 398 & 378 & 393 & 395\\
		7 & 796 & 747 & 743 & 670 & 711 & 738\\
		8 & 1169 & 956 & 938 & 771 & 864 & 929\\
		9 &  1291 & 631 & 606 & 425 & 494 & 589\\
	\end{tabular}
	}
\end{center}
\caption{Extended version of Table \ref{table1}.}
\label{table1_app}
\end{table}
\begin{table}
	\begin{center}
	\resizebox{10cm}{!}{%
		\begin{tabular}{ cccccccc }
			\hline
$|D_V|$ & Total & Rationally & Det & Det + rec& LF-HTC & LF-HTC + rec & Det + LF-HTC \\
		\hline
$0$ & 1 & 1 & 1 & 1& 1 & 1 & 1 \\
$1$ & 8 & 6 & 6 & 6& 6 & 6 & 6 \\
$2$ & 63 & 45 & 43 & 45 & 43  & 43 & 45 \\
$3$ & 391 & 255 & 238 & 255& 236 & 236 & 255 \\
$4$ & 1983 & 1171 & 882 & 1164& 1018 & 1018 & 1163 \\
$5$ & 7570 & 3898 & 1789 & 3742& 3028 & 3028 & 3691 \\
$6$ & 21,029 & 8960 & 1882 & 7704 & 5861 & 5861 & 7783 \\
		\end{tabular}}
	\resizebox{14cm}{!}{%
		\begin{tabular}{ cccccccc }
			\hline
$|D_V|$ & Total & Rationally & Det + LF-HTC + rec & eLF-HTC& eLF-HTC + rec & Det + eLF-HTC & Det + eLF-HTC + rec \\
		\hline
$0$ & 1 & 1 & 1 & 1& 1 & 1 & 1 \\
$1$ & 8 & 6  & 6 & 6& 6 & 6 & 6 \\
$2$ & 63 & 45 & 45 & 45& 45 & 45 & 45 \\
$3$ & 391 & 255 & 255 & 254 & 254 & 255 & 255 \\
$4$ & 1983 & 1171  & 1166 & 1140 & 1146 & 1166 & 1168 \\
$5$ & 7570 & 3898 & 3827 & 3601& 3642 & 3819 & 3850 \\
$6$ & 21,029 & 8960 & 8557 & 7563& 7750 & 8449 & 8675 \\
		\end{tabular}}
		\resizebox{7cm}{!}{%
		\begin{tabular}{ ccccccc }
		\hline
		 $|D_V|$ & Total & Rationally & No $W_z$-loop &  $\leq 10$ &  $\leq 100$ &  $\leq 500$ \\
		\hline
		0 & 1 & 1 & 1 & 1 & 1 & 1\\
		1 & 8 & 6  & 6 & 6 & 6 & 6\\
		2 & 63 & 45 & 45 & 45 & 45 & 45\\
	3 & 391 & 255 & 255 & 254 & 255 & 255\\
	4 & 1983 & 1171 & 1166 & 1149 & 1160 & 1167\\
		5 & 7570 & 3898 & 3832 & 3670 & 3766 & 3839\\
		6 & 21029 & 8960 & 8589 & 7882 & 8278 & 8617\\
	\end{tabular}
	}
	\end{center}
	\caption{Extended version of Table \ref{table2}.}
	\label{table2_app}
\end{table}
\bibliography{bib}

\end{document}